\documentclass[11pt]{amsart}
\usepackage[top=1.0in, bottom=1.0in, left=1.0in, right=1.0in]{geometry}

\usepackage{amsfonts, amssymb, amsmath, enumerate, mathtools,comment}
\usepackage{amsthm}
\usepackage{hyperref}
\usepackage[foot]{amsaddr}
\usepackage{bbm}
\usepackage{xcolor}
\usepackage[style=nature]{biblatex}
\usepackage{wasysym}
\hypersetup{
	colorlinks,
	linkcolor={red!40!gray},
	citecolor={blue!40!gray},
	urlcolor={blue!70!gray}
}

\usepackage{amsmath,amssymb}
\usepackage{xcolor,tikz}
\usetikzlibrary{matrix,positioning,calc,arrows.meta,fit,decorations.pathreplacing}
\usetikzlibrary{patterns}

\usepackage{graphicx}
\usepackage{dsfont}

\numberwithin{equation}{section}
\usepackage{amsfonts, amssymb, amsmath, enumerate}

\usepackage{mathrsfs}

\usepackage{mathtools}
\usepackage[normalem]{ulem}
\usepackage{comment}

\usepackage{bbm}
\usepackage{algorithm}
\usepackage{algorithmic}
\usepackage{tikz}

\usetikzlibrary{shapes.geometric, arrows.meta, positioning,calc,matrix,calc,positioning,arrows.meta}

\usepackage{amsmath,amssymb}
\usepackage{xcolor}

\usepackage{adjustbox}

\tikzstyle{roundedbox} = [
    rectangle, rounded corners,
    minimum width=3.5cm, minimum height=1cm,
    text centered, draw=black, fill=blue!20,
    text width=10cm, align=center
]
\tikzstyle{arrow} = [thick, ->, >=Stealth]

\usepackage{pgfplots}
\pgfplotsset{compat=newest}

\numberwithin{equation}{section}

\DeclareUnicodeCharacter{2113}{l}

\definecolor{darkgreen}{cmyk}{0.6,0,0.8,0}

\DeclareMathOperator{\ind}{\mathbbm{1}}

\DeclareMathOperator{\E}{\mathbb{E}}
\DeclareMathOperator{\R}{\mathbb{R}}
\DeclareMathOperator{\F}{\mathbb{F}}

\DeclareMathOperator{\Z}{\mathbb{Z}}

\DeclareMathOperator{\bS}{\mathbb{S}}
\DeclareMathOperator{\Pb}{\mathbb{P}}

\DeclareMathOperator*{\diag}{diag}

\DeclareMathOperator{\supp}{supp}
\DeclareMathOperator{\sign}{sign}

\DeclareMathOperator*{\tr}{tr}
\DeclareMathOperator*{\Tr}{Tr}

\DeclareMathOperator{\card}{card}

\DeclareMathOperator{\sparse}{Sparse}

\DeclareMathOperator{\poly}{poly}

\DeclareMathOperator{\cov}{Cov}

\DeclareMathOperator{\spec}{spec}
\DeclareMathOperator{\row}{Row}
\DeclareMathOperator{\col}{Col}

\DeclareMathOperator{\vol}{vol}

\def \etc {,\ldots,}
\def \P {\mathbb{P}}

\newcommand{\cN}{\mathcal{N}}

\newcommand{\cE}{\mathcal{E}}

\newcommand{\1}{\mathbbm 1}

\DeclareMathOperator{\incomp}{Incomp}

\DeclareMathOperator{\comp}{Comp}

\DeclareMathOperator{\e}{\varepsilon}

\DeclarePairedDelimiter{\norm}{\lVert}{\rVert}

\DeclarePairedDelimiter{\abs}{\lvert}{\rvert}

\DeclarePairedDelimiter{\paren}{(}{)}

\DeclarePairedDelimiter{\angbr}{\langle}{\rangle}
\DeclarePairedDelimiter{\cbr}{\{}{\}}

\newtheorem{theorem}{Theorem}[section]
\newtheorem{proposition}[theorem]{Proposition}
\newtheorem{corollary}[theorem]{Corollary}
\newtheorem{lemma}[theorem]{Lemma}

\newtheorem{remark}[theorem]{Remark}

\newtheorem{definition}[theorem]{Definition}

\begin{document}

\title[Well-Conditioned Oblivious Perturbations in Linear Space]{
Well-Conditioned Oblivious Perturbations in Linear Space 
}
\date{}
\author{Shabarish Chenakkod$^*$}
\email{shabari@umich.edu}
\author{Micha{\l} Derezi\'nski$^*$}
\email{derezin@umich.edu}
\author{Xiaoyu Dong$^\dagger$}
\email{xdong@nus.edu.sg}
\author{Mark Rudelson$^*$}
\email{rudelson@umich.edu}
\thanks{Partially supported by DMS 2054408, CCF 2338655, and a Google ML and Systems Junior Faculty Award. This work was done in part while MD was visiting the Simons Institute for the Theory of Computing.}
\address{$^*$University of Michigan, Ann Arbor, MI, USA}
\address{$^\dagger$National University of Singapore, Singapore}

\begin{abstract} Perturbing a deterministic $n$-dimensional matrix with small Gaussian noise is a cornerstone of smoothed analysis of algorithms [Spielman and Teng, JACM 2004], as it reduces the condition number of the input to $O(n)$, and with it the complexity of many matrix algorithms. However, when deployed algorithmically, these perturbations are expensive due to the cost of generating and storing $n^2$ Gaussian random variables. We propose a perturbation that requires generating and storing $O(n)$ random numbers in $O(\log n)$ bits of precision, and reduces the condition number of any deterministic matrix to $O(n)$, matching Gaussian perturbations. Our result takes a step towards derandomizing matrix perturbations in numerical linear algebra, and in particular, implies that we can solve an $n\times n$ linear system to within $\varepsilon$  backward error in $O(n^2\log(n)/\varepsilon)$ time using only $O(n\log n)$ random bits.

In our construction, we introduce the concept of a pattern matrix, which is a dense deterministic matrix that maps all sparse vectors into dense vectors in a way that goes beyond the classical uncertainty principle of the Fourier transform. We combine this with a sparse perturbation whose entries are dependent and located in a non-uniform fashion. In order to analyze this construction, we develop new techniques for lower bounding the smallest singular value of a random matrix with dependent entries.
\end{abstract}

\maketitle

\thispagestyle{empty}
\newpage
\setcounter{page}{1}

\section{Introduction.}
\label{s:intro}

Going beyond the worst-case guarantees is a central challenge in closing the theory-practice gap for modern algorithm analysis. One of the most effective tools towards achieving this goal has been \emph{smoothed analysis}. Proposed by Spielman and Teng \cite{spielman2004smoothed}, this framework considers the performance of an algorithm on an arbitrary input perturbed randomly by small Gaussian noise. Such small perturbations are known to escape the hard corner cases that arise in the analysis of numerous algorithms, which has led to new guarantees for classical methods such as the simplex algorithm in linear programming \cite{spielman2004smoothed} or Gaussian elimination for solving systems of linear equations \cite{sankar2006smoothed}, as well as a new understanding of other popular algorithms, including in machine learning \cite{kalai2009learning,rakhlin2011online}, game theory \cite{daskalakis2024smooth}, and computational linear algebra \cite{burgisser2013condition,cifuentes2022polynomial}.

Most techniques in smoothed analysis require ensuring that a random perturbation of the input escapes certain degenerate corner cases, which are typically instances where an appropriate matrix is nearly singular. This is attained by leveraging tools from random matrix theory to show that a perturbed matrix is sufficiently well-conditioned to ensure good performance of the algorithm. The simplest and most classical formulation of this perturbation property, which directly arises as a non-degeneracy condition for many algorithms in computational linear algebra, is the following.
\begin{definition}\label{d:oblivious-perturbation}
An $n\times n$ random matrix $R$ is a
  $(\kappa, \e,\delta)$-(well-conditioned oblivious) perturbation if for
  any deterministic matrix $A\in\R^{n\times n}$ with unit spectral norm (denoted $\|A\|=1$), we have:
  \begin{align*}
    \Pb\Big(\kappa(A+\varepsilon R)\leq
    \kappa\quad\text{and}\quad\|R\|\leq 1\Big)\geq 1-\delta,\qquad\text{where}\qquad\kappa(M)=\|M\|\|M^{-1}\|.
  \end{align*}
\end{definition}
In other words, a $(\kappa,\e,\delta)$-perturbation is one that, for any matrix, with probability $1-\delta$ introduces at most $\e$ relative error and reduces its condition number to at most $\kappa$. We refer to this as an \emph{oblivious} perturbation because the same distribution over $R$ has to apply to any input matrix $A$. Sankar, Spielman, and Teng \cite{sankar2006smoothed} showed that an appropriately scaled Gaussian matrix is a $(\kappa,\varepsilon,\delta)$-perturbation with condition number that scales linearly with $n$, namely $\kappa=O(n/\varepsilon\delta)$, and they used this to provide smoothed analysis of Gaussian elimination in finite-precision arithmetic. The condition number $\kappa$ also controls the complexity of many iterative numerical algorithms such as Krylov subspace methods \cite{greenbaumbook} and interior-point methods \cite{renegar1995incorporating}, and is a necessary component in the stability analysis of others, such as algorithms for matrix diagonalization \cite{banks2023pseudospectral}.

Smoothed analysis can be viewed from a purely algorithmic perspective, where one actually performs the perturbation before solving the problem instance. This viewpoint has been particularly prevalent in numerical linear algebra \cite{sankar2006smoothed,davies2008approximate,meyer2024unreasonable,peng2024solving}, where a predominant way to evaluate an algorithm is through \emph{backward error}, i.e., the distance to the nearest problem instance solved by its output \cite{Higham:2002:ASNA,derezinski2026towards}. In particular, an algorithm that applies a $(\kappa,\varepsilon,\delta)$-perturbation and then solves the resulting task incurs at most $\varepsilon$ backward error. The notion of backward error dates back to von Neumann and Goldstine \cite{von1947numerical} in the early days of computing, and it was developed as a rigorous framework by Wilkinson in the late 1950s \cite{wilkinson1960error,wilkinson1961error} in order to analyze the effect of perturbations occurring as a result of rounding errors. The success of smoothed analysis has shown that injecting deliberate perturbations into the data that go beyond the effect of rounding errors can in fact be beneficial in attaining strong backward error performance in modern algorithms. 

While Gaussian perturbations are often the most convenient for the analysis, they are typically not the most practical (or sometimes not even tractable) due to their high cost in both randomness and space requirements. For instance, an $n\times n$ Gaussian matrix requires generating $n^2$ Gaussian random variables, each taking $O(\log n)$ random bits when properly discretized, which is expensive to generate, and it may also exceed the memory requirements of the input data (e.g., when it is sparse or otherwise structured). This motivates the broader goal of derandomizing matrix perturbations, which was raised recently by B\"urgisser and Srivastava \cite{amsel2026linear} during a workshop at the Simons Institute in Berkeley. As a step towards that goal, we ask:
\begin{center}
    \textit{
    How many random bits do we need to generate an $n\times n$ well-conditioned oblivious perturbation?
    }
\end{center}
A natural strategy for reducing the random bit complexity of a random matrix is to sparsify it by replacing most of its entries with zeros. This places our question in the context of a line of works in random matrix theory studying the spectral properties of sparse matrices. In this context, Tao and Vu \cite{tao2008random} showed that a $\poly(n)$-conditioned perturbation can be obtained from a matrix with $n^{1+\theta}$ randomly placed non-zero entries for any fixed $\theta > 0$. Later, Shah, Srivastava, and Zeng \cite{shah2024sparse} optimized their analysis, reducing the sparsity to $O(n\log^6n)$ non-zero entries, i.e., $O(n\log^7n )$ random bits since $\log n$ bits are used for each entry, at the cost of a super-polynomial dependence in the condition number, $\kappa = n^{O(\frac{\log n}{\log\log n})}$. They also observed that $n\log n$ non-zeros and $n\log^2 n$ random bits is the natural limit for this approach since below that limit the resulting perturbation becomes singular due to coupon collector problem, and thus cannot be well-conditioned.

\subsection{Main result.}
In this work, we improve upon the prior results in two ways. First, we construct a well-conditioned oblivious perturbation using only $O(n\log n)$ random bits, overcoming the coupon collector barrier from previous approaches. In particular, our perturbation can be stored in linear space, i.e., using $O(n)$ numbers of $\poly(n)$ size. Second, we show that the condition number of this perturbation is $O(n)$, i.e., it scales linearly with the dimension, matching dense Gaussian matrices.
\begin{theorem}\label{t:main}
    Given $n\geq 1$ and $\varepsilon,\delta\in(0,1)$, using $n\log(n)\cdot\poly(1/\delta)$ expected random bits we can generate an $n\times n$ $(\kappa,\varepsilon,\delta)$-well-conditioned oblivious perturbation with condition number $\kappa= n\cdot\poly(1/\varepsilon\delta)$.
\end{theorem}
In order to overcome the coupon collector barrier for uniformly sparse perturbations, we develop a new construction, which adds together two crucial components:
\begin{enumerate}
\item A dense deterministic matrix whose rows and columns are multiplied by $O(n)$ random signs; 
\item A sparse random matrix with $O(n)$ random sign entries placed in a non-uniform fashion.
\end{enumerate}
Our analysis of this construction is centered on lower bounding the smallest singular value of the perturbed matrix, $s_n(M)$ for $M=A+\varepsilon R$, which is the infimum of $\|Mx\|_2$ over all $x$ in the unit sphere. Our starting point is a standard approach of decomposing the sphere into compressible vectors (those that are sufficiently close to a sparse vector) and incompressible vectors (all others). 

In order to handle the compressible vectors, we introduce the concept of a \emph{pattern matrix}, which is a dense deterministic matrix that maps all sparse vectors into dense vectors in a certain strong sense. At a high level, this idea is widely known as the \emph{uncertainty principle} \cite{donohostarkuncertainity}, which is attained for instance by the Discrete Fourier Transform and the Walsh-Hadamard matrix. However, our concept of a pattern matrix defines a strictly stronger property which is required for our guarantees and is not attained by these classical transforms.
 While the existence of a pattern matrix is relatively straightforward
(e.g., a dense random sign matrix), obtaining one with high
probability using $O(n\log n)$ random bits requires a more elaborate
construction. Equipped with such a matrix, we then randomize it by multiplying its rows and columns by independent random signs, and show that this perturbation is sufficient to handle all compressible vectors  (see Section \ref{s:overview-R1} for more details).

In order to handle the incompressible vectors, we return to the sparse perturbations, and show that $O(n)$ non-zero entries are sufficient in this case. However, a na\"ive implementation that places the entries uniformly at random hits an obstacle: it leads to a few \emph{heavy rows}, i.e., those with disproportionately more non-zeros than the rest, which blows up the norm of the perturbation. To address this, we can simply zero out those heavy rows, but this introduces a non-trivial dependency between the entries, which is not handled by standard techniques. Here, our key contribution is developing a novel argument for analyzing sparse matrices with such dependencies (see Section~\ref{s:overview-R2} for more details).

\subsection{Application: Solving linear systems in backward error.}\label{s:intro-linear}
As a motivating application for our main result, we consider solving a large system of linear equations, $Ax=b$, for an $n\times n$ matrix $A$ and an $n$-dimensional vector $b$. Controlling the condition number $\kappa$ of $A$ is crucial both for the stability analysis of direct solvers such as Gaussian elimination \cite{sankar2006smoothed}, as well as for the convergence analysis of iterative solvers such as Conjugate Gradient \cite{saad2003iterative}. By adding a perturbation to $A$, we can gain strong control over the condition number while introducing only a small backward error into the solution. 

To illustrate this, we show how Theorem \ref{t:main} can be used to reduce the random bit complexity of solving a linear system in backward error. We build on recent work by Derezi\'nski, Nakatsukasa, and Rebrova \cite{derezinski2026towards}, who show that when seeking a backward error approximation to a linear system, it is possible to overcome traditional computational barriers, such as the current cost of matrix multiplication, through the use of randomization. The following result combines our perturbation from Theorem~\ref{t:main} with their deterministic $O(\log(\kappa)/\varepsilon)$-degree polynomial approximation to a linear system in backward error \cite[Theorem 5.1]{derezinski2026towards}. The proof follows the same steps as their Theorem~5.3.
\begin{theorem}\label{t:linear-intro}
Given $\varepsilon\in(0,1)$, $A\in\R^{n\times n}$ with $A \ne 0$, and $b\in\R^n$, using $O(n^2\log(n)/\varepsilon)$ time and $O(n\log n)$ random bits we can compute $x$ such that, with probability $0.9$,
\begin{align*}
\tilde Ax=b\quad\text{for some $\tilde A$ such that }\|\tilde A-A\|\leq \varepsilon \|A\|.     
\end{align*}
\end{theorem}
\begin{remark} The guarantee is meaningful when $\varepsilon\gg 1/n^{\omega-2}$, where $\omega \approx 2.371$ is the current matrix multiplication exponent. Otherwise, an $O(n^\omega)$ time deterministic solver becomes more efficient. 
\end{remark}
\begin{remark}
Our perturbation has the key advantage that it uses \emph{linear space}, i.e., only as much as it takes to store the output vector $x$, assuming we use $\Theta(\log n)$ bits of precision to store it. This is particularly important when the input matrix is given as read-only or through matrix-vector products, so that it is not stored in the algorithm's working memory. We note that while Theorem \ref{t:linear-intro} is stated in exact arithmetic (Real RAM), the finite-precision analysis of adding our perturbation to a matrix or applying it to a vector is straightforward. We leave the finite-precision analysis of the polynomial approximation algorithm from \cite{derezinski2026towards} to future~work.
\end{remark}
This result matches the best known time complexity for solving an $n\times n$  linear system in backward error \cite[Theorem 5.3]{derezinski2026towards}, while reducing the amount of randomness required from the previous  $O(n^{1+\theta})$ random bits (for any $\theta>0$), which is attained by the sparse perturbation result of Tao and Vu \cite{tao2008random}. Alternatively, one could use a discretized version of the sparse perturbation of Shah, Srivastava and Zeng \cite{shah2024sparse}, which requires $O(n\log^7 n)$ random bits but attains a sub-optimal $O(n^2\log^2(n)/\varepsilon)$ time complexity. Nevertheless, the coupon collector barrier implies that the sparse perturbation \emph{cannot} achieve linear space (i.e., $O(n\log n)$ bits) even if one could push the analysis further, highlighting the need for our construction.
We remark that it is crucial for the above application that the linear system may be non-symmetric (for symmetric systems a simple diagonal shift is sufficient), since backward error is not preserved under standard symmetrizations such as Hermitian dilation or the normal equations.

\subsection{Further discussion of related work.}
One of the settings where the analysis of the smallest singular value of zero mean sparse random matrices with i.i.d.~entries perturbed by deterministic shifts has been previously considered is in the context of proving the circular law for random matrices  \cite{tao2008random, rudelsontikhomirov, basak2019circular}.  Formally, these look at the smallest singular value of a (complex valued) matrix $A+R$, for a deterministic matrix  $A$ and a random matrix $R$ with independent entries distributed as $x \cdot \delta$ where $x$ is some complex valued random variable and $\delta \sim \text{Bernoulli}(\rho)$. To prove the circular law, it suffices to consider a diagonal shift $A = zI$ for $z \in \mathbb{C}$. However, as discussed earlier, Tao and Vu \cite{tao2008random} prove a more general result with arbitrary shift which applies to our setting and requires $n^{1+\theta}$ random bits. Notably, \cite{rudelsontikhomirov, basak2019circular}  attain better sparsity, but their approaches only work for a diagonal shift, and thus do not apply to our setting.

Shah, Srivastava and Zeng \cite{shah2024sparse} study sparse perturbations in the context of fast matrix diagonalization algorithms. In this setting, they require bounding the so-called \emph{eigenvector condition number} and the \emph{minimum eigenvalue gap} of a perturbed matrix, in order to attain the so-called pseudospectral shattering \cite{banks2023pseudospectral}. A smallest singular value bound for a perturbed matrix, such as the one in our main result, appears as a central step in their argument. The problem of derandomizing pseudospectral shattering perturbations was recently posed by B\"urgisser and Srivastava \cite{amsel2026linear}. Our results can be viewed as a step in that direction.

Another approach towards using less randomness and less space in smoothed analysis is taken by Shah and Silwal \cite{shah2021smoothed}, who consider \emph{low-rank perturbations}. However, since a low-rank perturbation is inherently singular, their results are not directly comparable with our setting.

Randomized perturbations arise as part of several algorithmic approaches in computational linear algebra. 

For solving sparse linear systems, Peng and Vempala \cite{peng2024solving} use a random perturbation to improve the stability of a Krylov method. A similar strategy is used by Chen, Epperly, Meyer, Musco, Musco, and Rao \cite{meyer2024unreasonable,chen2026does} for Krylov-based low-rank approximation. These works rely on a diagonal perturbation of a symmetric matrix to ensure a sufficient minimum eigenvalue gap via the result of Minami \cite{minami1996local}. Finally, as discussed in Section \ref{s:intro-linear}, Derezi\'nski, Nakatsukasa, and Rebrova  \cite{derezinski2026towards} use random perturbations to avoid condition number dependence when solving dense non-symmetric linear systems in backward~error.

\section{Overview of the Techniques.}\label{sec:outline}
In this section, we provide a high-level overview of the construction and analysis for the well-conditioned oblivious perturbation that gives rise to Theorem \ref{t:main}.

Recall that our goal is to control the condition number of the perturbed matrix $M = A + \varepsilon R$ for a random perturbation that we will define gradually throughout the section. By assumption, the starting matrix $A$ has $\norm{A} \le 1$ and we will ensure by construction that $\norm{R} \leq 1$ with high probability, which means that $\norm{M} \leq 2$. Thus to control the condition number, it suffices to look at the smallest singular value. The following is our main technical result.

\begin{theorem} \label{coro:pertur}
     Let $0<\e<1$ and $0<\delta<1$. Then there exist $n \times n$ random matrices $R_n$ and constants $c_{\ref{coro:pertur}.1}, c_{\ref{coro:pertur}.2}, c_{\ref{coro:pertur}.3}$ such that $\|R_n\|\leq 1$ with probability $1-\delta$ and for any $n \times n$ deterministic matrix $A$ with  $\norm{A}= 1$ and $n\ge c_{\ref{coro:pertur}.1}\delta^{-5}$, we have
    \[ \Pb \paren*{ s_n(A+ \varepsilon R_n)  \ge c_{\ref{coro:pertur}.2} \e^3 \delta^2n^{-1} } \ge 1-\delta \]
    and $R_n$ can be generated from $c_{\ref{coro:pertur}.3}\delta^{-2}n \log n$ random bits. 
\end{theorem}

The standard approach \cite{litvak2005smallest,rudelson2008littlewood} to lower bounding $s_n(M) = \inf_{x\in\bS^{n-1}}\|Mx\|_2$ involves separately considering the infimum over compressible vectors and incompressible vectors for some $\alpha$ and $\nu$,
\[ \inf_{x \in \comp(\alpha, \nu)} \norm{Mx}_2 \quad \text{ and } \quad \inf_{x \in \incomp(\alpha, \nu)} \norm{Mx}_2, \]
where
\begin{align*}
    \comp(\alpha, \nu) &= \{ x \in \bS^{n-1}\  \mid\ \norm{x - v}_2 \le \nu \text{ for some } v \in \R^{n} \text{ with } \supp(v) \le \alpha n \}, \\
    \incomp(\alpha, \nu) &= \bS^{n-1} \backslash \comp(\alpha, \nu).
\end{align*}
In each case, we use a different property of $R$ to obtain a high probability lower bound for the infimum, and to get these properties, we define $R:= R_1 + R_2$ as a sum of two independent matrices, a dense matrix $R_1$ and a sparse matrix $R_2$. For the infimum over compressible vectors $x$, we condition on the matrix $R_2$ and treat $R_1$ as the perturbation, obtaining a high probability lower bound for $\norm{(A_1+\varepsilon R_1)x}_2$ where $A_1=A+\varepsilon R_2$. For incompressible vectors, we condition on $R_1$ and treat $R_2$ as the perturbation.

\subsection{Construction of $R_1$ and invertibility control for compressible vectors.}
\label{s:overview-R1}

The standard analysis of $\inf_{x \in \comp(\alpha, \nu)} \norm{Mx}_2$ involves obtaining a lower bound on $\norm{Mx}_2$ for a single appropriately sparse vector and then using a net argument to extend the lower bound to the whole set $\comp(\alpha, \nu)$. When the coordinates of the perturbation are independent, $Mx$ has independent coordinates and tensorization of a small ball probability estimate for each coordinate provides the lower bound. However, using this approach for a \textit{sparse} perturbation may not give us a sufficiently strong tail probability estimate for $\norm{Mx}_2$ to be able to take the union bound. Indeed, this is what prevents us from relying only on the sparse component $R_2$ in our construction.

\subsubsection{Dense perturbation via pattern matrices.}

To address this challenge, we introduce a novel construction for $R_1$. We define
    \begin{equation*}
        R_1=\frac{1}{\rho \sqrt{n}} \cdot D_1 V D_2 
    \end{equation*}
for $D_1 = \diag(\eta_{1,1} \etc \eta_{1,n})$ and $D_2 = \diag(\eta_{2,1} \etc \eta_{2,n})$ where $\eta_{1,1} \etc \eta_{1,n}$ and $\eta_{2,1} \etc \eta_{2,n}$ are i.i.d. Rademacher random variables, and $V$ is any deterministic \textit{pattern} matrix. 
Here, a pattern matrix is any $n \times n$ matrix $V$ with $\pm 1$ entries which has $\norm{V} = O(\sqrt{n})$ and, for any sparse unit vector $x$ within a specified sparsity threshold, both $Vx$ and $V^Tx$ have $\Omega(n)$ many coordinates bigger than some absolute constant (See Definition \ref{def:pattern} for the full details). Thus, the pattern matrix transforms a sparse vector to one with many large coordinates. 

After conditioning on $R_2$ and $D_2$, if $i \in [n]$ is such that $(VD_2x)_i > \beta$ for some absolute constant $\beta>0$, then 
$$(Mx)_i = ((A+\varepsilon R_2)x)_i + (\varepsilon/\rho\sqrt{n}) \eta_{1,i} (VD_2x)_i$$
is greater than $(\varepsilon\beta/\rho\sqrt{n})$ with probability more than half. Tensorizing this lower bound (Lemma~\ref{lem:tensorization}) across all such $i$, of which there are $\Omega(n)$ many by the definition of a pattern matrix, we get a high probability lower bound for $\norm{Mx}_2$ with the tail probability decaying exponentially in $n$. This is sufficient to take the union bound over an appropriate net of sparse vectors and complete the proof of the lower bound over all compressible vectors ($D_2$ does not play a role in this part of the argument, but rather, for the lower bound over incompressible vectors).

\subsubsection{Constructing pattern matrices.}
There exist well known linear transforms that map sparse vectors to ones with many nonzero coordinates, such as the Discrete Fourier Transform or the Walsh-Hadamard matrix. They satisfy the so-called \emph{uncertainty principle} according to which the product of the support of a vector and the support of its transform is at least $n$ \cite{donohostarkuncertainity}. Unfortunately, this uncertainty principle is tight, in the sense that one can find a $\sqrt{n}$-sparse vector whose transform is also $\sqrt{n}$-sparse (for completeness, we prove this in Proposition \ref{prop:hadamardctrex}). This is an issue since we require $\Omega(n)$ large (and not just nonzero) coordinates for all $cn$-sparse vectors, for some small $c$. By considering randomized Hadamard matrices, one can get stronger uncertainty principle for a fixed sparse vector (e.g., \cite[Lemma 3.3]{Tropp2011SRHT}) . However, this stronger uncertainty principle is not helpful for our purpose, i.e., to uniformly lower bound the norm of the image for all vectors. In fact, to extend the norm lower bound from a fixed sparse vector $x$ to all sparse vectors via a union bound, we need the bound to hold with probability $1 - \exp(-cn)$ whereas the stronger uncertainty principle \cite[Lemma 3.3]{Tropp2011SRHT} only holds with probability $1-n^{-O(1)}$. We can see why randomized Hadamard matrices do not improve the situation more directly. Observe that if $z$ is a bad vector for $H$ as in Proposition \ref{prop:hadamardctrex}, then $Dz$ is also a bad vector for the randomized-Hadamard $HD$, so applying randomization does not exclude any of the bad vectors identified in Proposition \ref{prop:hadamardctrex}.

A stronger additive version of the uncertainty principle (such as the one by Tao \cite{tao2005uncertainty}), which has the sum (instead of product) of supports at least $n$ would ensure that $Vx$ has $\Omega(n)$ many nonzero coordinates for $cn$-sparse vectors $x$, but this is still insufficient for our needs since we also need the magnitude of these $\Omega(n)$ nonzero coordinates to be greater than some fixed dimension independent threshold. Moreover, the result in \cite{tao2005uncertainty} only holds for the DFT over finite fields of prime order and does not readily extend to arbitrary sizes.

Therefore, we consider a randomized construction for $V$. While one can show that a matrix with independent Rademacher entries is a pattern matrix with high probability, we introduce a construction which requires much fewer random bits, and therefore can be stored efficiently.

This construction is inspired by one due to Lovett and Sodin \cite{LovettSodin2008}, who in a different context construct a randomized embedding from $\ell_2$ into $\ell_1$ using the idea of considering an entry-wise (Hadamard) product of two random matrix models with $\pm 1$ entries, with each model having its own desirable properties, to get a random matrix model which enjoys both sets of properties. We consider a model where $V = V_1 \odot V_2 \odot V_3^T$ for the  Hadamard product $\odot$ of three random matrices $V_1, V_2$ and $V_3$ with Rademacher entries. Here, $V_1$ has $C\log(n)$-wise independent entries, $V_2$ has independent rows with 4-wise independent coordinates and $V_3$ is an i.i.d. copy of $V_2$, i.e., $V_3^T$ has independent columns with 4-wise independent coordinates (See Figure \ref{fig:pattern}). In Proposition \ref{prop:pattern}, we prove that a random matrix generated with this model is a pattern matrix with high probability. We use each of the three components in the definition of $V$ to prove the three properties we want of the pattern matrix by conditioning on the other two components:
\begin{enumerate}
    \item After conditioning on $V_2$ and $V_3$, we use the moment method to get a $O(\sqrt{n})$ bound for $\norm{V_1}$, and consequently, for $\norm{V}$. Without $V_1$, using just the independence of the rows or columns, we would yield a poorer norm bound that we would have to scale $R_1$ by, which would result in a suboptimal smallest singular value estimate.
    \item After conditioning on $V_1$ and $V_3$, for any fixed unit vector $x$, $Vx$ has independent coordinates. Each coordinate $(Vx)_i$ is a linear combination of the coordinates of $x$ with 4-wise independent Rademacher coefficients. This is enough to show that $\abs{(Vx)_i}$ is greater than some constant with constant probability. The independence of rows allows us to use concentration arguments to claim that $\abs{(Vx)_i}$ is greater than a constant for sufficiently many $i$, with a probability that decays exponentially in $n$, which is crucial to be able to take a union bound across all sparse vectors.
    \item We use $V_3$ similarly as $V_2$, in order to show that $V^Tx$ has $\Omega(n)$ large coordinates.
\end{enumerate}

\begin{figure}
    \centering
    \begin{tikzpicture}[
    matbox/.style={draw, thick, minimum width=3.2cm, minimum height=3.2cm, inner sep=0pt},
    label font/.style={font=\large},
]

\node[matbox, fill=blue!6] (V1) at (0,0) {};
\node[above=-60pt] at (V1.south) {$V_1$};

\def\Nrows{5}    
\def\Ncols{5}    
\def\xspacing{0.64}  
\def\yspacing{0.64}  

\pgfmathsetmacro{\cnt}{0}
\foreach \r in {1,...,\Nrows}{
  \foreach \c in {1,...,\Ncols}{
    \pgfmathsetmacro{\cnt}{\cnt+1}
    
    \pgfmathtruncatemacro{\h}{mod(\r*3 + \c*7, 5)}
    \pgfmathtruncatemacro{\sign}{(-1)^\h}
    \pgfmathsetmacro{\xpos}{(\c - (\Ncols+1)/2) * \xspacing}
    \pgfmathsetmacro{\ypos}{((\Nrows+1)/2 - \r) * \yspacing}
    \ifnum\sign>0
      \node[font=\scriptsize, text=blue!50!black] at ($(V1.center)+(\xpos,\ypos)$) {$+1$};
    \else
      \node[font=\scriptsize, text=blue!50!black] at ($(V1.center)+(\xpos,\ypos)$) {$-1$};
    \fi
  }
}

\draw[decorate, decoration={brace, amplitude=5pt, mirror}, thick, blue!70!black]
  ($(V1.south west)+(0,-3pt)$) -- ($(V1.south east)+(0,-3pt)$)
  node[midway, below=8pt, font=\small, text=blue!70!black, align=center] 
  {$2\log n$-wise indep.\\(all entries)};

\foreach \r in {1,...,\the\numexpr\Nrows-1}{
  \pgfmathsetmacro{\yoff}{((\Nrows+1)/2 - \r - 0.5) * \yspacing}
  \draw[blue!20, thin] ($(V1.center)+(-1.6,\yoff)$) -- ($(V1.center)+(1.6,\yoff)$);
}
\foreach \c in {1,...,\the\numexpr\Ncols-1}{
  \pgfmathsetmacro{\xoff}{(\c - (\Ncols+1)/2 + 0.5) * \xspacing}
  \draw[blue!20, thin] ($(V1.center)+(\xoff,-1.6)$) -- ($(V1.center)+(\xoff,1.6)$);
}

\node[font=\LARGE] at (2.2,0) {$\odot$};

\node[matbox, fill=red!6] (V2) at (5.2,0) {};
\node[above=-60pt] at (V2.south) {$V_2$};

\foreach \i/\c in {0/red!15, 1/red!8, 2/red!15, 3/red!8, 4/red!15}{
  \fill[\c] ($(V2.north west)+(0,-\i*0.64)$) rectangle ($(V2.north east)+(0,-\i*0.64-0.64)$);
}

\foreach \i in {1,...,4}{
  \draw[red!30, thin] ($(V2.north west)+(0,-\i*0.64)$) -- ($(V2.north east)+(0,-\i*0.64)$);
}
\draw[thick] (V2.south west) rectangle (V2.north east);

\foreach \i in {0,...,4}{
  \draw[->, thick, red!70!black] ($(V2.north west)+(-0.35,-\i*0.64-0.32)$) 
    -- ($(V2.north west)+(-0.08,-\i*0.64-0.32)$);
}

\draw[decorate, decoration={brace, amplitude=5pt, mirror}, thick, red!70!black]
  ($(V2.north west)+(-0.40,0)$) -- ($(V2.south west)+(-0.40,0)$)
  node[midway, left=8pt, font=\small, text=red!70!black, align=center, rotate=90, xshift=25pt, yshift=5pt] 
  {indep.\ rows};

\node[below=10pt, font=\small, text=red!70!black, align=center] at (V2.south) 
  {4-wise indep.\\(within each row)};

\foreach \r in {1,...,\Nrows}{
  \foreach \c in {1,...,\Ncols}{
    \pgfmathtruncatemacro{\h}{mod(3*\r + 4*\c + 1, 5)}
    \pgfmathtruncatemacro{\sign}{(-1)^\h}
    \pgfmathsetmacro{\xpos}{(\c - (\Ncols+1)/2) * \xspacing}
    \pgfmathsetmacro{\ypos}{((\Nrows+1)/2 - \r) * \yspacing}
    \ifnum\sign>0
      \node[font=\scriptsize, text=red!40!black] at ($(V2.center)+(\xpos,\ypos)$) {$+1$};
    \else
      \node[font=\scriptsize, text=red!40!black] at ($(V2.center)+(\xpos,\ypos)$) {$-1$};
    \fi
  }
}

\node[font=\LARGE] at (7.8,0) {$\odot$};

\node[matbox, fill=green!6] (V3) at (10.4,0) {};
\node[above=-60pt] at (V3.south) {$V_3^T$};

\foreach \i/\c in {0/green!15, 1/green!8, 2/green!15, 3/green!8, 4/green!15}{
  \fill[\c] ($(V3.north west)+(\i*0.64,0)$) rectangle ($(V3.south west)+(\i*0.64+0.64,0)$);
}

\foreach \i in {1,...,4}{
  \draw[green!30!black!20, thin] ($(V3.north west)+(\i*0.64,0)$) -- ($(V3.south west)+(\i*0.64,0)$);
}
\draw[thick] (V3.south west) rectangle (V3.north east);

\foreach \i in {0,...,4}{
  \draw[->, thick, green!50!black] ($(V3.north west)+(\i*0.64+0.32, 0.35)$) 
    -- ($(V3.north west)+(\i*0.64+0.32, 0.08)$);
}

\draw[decorate, decoration={brace, amplitude=5pt}, thick, green!50!black]
  ($(V3.north west)+(0,0.40)$) -- ($(V3.north east)+(0,0.40)$)
  node[midway, above=4pt, font=\small, text=green!50!black, align=center] 
  {indep.\ columns};

\node[below=10pt, font=\small, text=green!50!black, align=center] at (V3.south) 
  {4-wise indep.\\(within each col)};

\foreach \r in {1,...,\Nrows}{
  \foreach \c in {1,...,\Ncols}{
    \pgfmathtruncatemacro{\h}{mod(\r*4 + \c*3 + 2, 5)}
    \pgfmathtruncatemacro{\sign}{(-1)^\h}
    \pgfmathsetmacro{\xpos}{(\c - (\Ncols+1)/2) * \xspacing}
    \pgfmathsetmacro{\ypos}{((\Nrows+1)/2 - \r) * \yspacing}
    \ifnum\sign>0
      \node[font=\scriptsize, text=green!40!black] at ($(V3.center)+(\xpos,\ypos)$) {$+1$};
    \else
      \node[font=\scriptsize, text=green!40!black] at ($(V3.center)+(\xpos,\ypos)$) {$-1$};
    \fi
  }
}

\node[font=\LARGE] at (13.0,0) {$=$};
\node[font=\Large] at (14.0,0) {$V$};

\end{tikzpicture}
    \caption{Illustration of the pattern matrix generation model for $n=5$.}
    \label{fig:pattern}
\end{figure}

\paragraph{\textit{Necessity of including a sparse random perturbation}}
The perturbation $R_1$ is a random $\pm 1$ matrix which is easy to generate. Unfortunately, this matrix alone cannot provide any meaningful smallest singular value guarantee. In fact, in Lemma \ref{lem:pattinsuff}, we show that for any $n \times n$ deterministic matrix $\hat{V}$ with $\norm{\hat{V}}=1$, one can exhibit an $A$ and a sufficiently small $\varepsilon$ for which $A+\varepsilon D_1 \hat{V} D_2$ 
is singular with high probability. To address this problem we introduce another random perturbation.

\subsection{Construction of $R_2$ and invertibility control for incompressible vectors.}
\label{s:overview-R2}
To obtain a lower bound on $\norm{Mx}_2$ over incompressible vectors, we construct the sparse random matrix $R_2$. Briefly speaking, in $R_2$, we choose nonzero positions (where we place i.i.d. random signs) uniformly in each column, and then trim the rows. We start with a simple construction where the non-zeros are placed uniformly at random, and then we show how to modify this construction via row trimming to control the norm of the perturbation.

\subsubsection{Simple construction with i.i.d.\ sparse entries.}

 To minimize the number of random bits needed, the most natural choice is to consider a random matrix with sparse Rademacher entries. We first construct a hashing matrix $\bar{H}$ where the columns of $\bar{H}$ are independent $\{0,1\}$ random vectors containing $K$ ones to fix the positions of our random signs, where $K$ will be a constant that only depends on $\e$ and $\delta$, and does not depend on $n$ (for simplicity, we ignore the dependence on $\e$ and $\delta$ in this overview). 
More formally, for $i\in [n]$, let $J_i$ be independent subsets of $[n]=\{1,...,n\}$ of size $K$, uniformly distributed among such subsets. Then,  for $i,j \in [n]$, we set $\bar{h}_{j,i}= \mathds{1}_{J_i} (j) $.

Next, we just put random signs on the nonzero positions determined by $\bar H$ with appropriate scaling. More formally, we define a sequence $\{X_{l,i} : i \in [n], l \in [K]\}$ of $Kn$ i.i.d.\ Rademacher random variables, and enumerate the elements of each
random set $J_{i}, i\in [n]$:
\[ J_i = \{j(i,1),...,j(i,K)\} \text{ where } j(i,1) <j(i,2) <\ldots<j(i,K) \]
Then we can define the matrix $\bar{R}_2$ by
\begin{equation*}
    {[\bar{R}_2]}_{j(i,l),i} = \begin{cases}
        \frac{1}{L} \cdot X_{l,i}, &\text{ for } l\in [K] \\
        0 & \text{ otherwise}
    \end{cases}
\end{equation*}
where $L > K$ is a parameter to be chosen later.

The random matrix $\bar R_2$ indeed gives us desired control of the smallest singular value (nevertheless, $\bar R_2$ is not completely sufficient for our purpose, as explained later).

First of all, we observe that by conditioning on $R_1$, the problem reduces to bounding $\norm{(A_1+\varepsilon R_2)x}_2$ over incompressible vectors where $A_1=A+\varepsilon R_1$ can be considered deterministic because of conditioning. Next, we can use the standard invertibility by distance method as in \cite{rudelson2008littlewood, basakrudelson}, which gives optimal control of $\inf_{x \in \incomp(\alpha, \nu)} \norm{Mx}_2$ (here optimal means  as good as for a dense gaussian perturbation). This involves reducing the problem of finding the infimum over incompressible vectors to finding a lower bound for the distance between a column of $M$ and the subspace spanned by the remaining columns of $M$. For concreteness, since the columns are identically distributed, we consider the distance between the last column and the subspace spanned by the first $n-1$ columns. The techniques for analyzing this distance use the independence between columns to condition on the first $n-1$ columns, which fixes the normal (i.e., orthogonal) vector to the subspace spanned by them, and then use the randomness of the last column to control the inner product between the normal vector and the last column, i.e., $\mathbb{P}\left(|\langle Y_n, Z \rangle|<t\right)$, where $Y_n$ is the last column of $A_1+\varepsilon \bar R_2$ and $Z$ is a random normal vector of the subspace spanned by the first $n-1$ columns. Although \cite{rudelson2008littlewood, basakrudelson} only considered the case without the shift $A_1$, adding this shift will not change the argument because the small ball probability bound we use to control $\mathbb{P}\left(|\langle Y_n, Z \rangle|<t\right)$ is not affected by deterministic~shifts. Note that Coupon Collector argument shows that if $\bar R_2$ were the whole perturbation, then with only $K=O(1)$ nonzeros per column it misses about \(e^{-K}n\) rows, so many coordinates receive no randomness at all, and this is why a purely sparse \(O(n)\)-nonzero perturbation cannot reliably fix every bad direction. However, using $R_1+\bar R_2$ avoids this obstruction because  \(R_1\) already controls the sparse directions, and \(R_2\) is used only when the relevant vector is spread over \(\alpha n\) coordinates, where hitting a few useful coordinates is easy even though covering all \(n\) rows is impossible.

\subsubsection{Controlling the norm and row trimming.}

The natural construction $\bar R_2$ with i.i.d.\ sparse entries gives optimal control on the smallest singular value with $K=O(1)$ nonzero entries per column. However, in this case, the norm of $R_2$ grows with $n$ since some rows may have more than $O(1)$ many entries with absolute value $\ge 1$ (in other words, some of the rows are heavy). Since we want our perturbations to have $O(1)$ norm, we can scale by the norm of $\bar R_2$, but this makes the resulting smallest singular value bound suboptimal by a $\log n$ factor, and cannot give us the $O(n)$ bound on the condition number to match Gaussian perturbations.

We resolve this issue by an intuitive modification, which is to truncate the matrix $\bar R_2$ by simply removing the problematic heavy rows. Formally, we choose an $L>K$ such that $L = O(1)$ and consider rows with more than $L$ nonzero entries, calling them \textit{heavy} rows. We show that the number of heavy rows is not too large, and zero them out. After this step, all columns have at most $K$ nonzero entries and all rows have at most $L$ nonzero entries, all with absolute value 1. We can now scale the resulting matrix by $1/L$ to ensure a norm bound of $1$ without degrading the smallest singular value estimate. 

Formally, we set $\text{Row}_j(H) = \text{Row}_j(\bar{H})$ if the number of non-zero entries in $\text{Row}_j(\bar{H})$ does not exceed $L$ and $\text{Row}_j(H) = 0$ otherwise. Note that the number
of non-zero entries in each row and column of the resulting matrix $H$ is at most $L$. Finally, we set $R_2 = H \odot \bar R_2$. In summary, the nonzero positions (where we place i.i.d. random signs) are uniformly chosen in each column in $R_2$, and then it is the rows that are being trimmed (see Figure \ref{fig:R2}).

\begin{figure}[t]
\centering
\begin{adjustbox}{max width=\textwidth}
\begin{tikzpicture}[>=Stealth]

\matrix (Rbar2) [
  matrix of nodes,
  nodes={draw, minimum size=8mm, font=\tiny, inner sep=0.6pt, align=center},
  row sep=-\pgflinewidth,
  column sep=-\pgflinewidth
] {
  |[fill=green!12]|$\tfrac{1}{L}X_{1,1}$ & \textcolor{gray}{0} & |[fill=green!12]|$\tfrac{1}{L}X_{1,3}$ & |[fill=green!12]|$\tfrac{1}{L}X_{1,4}$ & \textcolor{gray}{0} & \textcolor{gray}{0} & |[fill=green!12]|$\tfrac{1}{L}X_{1,7}$ & \textcolor{gray}{0}\\
  \textcolor{gray}{0} & \textcolor{gray}{0} & \textcolor{gray}{0} & \textcolor{gray}{0} & |[fill=green!12]|$\tfrac{1}{L}X_{1,5}$ & \textcolor{gray}{0} & \textcolor{gray}{0} & |[fill=green!12]|$\tfrac{1}{L}X_{1,8}$\\
  \textcolor{gray}{0} & \textcolor{gray}{0} & |[fill=green!12]|$\tfrac{1}{L}X_{2,3}$ & \textcolor{gray}{0} & \textcolor{gray}{0} & |[fill=green!12]|$\tfrac{1}{L}X_{1,6}$ & \textcolor{gray}{0} & \textcolor{gray}{0}\\
  \textcolor{gray}{0} & |[fill=green!12]|$\tfrac{1}{L}X_{1,2}$ & \textcolor{gray}{0} & \textcolor{gray}{0} & \textcolor{gray}{0} & \textcolor{gray}{0} & |[fill=green!12]|$\tfrac{1}{L}X_{2,7}$ & \textcolor{gray}{0}\\
  |[fill=green!12]|$\tfrac{1}{L}X_{2,1}$ & \textcolor{gray}{0} & \textcolor{gray}{0} & \textcolor{gray}{0} & |[fill=green!12]|$\tfrac{1}{L}X_{2,5}$ & \textcolor{gray}{0} & \textcolor{gray}{0} & \textcolor{gray}{0}\\
  \textcolor{gray}{0} & |[fill=green!12]|$\tfrac{1}{L}X_{2,2}$ & \textcolor{gray}{0} & \textcolor{gray}{0} & \textcolor{gray}{0} & \textcolor{gray}{0} & \textcolor{gray}{0} & \textcolor{gray}{0}\\
  \textcolor{gray}{0} & \textcolor{gray}{0} & \textcolor{gray}{0} & \textcolor{gray}{0} & \textcolor{gray}{0} & |[fill=green!12]|$\tfrac{1}{L}X_{2,6}$ & \textcolor{gray}{0} & \textcolor{gray}{0}\\
  \textcolor{gray}{0} & \textcolor{gray}{0} & \textcolor{gray}{0} & |[fill=green!12]|$\tfrac{1}{L}X_{2,4}$ & \textcolor{gray}{0} & \textcolor{gray}{0} & \textcolor{gray}{0} & |[fill=green!12]|$\tfrac{1}{L}X_{2,8}$\\
};

\draw[thick] (Rbar2-1-1.north west) rectangle (Rbar2-8-8.south east);

\node[font=\normalsize] at ($(Rbar2-1-4.north)+(0,10mm)$) {$\bar R_2$};

\foreach \c in {1,...,8} {
  \node[font=\scriptsize] at ($(Rbar2-1-\c.north)+(0,2.6mm)$) {$\c$};
}
\foreach \r in {1,...,8} {
  \node[font=\scriptsize] at ($(Rbar2-\r-1.west)+(-2.8mm,0)$) {$\r$};
}
\node[font=\scriptsize] at ($(Rbar2-1-1.north west)+(-3.4mm,2.6mm)$) {$j\backslash i$};

\draw[red,very thick] (Rbar2-1-1.west) -- (Rbar2-1-8.east);

\matrix (R) [
  matrix of nodes,
  nodes={draw, minimum size=8mm, font=\tiny, inner sep=0.6pt, align=center},
  row sep=-\pgflinewidth,
  column sep=-\pgflinewidth,
  right=16mm of Rbar2.east,
  anchor=west
] {
  |[fill=gray!18]|\textcolor{gray}{0} & |[fill=gray!18]|\textcolor{gray}{0} & |[fill=gray!18]|\textcolor{gray}{0} & |[fill=gray!18]|\textcolor{gray}{0} & |[fill=gray!18]|\textcolor{gray}{0} & |[fill=gray!18]|\textcolor{gray}{0} & |[fill=gray!18]|\textcolor{gray}{0} & |[fill=gray!18]|\textcolor{gray}{0}\\
  \textcolor{gray}{0} & \textcolor{gray}{0} & \textcolor{gray}{0} & \textcolor{gray}{0} & |[fill=green!12]|$\tfrac{1}{L}X_{1,5}$ & \textcolor{gray}{0} & \textcolor{gray}{0} & |[fill=green!12]|$\tfrac{1}{L}X_{1,8}$\\
  \textcolor{gray}{0} & \textcolor{gray}{0} & |[fill=green!12]|$\tfrac{1}{L}X_{2,3}$ & \textcolor{gray}{0} & \textcolor{gray}{0} & |[fill=green!12]|$\tfrac{1}{L}X_{1,6}$ & \textcolor{gray}{0} & \textcolor{gray}{0}\\
  \textcolor{gray}{0} & |[fill=green!12]|$\tfrac{1}{L}X_{1,2}$ & \textcolor{gray}{0} & \textcolor{gray}{0} & \textcolor{gray}{0} & \textcolor{gray}{0} & |[fill=green!12]|$\tfrac{1}{L}X_{2,7}$ & \textcolor{gray}{0}\\
  |[fill=green!12]|$\tfrac{1}{L}X_{2,1}$ & \textcolor{gray}{0} & \textcolor{gray}{0} & \textcolor{gray}{0} & |[fill=green!12]|$\tfrac{1}{L}X_{2,5}$ & \textcolor{gray}{0} & \textcolor{gray}{0} & \textcolor{gray}{0}\\
  \textcolor{gray}{0} & |[fill=green!12]|$\tfrac{1}{L}X_{2,2}$ & \textcolor{gray}{0} & \textcolor{gray}{0} & \textcolor{gray}{0} & \textcolor{gray}{0} & \textcolor{gray}{0} & \textcolor{gray}{0}\\
  \textcolor{gray}{0} & \textcolor{gray}{0} & \textcolor{gray}{0} & \textcolor{gray}{0} & \textcolor{gray}{0} & |[fill=green!12]|$\tfrac{1}{L}X_{2,6}$ & \textcolor{gray}{0} & \textcolor{gray}{0}\\
  \textcolor{gray}{0} & \textcolor{gray}{0} & \textcolor{gray}{0} & |[fill=green!12]|$\tfrac{1}{L}X_{2,4}$ & \textcolor{gray}{0} & \textcolor{gray}{0} & \textcolor{gray}{0} & |[fill=green!12]|$\tfrac{1}{L}X_{2,8}$\\
};

\draw[thick] (R-1-1.north west) rectangle (R-8-8.south east);
\node[font=\normalsize] at ($(R-1-4.north)+(0,10mm)$) {$R_2$};

\draw[->,thick] ($(Rbar2.east)+(2mm,0)$) -- node[above,font=\scriptsize] {$\odot\,H$} ($(R.west)+(-2mm,0)$);

\node[red,font=\scriptsize,align=center] (ann)
  at ($ (Rbar2-1-4.north)!0.5!(R-1-4.north) + (0,10mm) $)
  {$\mathrm{Row}_1(H)=0$\\row $1$ is zeroed};
\draw[red,->] (ann.west) .. controls +(-12mm,-10mm) and +(10mm,6mm) .. (Rbar2-1-8.east);

\end{tikzpicture}
\end{adjustbox}
\caption{Illustration of $R_2=H\odot \bar R_2$ for one realization. Since $\mathrm{Row}_1(H)=0$, the first row of $\bar R_2$ is zeroed in $R_2$.}\label{fig:R2}
\end{figure}

\subsubsection{Modifying the invertibility by distance argument for row-trimmed matrix $R_2$.} Unfortunately, these modifications break the previous conditioning argument since the entries of the first $n-1$ columns depend on the last column. The main novelty of our argument for controlling \[\inf_{x \in \incomp(\alpha, \nu)} \norm{Mx}_2\] is a modification of the invertibililty by distance argument to make it work here for our matrix $R_2$ which has dependent columns. In this direction, the first attempt is to condition on the positions of the nonzero entries in the last column (and not the random signs) in addition to the first $n-1$ columns. This way, we separate out the random signs $X_{l,n}$ in the last column which are independent from the part that we condition on. More precisely, recalling that $Y_n$ denotes the last column of $M$ and $Z$ is a unit vector orthogonal to the subspace spanned by the first $n-1$ columns, we can write
\begin{align*}
			\langle Y_n,Z\rangle
			=
			\big\langle \mathrm{Col}_n(A+\varepsilon R_1),\, Z\big\rangle
			+\frac{\varepsilon}{L}\sum_{l=1}^K X_{l,n}\, Z_{j(n,l)}\, \mathbf{1}_{\{j(n,l) \text{ will not be removed}\}}
		\end{align*}
which, after conditioning, becomes a linear combination of independent random signs $X_{l,n}$   with coefficients $W_l=Z_{j(n,l)}\, \mathbf{1}_{\{j(n,l) \text{ will not be removed}\}}$ (with a deterministic shift $\big\langle \mathrm{Col}_n(A+\varepsilon R_1),\, Z\big\rangle$ which does not affect Levy concentration).

Nevertheless, separating out the random signs $X_{l,n}$ does not make the argument work directly, because to control the small ball probability, we need to have sufficiently many good coordinates (here good means that the corresponding coefficient is nonzero and with magnitude of order $1/\sqrt{n}$) in the coefficients $(W_1,...,W_K)$, which requires additional effort to show. 

First, we show that the random vector $Z$ is incompressible with high probability (this is where the other random diagonal matrix $D_2$ from $R_1$ is needed). This means $Z$ has many coordinates of order $1/\sqrt{n}$, called ``spread'' coordinates (Lemma \ref{lem:sharpened-spread}), denoted by $\Theta_0$. Then, the set of the good coordinates is the intersection of $\Theta_0$, the nonzero positions selected by $J_n$, and the row indices that will not be removed. The difficulty is that these three sets are dependent. To unravel this dependence,  we define the index set of  heavy rows of $\bar H$ up to column $n-1$ and the index set of the light rows up to column $n-1$  as
\begin{align*}
		I_{\mathrm{heavy},n}
		=\bigg\{ i\in[n]: \sum_{j=1}^{n-1}\bar h_{i,j} \ge L\bigg\} \qquad\text{and}\qquad I_{\mathrm{light},n}
		=\bigg\{ i\in[n]: \sum_{j=1}^{n-1}\bar h_{i,j} < L\bigg\},
	\end{align*}
and observe that a nonzero position selected by $J_n$ will eventually reset the whole row to zero if and only if it is in a row that is already heavy enough in the $1,...,n-1$ columns. Therefore, the set of good coordinates must contain $\Theta_0 \cap J_n \cap I_{\mathrm{light},n}$. We observe that $Z$, and therefore $\Theta_0$ (the spread coordinates of $Z$), depends only on the first $n-1$ columns of $R_2$, which, in turn only depend on the first $n-1$ columns of $\bar R_2$ and $J_n\cap I_{\mathrm{heavy},n}$, because the nonzero entries in $J_n\cap I_{\mathrm{light},n}$ in the last column (i.e., what $J_n$ selects in $I_{\mathrm{light},n}$) will not cause a row deletion in any case.
This allows us to improve the naive conditioning argument described above. We condition on the first $n-1$ columns of $\bar R_2$ and only on the set $J_n \cap I_{\mathrm{heavy}, n}$ instead of the whole set $J_n$. This conditioning  fixes the set $\Theta_0$. Moreover, by choosing an appropriately large $L$, we can ensure that the set $\Theta_0 \cap I_{\mathrm{light},n}$ is large with high probability.
After that, we can take advantage of the extra randomness in the choice of the remaining part of $J_n$ to guarantee that, 
with high probability, sufficiently many of elements of the set $J_n \cap I_{\mathrm{light},n}$  fall in $\Theta_0 \cap I_{\mathrm{light},n}$. 

More formally, we can show that, conditioning on  the first $n-1$ columns of $\bar R_2$ and $J_n\cap I_{\mathrm{heavy},n}$, the set $J_n \cap I_{\mathrm{light},n}$ is a random uniform subset of $I_{\mathrm{light},n}$ with cardinality $K-|J_n\cap I_{\mathrm{heavy},n}|$ (see Lemma~\ref{lem:spreadlight} and Corollary~\ref{lem:lightspread} for details). And therefore, we can show that, with high probability, we have sufficiently many good coordinates in $(W_1,...,W_K)$ to control the small ball probability of
\begin{align*}
			\langle Y_n,Z\rangle
            =\big\langle \mathrm{Col}_n(A+\varepsilon R_1),\, Z\big\rangle
			+\frac{\varepsilon}{L}\sum_{l=1}^K X_{l,n}W_l
		\end{align*}
        and therefore prove invertibility for incompressible vectors (see Theorem~\ref{thm:incomp} for details).


\subsection{Outline of the paper.}
Section \ref{sec:preliminaries} explains the notation used in the paper and provides details of some basic mathematical estimates we use in the rest of the paper. Section \ref{sec:model} provides the full details of our perturbation model. Section \ref{sec:pattern} establishes that our model for pattern matrices achieves the required properties with high probability. Section \ref{sec:comp} and Section \ref{sec:incomp} prove the singular value lower bound over compressible and incompressible vectors, respectively. Section \ref{sec:mainthm} combines these lower bounds to prove Theorem \ref{coro:pertur}.

\section{Preliminaries.} \label{sec:preliminaries}

\subsection{Notation.} \label{subsec:notation}
The following notation and terminology will be used in the paper. The notation $[n]$ is used for the set $\{1,2,...,n\}$. Also, for two integers $a$ and $b$ with $a \le b$, we use the notation $[a:b]$ for the set $\{k \in \Z:a \le k \le b\}$. For $x \in \R$, we use the notation $\lfloor x \rfloor$ to denote the greatest integer less than or equal to $x$ and $\lceil x \rceil$ to denote the least integer greater than or equal to  $x$. In $\R^n$ (or $\R^m$ or $\R^d$), the $l$th coordinate vector is denoted by $e_l$. 

We use the notation $\mathbb{P}$ for the standard probability measure, and the notation $\mathbb{E}$ for the expectation with respect to this standard probability measure. The covariance of two random variables $X$ and $Y$ is denoted by $\cov(X,Y)$. The standard $L_q$ norm of a random variable $\xi$ is denoted by $\norm{\xi}_{L_q(\Pb)}$, for $1 \le q \le \infty$. When there is no possibility of confusion, we will use the short notation $\norm{\xi}_{q}$ for the $L_q$ norm of $\xi$. For convenience of iterated integration, we introduce the following notations. Let $X,Y$ be two random elements taking values in measurable spaces $\mathscr{S}_1$ and $\mathscr{S}_2$. Let $\phi:\mathscr{S}_1 \times \mathscr{S}_2 \to \R$ be a measurable function. Assume that $\phi(X,Y)$ is integrable. Define the function $\psi_2:\mathscr{S}_2 \to \R$ such that $\psi_2(y)=\E\phi(X,y)$. Then the notation $\E_X(\phi(X,Y))$ stands for $\psi_2(Y)$. Similarly, we define the function $\psi_1:\mathscr{S}_1 \to \R$ such that $\psi_1(x)=\E\phi(x,Y)$ and the notation $\E_Y(\phi(X,Y))$ stands for the random variable $\psi_1(X)$.

All matrices considered in this paper are real valued and the space of $m \times n$ matrices with real valued entries is denoted by $M_{m \times n}(\mathbb{R})$. Also, for a matrix $X \in M_{d \times d}(\mathbb{R})$, the notation $\Tr (X)$ denotes the trace of the matrix $X$, and $\tr (X) = \frac{1}{d} \Tr (X)$ denotes the normalized trace. The $d \times d$ identity matrix is denoted by $I_d$. We write the operator norm of a matrix $X$ as $\norm{X}$, and it is also denoted by $\norm{X}_{op}$ in some places where other norms appear for clarity. The spectrum of a matrix $X$ is denoted by $\spec(X)$. For matrices $A$ and $B$ with same dimensions, $A \odot B$ stands for the entry-wise (Hadamard) product of $A$ and $B$, i.e., $(A \odot B)_{ij}=A_{ij} \cdot B_{ij}$.

Throughout the paper, the symbols $c_1, c_2, ...$, and $Const, Const', ...$ denote absolute constants. 

\subsection{Decomposition of the sphere.}

Our method to bound the smallest singular value of the perturbed matrix uses the standard approach that separately control the invertibility for compressible vectors and for incompressible vectors. So we introduce these formal definition.
\begin{definition}
    Let $\alpha,\nu\in(0,1)$. A vector $x\in\mathbb{R}^n$ is called \emph{sparse} if $|\operatorname{supp}(x)|\le \alpha n$, and the set of sparse vectors is denoted by $\sparse(\alpha)$. A vector $x\in S^{n-1}$ is called \emph{compressible} if $x$ is within Euclidean distance $\nu$ from the set of all sparse vectors. A vector $x\in S^{n-1}$ is called \emph{incompressible} if it is not compressible. The sets compressible and incompressible vectors will be denoted by
\begin{align*}
    \comp(\alpha, \nu) &= \{ x \in \bS^{n-1} \vert \norm{x - v}_2 \le \nu \text{ for some } v \in \sparse(\alpha)\} \\
    \incomp(\alpha, \nu) &= \bS^{n-1} \backslash \comp(\alpha, \nu)
\end{align*}

\end{definition}

\begin{corollary}
		\label{lem:compressible-tail}
		Let \(0<\alpha\le1\) and \(\nu>0\).  If
		\(x\in\comp(\alpha,\nu)\), then there exists \(S\subset[n]\) such that
		\[
		|S|\le \alpha n
		\qquad\text{and}\qquad
		\norm{x_{[n]\setminus S}}_2\le \nu.
		\]
	\end{corollary}
	
	\begin{proof}
		Since \(x\in\comp(\alpha,\nu)\), by definition,
		there exists \(x\in\sparse(\alpha)\) such that
		\[
		\norm{x-v}_2\le\nu.
		\]
		Set \(S=\supp(v)\).  Then \(|S|\le\alpha n\).  Since \(v_{[n]\setminus S}=0\), we have
		\[
		x_{[n]\setminus S}=(x-v)_{[n]\setminus S}.
		\]
		Therefore, we conclude
		\[
		\norm{x_{[n]\setminus S}}_2
		=
		\norm{(x-v)_{[n]\setminus S}}_2
		\le
		\norm{x-v}_2
		\le \nu.
		\]
	\end{proof}

\begin{lemma}[Counting sparse supports]\label{lem:support-count}
For every $n\ge1$ and every $0<\alpha\le1/2$,
\[
 \left|\{S\subset[n]: |S|\le \alpha n\}\right|
 \le \exp\left(\alpha n\log\frac{e}{\alpha}\right).
\]
\end{lemma}
\begin{proof}
    The results follow from direct counting.
\end{proof}

\subsection{Covering numbers.}

We will need the following facts for covering numbers. For a set $K\subset\R^n$ and $\varepsilon>0$, $\mathcal N(K,\varepsilon)$ denotes the smallest cardinality of an $\varepsilon$-net of $K$ in the Euclidean metric, and $\mathcal P(K,\varepsilon)$ denotes the largest cardinality of an $\varepsilon$-separated subset of $K$ in the Euclidean metric.  The notation $|\cdot|$ denotes Lebesgue volume.

\begin{proposition}[Proposition 4.2.12 in \cite{vershynin2018high}]
\label{prop:vershynin-4212}
Let $K$ be a subset of $\R^n$ and $\varepsilon>0$.  Then
\[
 \frac{|K|}{|\varepsilon B_2^n|}
 \le
 \mathcal N(K,\varepsilon)
 \le
 \mathcal P(K,\varepsilon)
 \le
 \frac{|K+(\varepsilon/2)B_2^n|}{|(\varepsilon/2)B_2^n|}.
\]
Here $|\cdot|$ denotes the volume in $\R^n$, $B_2^n$ denotes the unit Euclidean ball in $\R^n$, so $\varepsilon B_2^n$ is a Euclidean ball with radius $\varepsilon$.
\end{proposition}

\subsection{L\'evy concentration.}
As is explained in the outline, the control of invertibility for incompressible vectors needs to use the control for L\'evy concentration. So we introduce the formal definition and standard results below.
    	\begin{definition}[L\'evy concentration function]
		For a real-valued random variable $Z$ and $r\ge 0$, its concentration function is
		\begin{equation}
			\label{eq:Levy}
			\mathcal L(Z,r):=\sup_{u\in\mathbb R}\mathbb P\{|Z-u|\le r\}.
		\end{equation}
	\end{definition}

\begin{lemma}[Scaling of L\'evy concentration]
		\label{lem:scale}
		Let $Z$ be real-valued and let $a\neq 0$. Then for all $r\ge 0$,
		\begin{equation}
			\label{eq:scale}
			\mathcal L(aZ,r)=\mathcal L\!\left(Z,\frac{r}{|a|}\right).
		\end{equation}
		In particular, if $a>0$ then $\mathcal L(Z/a,r)=\mathcal L(Z,ar)$.
	\end{lemma}
\begin{proof}
The result follows from straightforward calculation.

\end{proof}

\begin{lemma}[Rogozin's theorem {\cite[Theorem~1]{rogozin}}]
		\label{lem:rogozin}
		Let $\xi_1,\dots,\xi_n$ be independent real-valued random variables, let $r>0$, and let $r_1,\dots,r_n\in (0,r]$.
		Assume that
		\[
		\sum_{i=1}^n \bigl(1-\mathcal L(\xi_i,r_i)\bigr)r_i^2>0.
		\]
		Then
		\begin{equation}
			\label{eq:rogozin}
			\mathcal L\!\left(\sum_{i=1}^n \xi_i,r\right)\le \frac{c_{\ref{lem:rogozin}}\,r}{\sqrt{\sum_{i=1}^n \bigl(1-\mathcal L(\xi_i,r_i)\bigr)r_i^2}},
		\end{equation}
		where $c_{\ref{lem:rogozin}}>0$ is a universal constant.
\end{lemma}

\subsection{Measure concentration.}

We will need the following standard bounded difference inequality from \cite{boucheron2013concentration}.
\begin{lemma}[Bounded differences inequality]
		\label{lem:bd}
		Let $X_1,\dots,X_m$ be independent random elements taking values in a measurable space $\mathcal X$.
		Let $f:\mathcal X\times\cdots\times\mathcal X\to\mathbb{R}$ and assume that there exist constants $b_1,\dots,b_m\ge 0$ such that for every $k\in[m]$ and every two points $x,x'$ that differ only in the $k$th coordinate,
		\[
		|f(x)-f(x')|\le b_k.
		\]
		Define
		\[
		v:=\frac14\sum_{k=1}^m b_k^2.
		\]
		Then for every $t>0$, if $Z:=f(X_1,\dots,X_m)$,
		\[
		\mathbb{P}\{Z-\mathbb{E}Z>t\}\le \exp\Bigl(-\frac{t^2}{2v}\Bigr),
		\qquad
		\mathbb{P}\{Z-\mathbb{E}Z<-t\}\le \exp\Bigl(-\frac{t^2}{2v}\Bigr).
		\]
	\end{lemma}
	
\subsection{Random sampling without replacement.} Random sampling without replacement plays an important role in our construction of $R_2$ and in our proof for the invertibility for incompressible vectors. So we state the formal definitions (see e.g., \cite{anderson2017introduction}) and standard results  below.

\begin{definition}[Ordered random sampling without replacement]\label{def:orsamp}
    Let $k \le n$ be integers. Let $(\Omega,\mathscr{A}, \Pb)$ be a probability space. Let $\mathcal{Z}$ be the set of injections from $[k]$ to $[n]$.
    Let $T: \Omega \to \mathcal{Z}$ be a random variable such that for any $z \in \mathcal{Z}$, we have
    \begin{align*}
        \Pb(T=z)=\frac{1}{n \cdot (n-1) \cdots (n-k+1)}
    \end{align*}
    Then we say $T$ is a random sampling of an ordered \(k\)-element sequence uniformly from $[n]$ without replacement.
    Also, if $\mathcal S=(s_i)_{i \in [n]}=(s_1,\dots,s_n)$ is a finite ordered sequence, and $S(\omega)=s \circ T(\omega)$, i.e., $S(l)(\omega)=s_{T(l)(\omega)}$ for $l=1,2,...,k$. Then $S=(S(1),...,S(k))$  is called a random sampling of an ordered \(k\)-element sequence uniformly from population $\mathcal S$ without replacement.
\end{definition}

We will use the following results from \cite{hoeffding1963probability} to get concentration results for sum of random sampling without replacement.

\begin{lemma}[Theorem 4 in \cite{hoeffding1963probability}]\label{lem:BM11}
		Let $\mathcal S=(s_1,\dots,s_n)$ be a finite population of $n$ real points.
		Let $(X(1),\dots,X(k))$ denote a random sampling of an ordered \(k\)-element sequence uniformly from population $\mathcal S$ without replacement,
		and let $(Y(1),\dots,Y(k))$ denote a random sample \emph{with replacement} from $\mathcal S$ (i.i.d. random samples from $\mathcal S$).
		If $f:\mathbb R\to\mathbb R$ is continuous and convex, then
		\[
		\E f\Bigl(\sum_{i=1}^k X(i)\Bigr)\ \le\ \E f\Bigl(\sum_{i=1}^k Y(i)\Bigr).
		\]
	\end{lemma}
Combining Lemma~\ref{lem:BM11} with standard Hoeffding inequality for
i.i.d. sampling (i.e., with replacement), we obtain the following Hoeffding's inequality for random sampling without replacement.
	\begin{proposition}[Hoeffding's inequality for random sampling without replacement]\label{prop:BM12}
		Let $\mathcal S=(s_1,\dots,s_n)$ be a finite population of $n$ real points and let
		$(S(1),\dots,S(k))$ be a random sampling of an ordered \(k\)-element sequence uniformly from population $\mathcal S$ without replacement.
		Let
		\[
		a=\min_{1\le i\le n} s_i,\qquad b=\max_{1\le i\le n} s_i,
		\qquad
		\mu=\frac{1}{n}\sum_{i=1}^n s_i.
		\]
		Then for all $t>0$, we have
		\begin{equation}\label{eq:BM-hoeffding-upper}
			\Pb\Bigl(\frac1k\sum_{i=1}^k S(i)-\mu\ge t\Bigr)
			\le \exp\Bigl(-\frac{2kt^2}{(b-a)^2}\Bigr).
		\end{equation}
	\end{proposition}
	
	\begin{remark}[Lower tail]\label{rem:BM-lower}
		Applying Proposition~\ref{prop:BM12} to the population $(-s_1,\dots,-s_n)$ we can also obtain the lower tail bound
		\begin{equation}\label{eq:BM-hoeffding-lower}
			\Pb\Bigl(\frac1k\sum_{i=1}^k S(i)-\mu\le -t\Bigr)
			\le \exp\Bigl(-\frac{2kt^2}{(b-a)^2}\Bigr),
			\qquad t>0.
		\end{equation}
	\end{remark}

\begin{definition}[Unordered random sampling without replacement]\label{def:unorsamp}
Let $k \le n$ be integers. Let $(\Omega,\mathscr{A}, \Pb)$ be a probability space. Let $\mathcal{Z}$ be the set of injections from $[k]$ to $[n]$. Let $\mathcal S=(s_i)_{i \in [n]}=(s_1,\dots,s_n)$ is a finite ordered sequence. And assume that \(s_1,\dots,s_n\) are pairwise distinct. Let $\mathcal J=\mathrm{Range}(s)=\{s_1,\dots,s_n\}$.
Let \(T:\Omega\to\mathcal Z\) be a random sampling of an ordered \(k\)-element sequence uniformly from \([n]\)
without replacement. Define
\[
J(\omega)=\{s_{T(\omega)(1)},\dots,s_{T(\omega)(k)}\}=\mathrm{Range}(s \circ T(\omega)).
\]
Then \(J\) is called a random sampling of an unordered \(k\)-element subset
uniformly from $\mathcal J$ without replacement.
\end{definition}

\begin{remark}
    Let $\mathcal S=(s_i)_{i \in [n]}=(s_1,\dots,s_n)$ is a finite ordered sequence. And assume that \(s_1,\dots,s_n\) are pairwise distinct. Let $\mathcal J=\mathrm{Range}(s)=\{s_1,\dots,s_n\}$. Then the followings hold:
    
    (i) If \(J\) is a random sampling of an unordered \(k\)-element subset
uniformly from $\mathcal J$ without replacement, then for every $B \subset \mathcal J$ with $\card(B)=k$, we have
\[
\Pb(J=B)=\frac{1}{\binom{n}{k}},
\]  

    (ii) If $J$ is a random set such that for every $B \subset \mathcal J$ with $\card(B)=k$, we have
\[
\Pb(J=B)=\frac{1}{\binom{n}{k}},
\]
then (in a expanded probability space if necessary) \(J\) is a random sampling of an unordered \(k\)-element subset
uniformly from $\mathcal J$ without replacement.
\end{remark}
	
	\begin{lemma}[Conditional law of $\Gamma_2$ given $\Gamma_1$]\label{lem:cond-unif-mu}
		Fix $n\in\mathbb N$ and $K\in\{0,1,\dots,n\}$.
		Let $J$ be a random sampling of an unordered \(K\)-element subset
uniformly from $[n]$ without replacement. and let $I\subset[n]$ be deterministic.
		Define
		\[
		\Gamma_1:=J\cap I,
		\qquad
		\Gamma_2:=J\cap\bigl([n]\setminus I\bigr).
		\]
		Let $\mu(\theta_1,B)$ be a regular conditional probability of $\Gamma_2$ given $\Gamma_1$, i.e., $\mu$ is a probability kernel such that for any measurable set $B$ (in the space where $\Gamma_2$ takes values), we have $\mu(\Gamma_1,B)=\Pb(\Gamma_2 \in B|\Gamma_1)$ almost surely.
		Then for every fixed $\theta_1\subset I$ with $\abs{\theta_1}\le K$:
		\begin{enumerate}
			\item If $K-\abs{\theta_1}>\abs{[n]\setminus I}$, then $\Pb(\Gamma_1=\theta_1)=0$.
			\item If $K-\abs{\theta_1}\le \abs{[n]\setminus I}$, then on the event $\{\Gamma_1=\theta_1\}$, for every $\theta_2\subset [n]\setminus I$,
			\[
			\mu(\theta_1,\{\theta_2\})=
			\begin{cases}
				\binom{\abs{[n]\setminus I}}{K-\abs{\theta_1}}^{-1}, & \text{if }\abs{\theta_2}=K-\abs{\theta_1},\\[4pt]
				0, & \text{otherwise.}
			\end{cases}
			\]
			In other words, conditional on $\Gamma_1$, the random set $\Gamma_2$ is a random sampling of an unordered  $(K-\abs{\Gamma_1})$-element subset
uniformly from  $[n]\setminus I$ without replacement, or more precisely, any random element with law of distribution $\mu(\theta_1,\cdot )$ is a random sampling of an unordered  $(K-\abs{\theta_1})$-element subset
uniformly from  $[n]\setminus I$ without replacement (in an expanded probability space if necessary).
		\end{enumerate}
	\end{lemma}
	
	\begin{proof}
		Fix $\theta_1\subset I$ with $\abs{\theta_1}\le K$ and put $s:=K-\abs{\theta_1}$.
		If $s>\abs{[n]\setminus I}$ then there is no $K$--subset $J$ with $J\cap I=\theta_1$, so $\Pb(\Gamma_1=\theta_1)=0$.
		
		It therefore suffices to consider $s\le \abs{[n]\setminus I}$.
		Let $\theta_2\subset[n]\setminus I$.
		If $\abs{\theta_2}\neq s$, then the event $\{\Gamma_1=\theta_1,\Gamma_2=\theta_2\}$ is empty, so
		$\Pb(\Gamma_2=\theta_2,\Gamma_1=\theta_1)=0$ and therefore $\mu(\theta_1,\{\theta_2\})=0$ on $\{\Gamma_1=\theta_1\}$.
		
		If $\abs{\theta_2}=s$, then $\Gamma_1=\theta_1$ and $\Gamma_2=\theta_2$ is equivalent to $J=\theta_1\cup \theta_2$ (disjoint union).
		Since $J$ is uniform on $\binom{[n]}{K}$, we have
		\[
		\Pb(\Gamma_2=\theta_2,\Gamma_1=\theta_1)=\Pb(J=\theta_1\cup \theta_2)=\binom{n}{K}^{-1}.
		\]
		In addition, we have
		\[
		\Pb(\Gamma_1=\theta_1)=\sum_{\substack{\theta_2'\subset[n]\setminus I\\ \abs{\theta_2'}=s}}
		\Pb(J=\theta_1\cup \theta_2')
		=\binom{\abs{[n]\setminus I}}{s}\binom{n}{K}^{-1}.
		\]
		Combining above calculations, we have
		\[
		\Pb(\Gamma_2=\theta_2\mid \Gamma_1=\theta_1)
		=\frac{\Pb(\Gamma_2=\theta_2,\Gamma_1=\theta_1)}{\Pb(\Gamma_1=\theta_1)}
		=\binom{\abs{[n]\setminus I}}{s}^{-1}.
		\]
		Since $\Gamma_1$ takes values in a finite set, a version of $\mu$ can be chosen so that
		$\mu(\theta_1,\{\theta_2\})=\Pb(\Gamma_2=\theta_2\mid \Gamma_1=\theta_1)$ whenever $\Pb(\Gamma_1=\theta_1)>0$.
	\end{proof}

\section{Oblivious Perturbation Model.} \label{sec:model}

In this section, we provide the full details of our oblivious perturbation model. Recall that the perturbation $R$ is defined as a sum of two independent random matrices a dense matrix $R_1$ and a sparse matrix $R_2$.

\subsection{The dense perturbation $R_1$.}

As mentioned in Section \ref{s:overview-R1},
    \begin{equation}\label{eq:r1defn}
        R_1=\frac{1}{\rho \sqrt{n}} \cdot D_1 V D_2 
    \end{equation}
for $D_1 = \diag(\eta_{1,1} \etc \eta_{1,n})$ and $D_2 = \diag(\eta_{2,1} \etc \eta_{2,n})$ where $\eta_{1,1} \etc \eta_{1,n}$ and $\eta_{2,1} \etc \eta_{2,n}$ are i.i.d. Rademacher random variables, and $V$ is any deterministic \textit{pattern} matrix with parameters $\alpha, \beta, \gamma$, and $\rho$ as defined in the following Definition \ref{def:pattern},

\begin{definition}\label{def:pattern}
    Let $0<\alpha, \beta, \gamma <1$ and $\rho>0$. An $n\times n$ matrix $V$ with $\pm 1$ entries will be called an $n$ dimensional pattern matrix with parameters $\alpha, \beta, \gamma,\rho$ if all of the following hold.
    \begin{enumerate}
        \item $\norm{V}\le \rho\sqrt{n}$,
        \item For any $\eta_1 \etc \eta_n \in \{ -1, 1 \}$, and for any $\alpha n$-sparse vector $x \in \bS^{n-1}$, we have
        \[ \abs{ \{ j \in [n] : \abs{e_j^T V \cdot \diag(\eta_{1} \etc \eta_{n}) x} \ge \beta \} } \ge \gamma n \]
        and
        \[ \abs{ \{ j \in [n] : \abs{e_j^T V^T \cdot \diag(\eta_{1} \etc \eta_{n})x} \ge \beta \} } \ge \gamma n \]
    \end{enumerate}
\end{definition}

 For discussion about the existence of pattern matrices, see Section \ref{sec:pattern}.

 \subsection{Necessity of a second perturbation.}

 Before describing the construction of the second matrix $R_2$, we note via the following lemma that setting $R= R_1$ alone does not suffice to obtain a condition number bound.

 \begin{lemma} \label{lem:pattinsuff}
	Let $\hat{V}$ be an $n\times n$ deterministic matrix with $\norm{\hat{V}}= 1$. Let $D_1$ and $D_2$ be as in \eqref{eq:r1defn}. Then, for $\varepsilon = 1/\sqrt{n}$, there exists a deterministic $n \times n$ matrix $A$ with $\norm{A}\le 1$ such that 
	\begin{equation*}
		\P \paren*{ s_n(A+ \varepsilon D_1 \hat{V} D_2) = 0 } \ge \frac{1}{2}
	\end{equation*}
 \end{lemma}

\begin{proof}
Write
\[
D_1=\diag(\eta_{1,1},\ldots,\eta_{1,n})
\qquad\text{and}\qquad
D_2=\diag(\eta_{2,1},\ldots,\eta_{2,n}),
\]
where the two families of signs are independent.

Choose \(u\in\bS^{n-1}\) orthogonal to
\(\col_2(\hat V),\ldots,\col_n(\hat V)\), and define the random vector
\[
v_{\eta_1}=D_1u.
\]
We observe that \(v_{\eta_1}\in\bS^{n-1}\), and, for every \(k\ge 2\), we have
\begin{align*}
\angbr*{\col_k(D_1\hat V D_2),v_{\eta_1}}
&=
\eta_{2,k}\angbr*{D_1\col_k(\hat V),D_1u} \\
&=
\eta_{2,k}\angbr*{\col_k(\hat V),u}
=0.
\end{align*}
For the first column, we similarly have
\[
\angbr*{\col_1(D_1\hat V D_2),v_{\eta_1}}
=
\eta_{2,1}\angbr*{\col_1(\hat V),u}.
\]
Therefore, we obtain
\[
v_{\eta_1}^T D_1\hat V D_2
=
\paren*{
\eta_{2,1}\angbr*{\col_1(\hat V),u},
0,\ldots,0
}.
\]

Choose \(j\in[n]\) such that \(\abs{u_j}\ge 1/\sqrt n\), and define
\[
A=
\frac{\angbr*{\col_1(\hat V),u}}
     {\sqrt n\,u_j}\,e_je_1^T.
\]
Then \(A\) is deterministic and
\[
\norm{A}
=
\frac{\abs{\angbr*{\col_1(\hat V),u}}}
     {\sqrt n\,\abs{u_j}}
\le
\norm{\col_1(\hat V)}
\le
\norm{\hat V}
=1.
\]
Furthermore, we observe
\[
v_{\eta_1}^T A
=
\paren*{
\frac{\eta_{1,j}}{\sqrt n}
\angbr*{\col_1(\hat V),u},
0,\ldots,0
}.
\]
Therefore, for \(\varepsilon=1/\sqrt n\), we have
\[
v_{\eta_1}^T
\paren*{A+\varepsilon D_1\hat V D_2}
=
\frac{\angbr*{\col_1(\hat V),u}}{\sqrt n}
\paren*{\eta_{1,j}+\eta_{2,1},0,\ldots,0}.
\]
On the event
\[
\eta_{1,j}+\eta_{2,1}=0,
\]
the nonzero vector \(v_{\eta_1}\) is a left null vector of
\(A+\varepsilon D_1\hat V D_2\), so this matrix is singular. Since
\(\eta_{1,j}\) and \(\eta_{2,1}\) are independent Rademacher random
variables, we conclude that
\begin{align*}
\P\paren*{s_n(A+\varepsilon D_1\hat V D_2)=0}
&\ge
\P\paren*{\eta_{1,j}+\eta_{2,1}=0} \\
&=\frac12.
\end{align*}
\end{proof}

 \subsection{The sparse perturbation $R_2$.}

We now explain the construction of $R_2$ once again with additional figures and conclude with a full definition of our perturbation model $R$ (Definition \ref{def:perturb}) for completeness. $R_2$ is a sparse random matrix defined with parameters $K$ and $L$ with $K<L$ by first specifying the positions of the nonzero entries by a hashing matrix $H$ and populating these positions by independent scaled Rademacher random variables. Before constructing $H$, we first construct an auxillary hashing matrix $\bar{H}$. The columns of the matrix $\bar{H}$ are independent $\{0,1\}$ random vectors containing
$K$ ones. More formally, for $i\in [n]$, let $J_i$ be independent subsets of $[n]$ of size $K$ uniformly distributed among such subsets. Then,  for $i,j \in [n]$, we set $\bar{h}_{j,i}= \mathds{1}_{J_i} (j) $ (the matrix $\bar H$ is illustrated in Figure \ref{fig:barH}).

\begin{figure}[h!]
  \centering

  \begin{tikzpicture}

		\matrix (Hbar) [
		matrix of nodes,
		nodes={draw, minimum size=8mm, font=\scriptsize, inner sep=0pt, anchor=center},
		row sep=-\pgflinewidth,
		column sep=-\pgflinewidth,
		] at (0,-0.55) {
			|[fill=blue!12]|\textcolor{blue}{1} & 0 & |[fill=blue!12]|\textcolor{blue}{1} & |[fill=blue!12]|\textcolor{blue}{1} & 0 & 0 & |[fill=blue!12]|\textcolor{blue}{1} & 0\\
			0 & 0 & 0 & 0 & |[fill=blue!12]|\textcolor{blue}{1} & 0 & 0 & |[fill=blue!12]|\textcolor{blue}{1}\\
			0 & 0 & |[fill=blue!12]|\textcolor{blue}{1} & 0 & 0 & |[fill=blue!12]|\textcolor{blue}{1} & 0 & 0\\
			0 & |[fill=blue!12]|\textcolor{blue}{1} & 0 & 0 & 0 & 0 & |[fill=blue!12]|\textcolor{blue}{1} & 0\\
			|[fill=blue!12]|\textcolor{blue}{1} & 0 & 0 & 0 & |[fill=blue!12]|\textcolor{blue}{1} & 0 & 0 & 0\\
			0 & |[fill=blue!12]|\textcolor{blue}{1} & 0 & 0 & 0 & 0 & 0 & 0\\
			0 & 0 & 0 & 0 & 0 & |[fill=blue!12]|\textcolor{blue}{1} & 0 & 0\\
			0 & 0 & 0 & |[fill=blue!12]|\textcolor{blue}{1} & 0 & 0 & 0 & |[fill=blue!12]|\textcolor{blue}{1}\\
		};

		\draw[thick] (Hbar-1-1.north west) rectangle (Hbar-8-8.south east);

		\foreach \c in {1,...,8} {
			\node[font=\scriptsize] at ($(Hbar-1-\c.north)+(0,3.2mm)$) {$i=\c$};
		}
		\foreach \r in {1,...,8} {
			\node[font=\scriptsize] at ($(Hbar-\r-1.west)+(-6mm,0)$) {$j=\r$};
		}

		\node[font=\normalsize] at ($(Hbar-1-8.north east)+(-4, +1)$) {$\bar H$};

		\node[align=left,font=\scriptsize] (Jlist) at (6.2,0.2) {$J_1=\{1,5\}$\\
			$J_2=\{4,6\}$\\
			$J_3=\{1,3\}$\\
			$J_4=\{1,8\}$\\
			$J_5=\{2,5\}$\\
			$J_6=\{3,7\}$\\
			$J_7=\{1,4\}$\\
			$J_8=\{2,8\}$};
		
		\node[align=left,font=\scriptsize] at (6.2,1.85) {Each $|J_i|=K=2$.};

		\draw[-{Stealth[length=2.2mm]},thick] (Jlist.west) -- ($(Hbar.east)+(0.25,0)$);
	\end{tikzpicture}

  \caption{{Illustration of $\bar H$ with} $n=8$, $K=2$, $L=3$.\\
			Each column $i$ selects an independent uniform subset $J_i\subset[n]$ of size $K$ and sets $\bar h_{j,i}=\mathbf{1}_{J_i}(j)$.}
  \label{fig:barH}
\end{figure}

Next, define a sequence $\{X_{l,i} : i \in [n], l \in [K]\}$ of $Kn$ i.i.d. Rademacher random variables, and enumerate the elements of each
random set $J_{i}, i\in [n]$:
\[ J_i = \{j(i,1),...,j(i,K)\} \text{ where } j(i,1) <j(i,2) <\ldots<j(i,K) \]

Define the matrix $\bar{R}_2$ by
\begin{equation*}
    {[\bar{R}_2]}_{j(i,l),i} = \begin{cases}
        \frac{1}{L} \cdot X_{l,i}, &\text{ for } l\in [K] \\
        0 & \text{ otherwise}
    \end{cases}
\end{equation*}
In other words, we use
${\bar{H}}$ as a hashing matrix to create
$\bar{R_2}$, which is illustrated in Figure \ref{fig:barR2}.
\begin{figure}[h!]
  \centering

  \begin{tikzpicture}
		
		\matrix (Rbar) [
		matrix of nodes,
		nodes={draw, minimum width=8mm, minimum height=8mm, font=\tiny, inner sep=1pt, align=center},
		row sep=-\pgflinewidth,
		column sep=-\pgflinewidth,
		] at (0,-0.55) {
			|[fill=green!10]|$\frac{1}{L}X_{1,1}$ & \textcolor{gray}{0} & |[fill=green!10]|$\frac{1}{L}X_{1,3}$ & |[fill=green!10]|$\frac{1}{L}X_{1,4}$ & \textcolor{gray}{0} & \textcolor{gray}{0} & |[fill=green!10]|$\frac{1}{L}X_{1,7}$ & \textcolor{gray}{0}\\
			\textcolor{gray}{0} & \textcolor{gray}{0} & \textcolor{gray}{0} & \textcolor{gray}{0} & |[fill=green!10]|$\frac{1}{L}X_{1,5}$ & \textcolor{gray}{0} & \textcolor{gray}{0} & |[fill=green!10]|$\frac{1}{L}X_{1,8}$\\
			\textcolor{gray}{0} & \textcolor{gray}{0} & |[fill=green!10]|$\frac{1}{L}X_{2,3}$ & \textcolor{gray}{0} & \textcolor{gray}{0} & |[fill=green!10]|$\frac{1}{L}X_{1,6}$ & \textcolor{gray}{0} & \textcolor{gray}{0}\\
			\textcolor{gray}{0} & |[fill=green!10]|$\frac{1}{L}X_{1,2}$ & \textcolor{gray}{0} & \textcolor{gray}{0} & \textcolor{gray}{0} & \textcolor{gray}{0} & |[fill=green!10]|$\frac{1}{L}X_{2,7}$ & \textcolor{gray}{0}\\
			|[fill=green!10]|$\frac{1}{L}X_{2,1}$ & \textcolor{gray}{0} & \textcolor{gray}{0} & \textcolor{gray}{0} & |[fill=green!10]|$\frac{1}{L}X_{2,5}$ & \textcolor{gray}{0} & \textcolor{gray}{0} & \textcolor{gray}{0}\\
			\textcolor{gray}{0} & |[fill=green!10]|$\frac{1}{L}X_{2,2}$ & \textcolor{gray}{0} & \textcolor{gray}{0} & \textcolor{gray}{0} & \textcolor{gray}{0} & \textcolor{gray}{0} & \textcolor{gray}{0}\\
			\textcolor{gray}{0} & \textcolor{gray}{0} & \textcolor{gray}{0} & \textcolor{gray}{0} & \textcolor{gray}{0} & |[fill=green!10]|$\frac{1}{L}X_{2,6}$ & \textcolor{gray}{0} & \textcolor{gray}{0}\\
			\textcolor{gray}{0} & \textcolor{gray}{0} & \textcolor{gray}{0} & |[fill=green!10]|$\frac{1}{L}X_{2,4}$ & \textcolor{gray}{0} & \textcolor{gray}{0} & \textcolor{gray}{0} & |[fill=green!10]|$\frac{1}{L}X_{2,8}$\\
		};
		
		\draw[thick] (Rbar-1-1.north west) rectangle (Rbar-8-8.south east);

		\foreach \c in {1,...,8} {
			\node[font=\scriptsize] at ($(Rbar-1-\c.north)+(0,3.2mm)$) {$i=\c$};
		}
		\foreach \r in {1,...,8} {
			\node[font=\scriptsize] at ($(Rbar-\r-1.west)+(-6mm,0)$) {$j=\r$};
		}
		
		\node[font=\normalsize] at ($(Rbar-1-8.north east)+(-4,+1)$) {$\bar R_2$};

		\node[align=left,font=\scriptsize,anchor=west] at ($(Rbar.east)+(0.8,1.2)$) {$\ell=1$ is the first selected row in $J_i$,\\
			$\ell=2$ is the second selected row in $J_i$.};
	\end{tikzpicture}
    \caption{{Illustration of $\bar R_2$ with} $n=8$, $K=2$, $L=3$.\\
			Here, we enumerate $J_i=\{j(i,1)<\cdots<j(i,K)\}$ and set
			$\bar r_{j(i,\ell),i}=\frac{1}{L}X_{\ell,i}$ for $\ell\in[K]$, else $0$.}
  \label{fig:barR2}
\end{figure}
As mentioned before in Section~\ref{sec:outline}, this construction $\bar{R_2}$ is not sufficient for our purpose because $\norm{\bar{R_2}}$ may be large.

To reduce the norm of the random perturbation, we remove the heavy rows of $\bar H$ and $\bar R_2$ to construct the matrices $H$ and $R_2$.  Formally, we set $\text{Row}_j(H) = \text{Row}_j(\bar{H})$ if the number of non-zero entries in $\text{Row}_j(\bar{H})$ does not exceed $L$ and $\text{Row}_j(H) = 0$ otherwise. Note that the number
of non-zero entries in each row and column of the matrix $H$ is at most $L$. The process of transforming $\bar H$ to $H$ is illustrated in Figure \ref{fig:H}.
\begin{figure}[h!]
  \centering

  	\begin{tikzpicture}

		\matrix (Hbar2) [
		matrix of nodes,
		nodes={draw, minimum size=5.2mm, font=\scriptsize, inner sep=0pt, anchor=center},
		row sep=-\pgflinewidth,
		column sep=-\pgflinewidth,
		] at (-4.1,-0.55) {
			|[fill=blue!12]|\textcolor{blue}{1} & 0 & |[fill=blue!12]|\textcolor{blue}{1} & |[fill=blue!12]|\textcolor{blue}{1} & 0 & 0 & |[fill=blue!12]|\textcolor{blue}{1} & 0\\
			0 & 0 & 0 & 0 & |[fill=blue!12]|\textcolor{blue}{1} & 0 & 0 & |[fill=blue!12]|\textcolor{blue}{1}\\
			0 & 0 & |[fill=blue!12]|\textcolor{blue}{1} & 0 & 0 & |[fill=blue!12]|\textcolor{blue}{1} & 0 & 0\\
			0 & |[fill=blue!12]|\textcolor{blue}{1} & 0 & 0 & 0 & 0 & |[fill=blue!12]|\textcolor{blue}{1} & 0\\
			|[fill=blue!12]|\textcolor{blue}{1} & 0 & 0 & 0 & |[fill=blue!12]|\textcolor{blue}{1} & 0 & 0 & 0\\
			0 & |[fill=blue!12]|\textcolor{blue}{1} & 0 & 0 & 0 & 0 & 0 & 0\\
			0 & 0 & 0 & 0 & 0 & |[fill=blue!12]|\textcolor{blue}{1} & 0 & 0\\
			0 & 0 & 0 & |[fill=blue!12]|\textcolor{blue}{1} & 0 & 0 & 0 & |[fill=blue!12]|\textcolor{blue}{1}\\
		};
		\draw[thick] (Hbar2-1-1.north west) rectangle (Hbar2-8-8.south east);
		\node[font=\normalsize] at ($(Hbar2-1-8.north east)+(-2,+1)$) {$\bar H$};

		\foreach \c in {1,...,8} {\node[font=\scriptsize] at ($(Hbar2-1-\c.north)+(0,3.2mm)$) {$\c$};}
		\foreach \r in {1,...,8} {\node[font=\scriptsize] at ($(Hbar2-\r-1.west)+(-3.2mm,0)$) {$\r$};}

		\draw[red,very thick] (Hbar2-1-1.west) -- (Hbar2-1-8.east);
		\node[red,font=\scriptsize,align=left] at ($(Hbar2-1-8.east)+(2,0.35)$) {$\|\mathrm{Row}_1(\bar H)\|_0=4>L$\\(deleted)};

		\matrix (H) [
		matrix of nodes,
		nodes={draw, minimum size=5.2mm, font=\scriptsize, inner sep=0pt, anchor=center},
		row sep=-\pgflinewidth,
		column sep=-\pgflinewidth,
		] at (4.1,-0.55) {
			|[fill=gray!15]|0 & |[fill=gray!15]|0 & |[fill=gray!15]|0 & |[fill=gray!15]|0 & |[fill=gray!15]|0 & |[fill=gray!15]|0 & |[fill=gray!15]|0 & |[fill=gray!15]|0\\
			0 & 0 & 0 & 0 & |[fill=blue!12]|\textcolor{blue}{1} & 0 & 0 & |[fill=blue!12]|\textcolor{blue}{1}\\
			0 & 0 & |[fill=blue!12]|\textcolor{blue}{1} & 0 & 0 & |[fill=blue!12]|\textcolor{blue}{1} & 0 & 0\\
			0 & |[fill=blue!12]|\textcolor{blue}{1} & 0 & 0 & 0 & 0 & |[fill=blue!12]|\textcolor{blue}{1} & 0\\
			|[fill=blue!12]|\textcolor{blue}{1} & 0 & 0 & 0 & |[fill=blue!12]|\textcolor{blue}{1} & 0 & 0 & 0\\
			0 & |[fill=blue!12]|\textcolor{blue}{1} & 0 & 0 & 0 & 0 & 0 & 0\\
			0 & 0 & 0 & 0 & 0 & |[fill=blue!12]|\textcolor{blue}{1} & 0 & 0\\
			0 & 0 & 0 & |[fill=blue!12]|\textcolor{blue}{1} & 0 & 0 & 0 & |[fill=blue!12]|\textcolor{blue}{1}\\
		};
		\draw[thick] (H-1-1.north west) rectangle (H-8-8.south east);
		\node[font=\normalsize] at ($(H-1-8.north east)+(-2,+1)$) {$H$};

		\foreach \c in {1,...,8} {\node[font=\scriptsize] at ($(H-1-\c.north)+(0,3.2mm)$) {$\c$};}
		\foreach \r in {1,...,8} {\node[font=\scriptsize] at ($(H-\r-1.west)+(-3.2mm,0)$) {$\r$};}

		\draw[-{Stealth[length=2.5mm]},thick] ($(Hbar2.east)+(0.4,0)$) -- node[above,font=\scriptsize] {row trimming} ($(H.west)+(-0.4,0)$);
	\end{tikzpicture}

  \caption{{Illustration of constructing $ H$ from $\bar H$ with} $n=8$, $K=2$, $L=3$.\\
			Here, we keep row $j$ if $\|\mathrm{Row}_j(\bar H)\|_0\le L$, else set it to $0$.}
  \label{fig:H}
\end{figure}

Finally, set $R_2 = H \odot \bar{R_2}$ where $\odot$ stands for the entry-wise (Hadamard) product, which was illustrated earlier in Figure \ref{fig:R2} in Section~\ref{sec:outline}. Formally, we have the following definition.

\begin{definition} \label{def:perturb}
    Let $V$ be an $n$-dimensional deterministic pattern matrix with parameters $\alpha, \beta, \gamma,\rho$. Let $\eta_{1,1} \etc \eta_{1,n}$ and $\eta_{2,1} \etc \eta_{2,n}$ be i.i.d. Rademacher random variables. Let
    \begin{align*}
        R_1=\frac{1}{\rho\sqrt{n}} \cdot D_1 V D_2
    \end{align*}
    where $D_1 = \diag(\eta_{1,1} \etc \eta_{1,n})$ and $D_2 = \diag(\eta_{2,1} \etc \eta_{2,n})$. Let $K$ and $L$ be positive integers with $1\le K\le n$ and $L>K$. For $i\in [n]$, let $J_i$ be independent subsets of $[n]$ of size $K$ uniformly distributed among such subsets. Let $\bar H$ be the random matrix with entries $\bar{h}_{j,i}$  such that, for $i,j \in [n]$, we have
\[ \bar{h}_{j,i}= \mathds{1}_{J_i} (j) \]
We define the matrix $H$ by setting $\text{Row}_j(H) = \text{Row}_j(\bar{H})$ if the number of non-zero entries in $\text{Row}_j(\bar{H})$ does not exceed $L$ and setting $\text{Row}_j(H) = 0$ otherwise. Let $\{X_{l,i} : i \in [n], l \in [K]\}$ be i.i.d. Rademacher random variables, and enumerate the elements of each
random set $J_{i}, i\in [n]$:
\[ J_i = \{j(i,1),...,j(i,K)\} \text{ where } j(i,1) <j(i,2) <\ldots<j(i,K) \]
The signs $\eta_{1,j}$, $\eta_{2,j}$, $X_{l,i}$, and the subsets
$J_i$ are mutually independent.
Let $\bar{R}_2$ be the random matrix with entries
\begin{equation*}
    [\bar{R}_2]_{j(i,l),i} = \begin{cases}
        \frac{1}{L} \cdot X_{l,i}, &\text{ for } l\in [K] \\
        0 & \text{ otherwise}
    \end{cases}
\end{equation*}
Let $R_2 = H \odot \bar{R_2}$. Then $(R_1,\{J_i\}_{i \in [n]},\overline{H},H,\{X_{l,i}\}_{i \in [n],l \in[K]},\overline{R}_2,R_2)$ is called a $n$-dimensional oblivious perturbation model with pattern matrix $V$ and parameters $\rho$, $K$, and $L$.
\end{definition}

The random perturbation constructed in the above manner has an $O(1)$ norm bound.

\begin{remark}[Norm of the perturbation]\label{rem:norm}
    Let $0<\alpha, \beta, \gamma <1$ and $\rho>0$. Let $V$ be an $n$ dimensional pattern matrix with parameters $\alpha, \beta, \gamma,\rho$. Let $K$ and $L$ with $L>K$ be positive integers. Let $(R_1,\{J_i\}_{i \in [n]},\overline{H},H,\{X_{l,i}\}_{i \in [n],l \in[K]},\overline{R}_2,R_2)$ be a $n$-dimensional oblivious perturbation model with pattern matrix $V$ and parameters $\rho$, $K$, and $L$. Then, $\norm{R_1 + R_2} \le 2$.
\end{remark}

\begin{proof}
    Note that the property $\norm{V}\le \rho\sqrt{n}$ leads to $\norm{R_1} \le 1$, and by the Schur test, we have
\[ \norm{R_2} \le \max_{j\in [n]} \cbr*{ \norm{\row_j(R_2)}_1 , \norm{\col_j(R_2)}_1 } \le 1 \]
Thus, in conclusion, we have $\norm{R_1+R_2} \le \norm{R_1} +  \norm{R_2} \le 2$.
\end{proof}

\section{Pattern Matrices.} \label{sec:pattern}

In this section, we first give an example of a $\sqrt{n}$-sparse vector whose Walsh-Hadamard transform is also $\sqrt{n}$-sparse and then provide details of the randomized construction of pattern matrices mentioned in Section \ref{s:overview-R1}.

\begin{proposition} \label{prop:hadamardctrex}
    Let $n=2^{2k}$ for $k>0$, and let $H=H_n$ be the $n$-dimensional Walsh Hadamard matrix. Then, there exists $x \in  \R^n$ with $\supp(x) = 2^k = \sqrt{n}$ such that $\supp(Hx) = 2^k = \sqrt{n}$.
\end{proposition}
\begin{proof}
    Let $n = 2^{2k}$ for $k>0$. Consider an $n \times n$ Walsh Hadamard matrix $H$ and index its rows and columns with tuples from $\F_2^{2k} \cong \F_2^{k} \times \F_2^{k} := \cbr*{ (a_1 \etc a_{k}, b_1 \etc b_{k}) \vert a_i, b_i \in \cbr*{0,1} }$. That is to say, the $(i,j)\textsuperscript{th}$ entry of $H$ is referenced as the $((i_a, i_b), (j_a, j_b))\textsuperscript{th}$ entry of $H$, where
    \begin{align*}
       (i_a, i_b) &= (i_{a_1} \etc i_{a_{k}}, i_{b_1} \etc i_{b_{k}}) \in \{0,1\}^{2k} \\
       (j_a, j_b) &= (j_{a_1} \etc j_{a_{k}}, j_{b_1} \etc j_{b_{k}}) \in \{0,1\}^{2k}
    \end{align*}
    are the binary representations of $i$ and $j$ in $2k$ bits. One can then check that $H_{(i,j)} = H_{((i_a, i_b), (j_a, j_b))} = (-1)^{i_a \cdot j_a + i_b \cdot j_b}$ where $\cdot$ represents the usual component wise inner product of tuples. 

    Now consider the vector $x = (x_1 \etc x_n) \in \R^n$ given by
    \begin{align*}
        x_j = x_{(j_a, j_b)} = \begin{cases}
            1 \quad \text{if } j_a = (0 \etc 0) \\
            0 \quad \text{otherwise}
        \end{cases}
    \end{align*}
In other words,    $x$ is only supported on $2^k$ coordinates, so is $ \sqrt{n}$-sparse, and it also satisfies
    \begin{align*}
        (Hx)_i &= (Hx)_{(i_a, i_b)} \\
        &= \sum_{j_a, j_b} (-1)^{i_a \cdot j_a + i_b \cdot j_b} \ind_{ j_a = (0 \etc 0)} \\
        &= \sum_{j_b} (-1)^{i_b \cdot j_b}
    \end{align*}
    Observe that the sum is non-zero if and only if $i_b = (0 \etc 0)$. So, $Hx$ is supported on $2^k$ coordinates as well.
\end{proof}

The following proposition proves that the model for pattern matrices described in Section \ref{s:overview-R1} results in matrices that satisfy the relevant properties (as in Definition \ref{def:pattern}) with high probability.

\begin{proposition} \label{prop:pattern}
Let $V_1$ be a matrix with Rademacher entries that are $2q$-wise independent where $q=\lceil5\log n\rceil$. Let $W$ be a vector with 4-wise independent Rademacher coordinates. Let $V_2$ and $V_3$ be
a random matrices whose rows are independent copies of $W$. Let 
\[ V = V_1 \odot V_2 \odot V_3^T\]
where we recall that $\odot$ stands for the entry-wise (Hadamard) product and $V_1, V_2, V_3$ are independent of each other. There exists constants $c_{\ref{prop:pattern}.1},c_{\ref{prop:pattern}.2},c_{\ref{prop:pattern}.3},c_{\ref{prop:pattern}.4}, c_{\ref{prop:pattern}.5}>0$ with $c_{\ref{prop:pattern}.3}, c_{\ref{prop:pattern}.5}<1$, such that for any $n \ge c_{\ref{prop:pattern}.1}$, with probability at least $1-1/n^2$, $V$ satisfies the following properties:
    \begin{enumerate}
        \item $\norm{V}\le c_{\ref{prop:pattern}.2}\sqrt{n}$,
        \item For any $\eta_1 \etc \eta_n \in \{ -1, 1 \}$, and for any $\alpha n$-sparse vector $x \in \bS^{n-1}$ with $\alpha \le c_{\ref{prop:pattern}.3}$, we have 
        \[ \abs{ \{ j \in [n] : \abs{e_j^T V\diag(\eta_{1} \etc \eta_{n})x} \ge c_{\ref{prop:pattern}.4} \} } \ge c_{\ref{prop:pattern}.5} n \]
        and
        \[ \abs{ \{ j \in [n] : \abs{e_j^T V^T\diag(\eta_{1} \etc \eta_{n})x} \ge c_{\ref{prop:pattern}.4} \} } \ge c_{\ref{prop:pattern}.5} n \]
    \end{enumerate}
\end{proposition}

\begin{proof}
We first look at the norm. Conditioned on $V_2$ and $V_3$, $V$ is
a random matrix with Rademacher entries that are $2q$-wise
independent. Since $q=\lceil5\log n\rceil$, we can apply
in \cite[Lemma 2]{LovettSodin2008} and obtain
\begin{align*}
    \P \paren*{ \norm{V} \ge 3\sqrt{n} \vert V_2,V_3}
    &\le 2n\paren*{\frac{3}{2}}^{-2q} \\
    &\le 2n^{1-10\log(3/2)} \\
    &\le \frac{2}{n^3}.
\end{align*}
Taking expectation over $V_2$ and $V_3$, we obtain
\[
    \P \paren*{ \norm{V} \ge 3\sqrt{n}} \le \frac{2}{n^3}.
\]

Next, we consider the number of ''good'' rows. Fix $x \in \bS^{n-1}$, $\eta_1 \etc \eta_n \in \{ -1, 1 \}$ and $j \in [n]$, and let $X_j = X_j(\eta_1 \etc \eta_n, \beta, x)$ be the indicator random variable for the event that $\abs{e_j^T VD_{\eta}x}  \le  \beta$, for some $\beta>0$ where $D_{\eta} = \diag(\eta_{1} \etc \eta_{n})$. Conditioned on $V_1$ and $V_3$, the rows of $V$ are independent
and each row has 4-wise independent Rademacher coordinates. Let $Z_j = \paren*{e_j^T VD_{\eta}x}^2$. Then, $\E Z_j = 1$. By Paley-Zygmund inequality, we have
\[ \Pb_{V_2} \paren*{Z_j > \theta} \ge \frac{(1-\theta)^2}{\E Z_j^2} \]
Since the expression for $Z_j^2$ only has products of order 4 in the entries of $V$, and the coordinates in each row of $V$ are 4-wise independent, $\E_{V_2} Z_j^2$ remains unchanged if we assume that the rows of $V$ have independent $\pm 1$ entries. In this case, by using Hoeffding's inequality and sub-gaussianity, we have that $1 \le \E Z_j^2 \le c_3$. Therefore, there exist $c_4$ and $c_5$ such that $\Pb_{V_2} \paren*{\abs{e_j^T VD_{\eta}x}\le c_4} = \E_{V_2} X_j(\eta_1 \etc \eta_n, c_4, x) \le c_5 < 1$. By Hoeffding's inequality, 
we can find $c_6<1$ and $c_7>0$ so that
\begin{align*}
    \Pb_{V_2} \paren*{ \sum_{j=1}^n X_j(\eta_1 \etc \eta_n, c_4, x) \ge c_6 n} &\le \exp (-c_7 n)\end{align*} which implies \begin{align*}
    \Pb \paren*{ | \{j \in [n] : |e_j^TVD_{\eta}x|\le c_4 \}|\ge c_6 n } &\le \exp(-c_7n)
\end{align*}
where we integrate over $V_1$ and $V_3$.

If we further assume that $x$ is $\alpha n$-sparse, then $e_j^TVD_{\eta}x$ depends on only those $\eta_i$ for which $i$ is in the support of $x$. Thus, using the union bound, we obtain
\begin{align*}
    \Pb \paren*{ \exists \eta \in \{-1,1 \}^n : \vert \{j \in [n] : |e_j^TVD_{\eta}x|\le c_4 \}|\ge c_6 n } &\le 2^{\alpha n} \exp(-c_7n) \\
    &\le \exp(-c_7n/2)
\end{align*}
when $\alpha < c_7/(2\log 2)$. 

Let $\cN$ be a $\nu$-net of $\alpha n$-sparse unit vectors for the set of all $\alpha n$-sparse unit vectors. This means that for every $\alpha n$-sparse $y \in \bS^{n-1}$, there exists an $\alpha n$-sparse unit $x \in \cN$ such that $\norm{x-y}_2 \le \nu$.

When $\log \abs{\cN} \le c_7n/4$ (we shall later show that this is the case when $\alpha$ is small enough), by a union bound over $\cN$, we get
\begin{align*}
    \Pb \paren*{ \exists x \in \cN, \eta \in \{-1,1 \}^n : \lvert \{j \in [n] : |e_j^TVD_{\eta}x|\le c_4 \}\rvert \ge c_6 n } &\le \exp(-c_7n/4)
\end{align*}
Define the event $\cE := \{ \lvert \{j \in [n] : |e_j^TVD_{\eta}x|\ge c_4 \}\rvert \ge (1-c_6)n \quad \forall x \in \cN, \eta \in \{-1,1 \}^n  \}$. We have shown that $\Pb(\cE) \ge 1-\exp(-c_7n/4)$.

Now let $y$ be an arbitrary $\alpha n$-sparse unit vector. Then there is $x \in \cN$ such that $\norm{x-y} \le \nu$. Conditioning on the event that $\norm{V} \le c_2 \sqrt{n}$, we have $\norm{VD_{\eta}(x-y)} \le c_2 \nu \sqrt{n}$ for all $\eta$. Then, 
\[ \lvert \{j \in [n] : |e_j^TVD_{\eta}(x-y)| \ge \lambda c_2 \nu  \}\rvert \le \frac{n}{\lambda^2} \]

Let $\lambda$ be such that $1/\lambda^2 = (1-c_6)/2$. Once $\lambda$ has been chosen, let $\nu$ be chosen such that $\lambda c_2 \nu = c_4/2$. Then, conditioned on $\cE$ and on the event that $\norm{V} \le c_2 \sqrt{n}$, we have, for all $\eta \in \{-1,1 \}^n$,
\begin{align*}
    \lvert \{j \in [n] : |e_j^TVD_{\eta}x|\ge c_4 \text{ and } |e_j^TVD_{\eta}(x-y)| \le c_4/2  \}\rvert \ge (1-c_6)n/2 \\
    \implies \lvert \{j \in [n] : |e_j^TVD_{\eta}y|\ge c_4/2  \}\rvert \ge (1-c_6)n/2
\end{align*}

We now show that $\log \abs{\cN} \le (c_7/4)n$ when $\alpha$ is small enough. From \cite[Corollary 4.2.11]{vershynin2018high} we have,
\begin{align*}
      \abs{\cN} \le& \binom{n}{\lfloor\alpha n\rfloor}
    \paren*{\frac{3}{\nu}}^{\lfloor\alpha n\rfloor} 
\\\le& \paren*{\frac{3e}{\nu\alpha}}^{\alpha n}
\end{align*}

When $\alpha < \min \{ \nu, 1/(3e) \}$, by standard calculations, we obtain
\begin{align*}
    \log \abs{\cN} &\le \alpha n \log \paren*{\frac{3e}{\nu\alpha}} \\
    &\le \alpha n \log(3e) + \alpha n \log \paren*{\frac{1}{\nu}} + \alpha n \log \paren*{\frac{1}{\alpha}} \\
    &\le 3 \alpha n \log \paren*{\frac{1}{\alpha}}
\end{align*}
We further let $\alpha$ be small enough that $3 \alpha  \log \paren*{\frac{1}{\alpha}} \le c_7/4$.

We can repeat the same arguments by looking at $V^T$ with $V_3$ in place of $V_2$, to get the third property. Finally, by union bound, the probability that any
of the required properties fails is at most
\[
    \frac{2}{n^3}+2\exp(-c_7n/4)\le\frac{1}{n^2}
\]
for all sufficiently large $n$.

\end{proof}

\begin{corollary}\label{cor:randomnessV}
    For any $n \ge c_{\ref{prop:pattern}.1}$, $\rho \ge c_{\ref{prop:pattern}.2}$, $0<\alpha \le c_{\ref{prop:pattern}.3}$, $0<\beta \le c_{\ref{prop:pattern}.4}$, and $0<\gamma \le c_{\ref{prop:pattern}.5}$, there exists an $n$-dimensional pattern matrix $V$ with parameters $\alpha, \beta, \gamma,\rho$. In addition, there exists an algorithm using $O(n\log n)$ random bits that can generate such a pattern matrix with success probability at least $1-1/n^2$.
\end{corollary}

\begin{proof}
    We only need to justify the claim about random bit complexity.
    Considering the construction of $V$, we need to generate $n^2$
    many $2q$-wise independent signs, and also $2n$ independent families where each family contains $n$ many 4-wise independent signs.

    By standard results in random bits complexity (e.g., \cite[Theorem 5.11.3]{cover2006elements}, \cite[Corollary 3.34]{vadhan2012pseudorandomness}), we can
    choose the entries of $V_1$ and the distribution of $W$ so that
    the number of random bits needed is at most
    \[
        2q\lceil\log_2(n^2)\rceil
        +8n\lceil\log_2 n\rceil
        =O(n\log n).
    \]
    We use independent seeds for $V_1$ and for all the rows of
    $V_2$ and $V_3$. The probability estimate follows from
    Proposition \ref{prop:pattern}.
\end{proof}

\section{Invertibility for compressible vectors.} \label{sec:comp}

This section is devoted to proving the following proposition, which establishes a lower bound for the infimum of $\norm{Mx}_2$ over compressible vectors. 

\begin{theorem} \label{prop:comp}
    Let $A$ be an $n \times n$ deterministic matrix such that $\norm{A}= 1$. Let $0<\e<1$. Let $0<\alpha, \beta, \gamma <1$ and $\rho \ge 1$. Let $V$ be an $n$ dimensional pattern matrix with parameters $\alpha, \beta, \gamma,\rho$. Let $K$ and $L$ with $L>K$ be positive integers. Let $(R_1,\{J_i\}_{i \in [n]},\overline{H},H,\{X_{l,i}\}_{i \in [n],l \in[K]},\overline{R}_2,R_2)$ be a $n$-dimensional oblivious perturbation model with pattern matrix $V$ and parameters $\rho$, $K$, and $L$. Let $M= A+ \varepsilon (R_1+R_2)$. Then there exist constants $c_{\ref{prop:comp}.1}>0$ and $0<c_{\ref{prop:comp}.2}, c_{\ref{prop:comp}.3}, c_{\ref{prop:comp}.4} <1$ such that, for $\alpha <
 \frac{\gamma}{
 c_{\ref{prop:comp}.1}
 \log\left(
 \frac{c_{\ref{prop:comp}.1}\rho}{\beta\gamma^{3/2}}
 \right)}$ and $\nu = c_{\ref{prop:comp}.2} \varepsilon \frac{\beta\sqrt\gamma}{\rho}$, we have 
    \[ \Pb \paren*{ \inf_{v \in \comp(\alpha, \nu)} \norm{Mv}_2  \ge { c_{\ref{prop:comp}.3} \varepsilon \frac{\beta\sqrt\gamma}{\rho} } } \ge 1-\exp \paren*{-c_{\ref{prop:comp}.4} \gamma n} \]
\end{theorem}

\begin{proof}

Fix $x\in \bS^{n-1}$ with $|\supp(x)|\le\alpha n$.  Condition on $D_{2}$ and $R_2$ and let $A_1 = A + \varepsilon R_2$. Then, $Mx= A_1x + \frac{\varepsilon}{\rho \sqrt{n}} \cdot D_1 V D_{2}x $. Let $J(x) = \{j \in [n] : |\row_j(V)D_{2}x|\ge \beta\}$. By the definition of a pattern matrix, we have $\abs{J(x)} \ge \gamma n$. Then, for $j \in J(x)$, we have

\[ \Pb \paren*{ |(Mx)_j| \ge \frac{\varepsilon \beta}{\rho\sqrt{n}} \bigg \vert R_2, D_{2} } = \Pb \paren*{ \abs*{(A_1x)_j + \frac{\varepsilon}{\rho\sqrt{n}} \cdot \eta_{1,j} \row_j(V) D_2x } \ge \frac{\varepsilon \beta}{\rho\sqrt{n}} \bigg\vert R_2, D_{2}} \ge 1/2 \]
To see why observe that, since \(j\in J(x)\), we have
\[
 \left|
 \frac{\varepsilon}{\rho\sqrt n}\row_j(V)D_2x
 \right|
 \ge
 \frac{\varepsilon\beta}{\rho\sqrt n},
\]
and therefore the two numbers
\[
 (A_1x)_j
 +
 \frac{\varepsilon}{\rho\sqrt n}\row_j(V)D_2x,
 \qquad
 (A_1x)_j
 -
 \frac{\varepsilon}{\rho\sqrt n}\row_j(V)D_2x
\]
cannot both have absolute value strictly smaller than
\(\varepsilon\beta/(\rho\sqrt n)\).

Since the coordinates of $Mx$ are independent after conditioning on  $D_{2}$ and $R_2$, we can use the following tensorization lemma,

\begin{lemma}[\cite{rudelson2008littlewood}, Lemma 2.2 (2)]\label{lem:tensorization}
    Let $X_1 \etc X_N$ be independent random variables satisfying $$\P \paren*{\abs{X_i} \ge 1} \ge c_{\ref{lem:tensorization}.1}$$ for some $c_{\ref{lem:tensorization}.1} \in (0,1)$. Then, there exists constants $c_{\ref{lem:tensorization}.2}, c_{\ref{lem:tensorization}.3} > 0$, depending on $c_{\ref{lem:tensorization}.1}$, so that for the random vector $X := (X_1 \etc X_N)$, we have, 
    \[ \P \paren*{ \norm{X}_2 \le c_{\ref{lem:tensorization}.2}\sqrt{N} } \le e^{-c_{\ref{lem:tensorization}.3}N} \]
\end{lemma}

Applying Lemma \ref{lem:tensorization} with $X_j = \frac{\rho\sqrt{n}}{\varepsilon \beta} (Mx)_j$ for $j \in J(x)$, and observing that \[\norm{Mx} \ge \sqrt{\sum_{j \in J(x)} (Mx)_j^2},\] we get
\begin{align*}
\Pb_{D_1} \paren*{ \norm{Mx}_2 \le c_{\ref{lem:tensorization}.2} \frac{\beta \varepsilon }{\rho\sqrt{n}} \sqrt{\abs{J(x)}} } &\le e^{-c_{\ref{lem:tensorization}.3}\abs{J(x)}}
\end{align*}
or equivalently we can write
\begin{align*}\Pb_{D_1} \paren*{ \norm{Mx}_2 \le {c_1\varepsilon \frac{\beta\sqrt\gamma}{\rho}  } } &\le e^{-c_2 \gamma n}
\end{align*}
for absolute constants $c_1$ and $c_2$, using the fact that $\abs{J(x)} \ge \gamma n$.

Define
\begin{equation}\label{eq:GA1-definition}
 G_{A_1}=\left\{v\in\R^n:
 \norm v_2\le1
 \text{ and }
 \norm{A_1v}_2\le2\varepsilon
 \right\}.
\end{equation}
For each $I\subset[n]$, the set $G_{A_1}\cap E_I$ is convex and symmetric
about the origin in the coordinate subspace $E_I$.

Choose
\begin{equation}\label{eq:eta-definition}
 \eta=c_1\frac{\beta\sqrt\gamma}{96\rho}.
\end{equation}
Since $0<c_1,\beta,\gamma<1$ and $\rho>1$, we have $0<\eta<1$.

Let $G_I=G_{A_1}\cap E_I$.  By the standard volumetric estimate
as in Proposition \ref{prop:vershynin-4212}, applied in the subspace $E_I$,
the set $G_I$ can be covered by at most
\[
 \left(1+\frac{2}{\eta}\right)^{\abs I}
\]
translates of $\eta G_I$.  To see why, if $\{z_k\}_{k=1}^N\subset G_I$ is a
maximal set such that the translates
$z_k+\frac{\eta}{2}G_I$ are pairwise disjoint, then maximality implies that
the translates $z_k+\eta G_I$ cover $G_I$. Also we have
\[
 \bigcup_{k=1}^N\left(z_k+\frac{\eta}{2}G_I\right)
 \subset
 G_I+\frac{\eta}{2}G_I
 =
 \left(1+\frac{\eta}{2}\right)G_I.
\]
Since we can identify $E_I$ with $\R^{|I|}$ by an isometry, we can define $|I|$-dimensional volume $\vol_I$, i.e., Lebesgue measure, on $E_I$, such that $\vol_I(\frac{\eta}{2}G_I)=\left(\frac{\eta}{2}\right)^{\abs I} \vol_I(G_I)$. Therefore we have
\[
 N\left(\frac{\eta}{2}\right)^{\abs I}\vol_I(G_I)
 \le
 \left(1+\frac{\eta}{2}\right)^{\abs I}\vol_I(G_I)
\]
which implies
\[
 N\le \left(1+\frac{2}{\eta}\right)^{\abs I}.
\]

Now we select one point from each translate which intersects
\[
 (S^{n-1}\cap E_I)\cap G_{A_1},
\]
to form a net $\mathcal N_I$.  Then every point of
$\mathcal N_I$ is a unit vector supported in $I$, and
\begin{equation}\label{eq:NI-cardinality}
 \abs{\mathcal N_I}
 \le
 \left(1+\frac{2}{\eta}\right)^{\abs I}
 \le
 \left(\frac{3}{\eta}\right)^{\abs I}
 =
 \left(\frac{288\rho}{c_1\beta\sqrt\gamma}\right)^{\abs I}.
\end{equation}
We claim that, for every $x\in (S^{n-1}\cap E_I)\cap G_{A_1}$ there exists
$y\in\mathcal N_I$ such that
\begin{equation}\label{eq:ellipsoid-difference}
 x-y\in 2\eta(G_{A_1}\cap E_I).
\end{equation}
To explain why, if $x$ lies in a translate $z+\eta(G_{A_1}\cap E_I)$, and
$y\in\mathcal N_I$ is the chosen representative from the same translate, then
$x=z+\eta u$ and $y=z+\eta w$ for some $u,w\in G_{A_1}\cap E_I$.  Since
$G_{A_1}\cap E_I$ is convex and centrally symmetric,
 we have $(u-w)/2\in G_{A_1}\cap E_I$, which implies
$x-y=2\eta((u-w)/2)\in2\eta(G_{A_1}\cap E_I)$.

Fix $I\subset[n]$ with $\abs I\le\alpha n$.  Suppose that there exists
$x\in S^{n-1}\cap E_I$ such that
\begin{equation}\label{eq:bad-x-sparse}
 \norm{Mx}_2\le \frac{1}{16}c_1\frac{\beta\sqrt\gamma}{\rho}\varepsilon.
\end{equation}
We first show that $x\in G_{A_1}\cap E_I$.  Since
$M=A_1+\varepsilon R_1$ and $\norm{R_1}\le1$, we have
\[
 \norm{A_1x}_2
 \le \norm{Mx}_2+\varepsilon\norm{R_1x}_2
 \le \frac{1}{16}c_1\frac{\beta\sqrt\gamma}{\rho}\varepsilon+\varepsilon.
\]
Since $0<c_1,\beta,\gamma<1$ and $\rho>1$, we have
$\beta\sqrt\gamma/\rho\le1$.  Therefore, we have
\[
 \norm{A_1x}_2\le \frac{17}{16}\varepsilon\le2\varepsilon.
\]
So $x\in G_{A_1}\cap E_I$.

Choose $y\in\mathcal N_I$ satisfying \eqref{eq:ellipsoid-difference}.  From
$x-y\in2\eta(G_{A_1}\cap E_I)$ and the definition of $G_{A_1}$, we obtain
\[
 \norm{x-y}_2\le2\eta,
 \qquad
 \norm{A_1(x-y)}_2\le4\eta\varepsilon.
\]
Using $M=A_1+\varepsilon R_1$ and $\norm{R_1}\le1$, we get
\[
 \norm{M(x-y)}_2
 \le \norm{A_1(x-y)}_2+\varepsilon\norm{R_1(x-y)}_2
 \le 4\eta\varepsilon+2\eta\varepsilon
 =6\eta\varepsilon.
\]
By definition \eqref{eq:eta-definition}, we have
\[
 6\eta\varepsilon=\frac{1}{16}c_1\frac{\beta\sqrt\gamma}{\rho}\varepsilon.
\]
Combining this with \eqref{eq:bad-x-sparse}, we have
\begin{equation}\label{eq:bad-net-point}
 \norm{My}_2
 \le \norm{Mx}_2+\norm{M(x-y)}_2
 \le \frac{1}{8}c_1\frac{\beta\sqrt\gamma}{\rho}\varepsilon.
\end{equation}
In conclusion, for every support $I$, the existence of a bad vector
$x\in S^{n-1}\cap E_I$ implies the existence of a unit net point
$y\in\mathcal N_I$ satisfying \eqref{eq:bad-net-point}.

Therefore, by union bound, for fixed $I$, we have
\begin{align*}
 &\Pb_{D_1}\left\{
 \exists x\in S^{n-1}\cap E_I:
 \norm{Mx}_2\le \frac{1}{16}c_1\frac{\beta\sqrt\gamma}{\rho}\varepsilon
 \right\}
 \\\le&
 \sum_{y\in\mathcal N_I}
 \Pb_{D_1}\left\{\norm{My}_2\le\frac{1}{8}c_1\frac{\beta\sqrt\gamma}{\rho}\varepsilon\right\}
 \le
 \left(\frac{288\rho}{c_1\beta\sqrt\gamma}\right)^{\abs I}
 \exp\left(-c_2 \gamma n\right).
\end{align*}

Choose
\[
 B=\frac{288e\rho}{c_1\beta\gamma^{3/2}}.
\]
Since $0<c_1,\beta,\gamma<1$ and $\rho>1$, we have $B\ge288e$ and therefore
$\log B>5$.

Choose the constant $c_{\ref{prop:comp}.1}>0$ sufficiently large so that
\[
 c_{\ref{prop:comp}.1}\ge \frac{288e}{c_1}
\]
and, for all $t\ge e^5$, we have
\begin{equation}\label{eq:constant-choice}
 \frac{1}{c_{\ref{prop:comp}.1}}
 \frac{\log(c_{\ref{prop:comp}.1}t\log t)}{\log t}
 \le \frac{c_2}{2}.
\end{equation}
This is possible because
\[
 \frac{\log(c_{\ref{prop:comp}.1}t\log t)}
 {c_{\ref{prop:comp}.1}\log t}
 \le
 \frac{\log c_{\ref{prop:comp}.1}}{5c_{\ref{prop:comp}.1}}
 +
 \frac{2}{c_{\ref{prop:comp}.1}}
\]
for all $t\ge e^5$, and the right-hand side tends to $0$ as
$c_{\ref{prop:comp}.1}\to\infty$.

By the assumption on $\alpha$ and by the choice
$c_{\ref{prop:comp}.1}\ge 288e/c_1$, we have
\[
 \alpha
 <
 \frac{\gamma}{
 c_{\ref{prop:comp}.1}
 \log\left(
 \frac{c_{\ref{prop:comp}.1}\rho}{\beta\gamma^{3/2}}
 \right)}
 \le
 \frac{\gamma}{c_{\ref{prop:comp}.1}\log B}
 <\frac12.
\]
Using Lemma \ref{lem:support-count} and taking the union bound over all
$I\subset[n]$ with $\abs I\le\alpha n$, we have
\begin{align}\label{eq:union-bound-before-alpha}
&\Pb_{D_1}\left\{
 \exists x\in S^{n-1}:\abs{\supp(x)}\le\alpha n,
 \norm{Mx}_2\le \frac{1}{16}c_1\frac{\beta\sqrt\gamma}{\rho}\varepsilon
 \right\}
\notag\\
&\qquad\le
 \exp\left(\alpha n\log\frac{e}{\alpha}\right)
 \left(\frac{288\rho}{c_1\beta\sqrt\gamma}\right)^{\alpha n}
 \exp\left(-c_2 \gamma n\right)
\notag\\
&\qquad=
 \exp\left(
 \alpha n\log\left(\frac{288e\rho}{c_1\alpha\beta\sqrt\gamma}\right)
 -c_2 \gamma n
 \right).
\end{align}

We now verify that $
 \alpha n\log\left(\frac{288e\rho}{c_1\alpha\beta\sqrt\gamma}\right)$ can be absorbed by $c_2 \gamma n$.  The function
\[
 f(s)=s\log\left(\frac{288e\rho}{sc_1\beta\sqrt\gamma}\right),
 \qquad 0<s\le1,
\]
is increasing on $(0,1]$, because
\[
 f'(s)=
 \log\left(\frac{288e\rho}{sc_1\beta\sqrt\gamma}\right)-1
 \ge \log(288e)-1>0
\]
where we use
\[
 \frac{288e\rho}{c_1\beta\sqrt\gamma}\ge 288e.
\]
Therefore, we have
\begin{align}\label{eq:entropy-needed}
 \alpha\log\left(\frac{288e\rho}{c_1\alpha\beta\sqrt\gamma}\right)
 \le
 \frac{\gamma}{c_{\ref{prop:comp}.1}\log B}
 \log\left(
 \frac{288e\rho}{c_1\beta\sqrt\gamma}
 \cdot
 \frac{c_{\ref{prop:comp}.1}\log B}{\gamma}
 \right)  
 =
 \frac{\gamma}{c_{\ref{prop:comp}.1}\log B}
 \log\left(c_{\ref{prop:comp}.1}B\log B\right)  
 \le
 \frac{c_2\gamma}{2}
\end{align}
where the last inequality follows from \eqref{eq:constant-choice} and 
$\log B>5$.

Combining \eqref{eq:union-bound-before-alpha} and \eqref{eq:entropy-needed} we have
\begin{equation}\label{eq:sparse-final}
\Pb_{D_1}\left\{
 \inf_{x\in S^{n-1}:\abs{\supp(x)}\le\alpha n}
 \norm{Mx}_2
 \le \frac{1}{16}c_1\frac{\beta\sqrt\gamma}{\rho}\varepsilon
 \right\}
\le
 \exp\left(-\frac{c_2\gamma n}{2}\right).
\end{equation}

Finally, we extend the estimates from the sparse vectors to all compressible
vectors. Let
\[
 \nu=\frac{1}{192}c_1\frac{\beta\sqrt\gamma}{\rho}\varepsilon.
\]
On the complement of the event in \eqref{eq:sparse-final}, every
$\alpha n$-sparse unit vector $x$ satisfies
\[
 \norm{Mx}_2>\frac{1}{16}c_1\frac{\beta\sqrt\gamma}{\rho}\varepsilon.
\]
Let $v\in\operatorname{Comp}(\alpha,\nu)$.  Then there exists an
$\alpha n$-sparse vector $x$ with $\norm{v-x}_2\le\nu$.
Since $\norm{x}_2\ge1-\nu>0$, the vector $x/\norm{x}_2$ is a sparse
unit vector. We also have
\begin{align*}
 \norm{v-\frac{x}{\norm{x}_2}}_2
 &\le \norm{v-x}_2+\norm{x-\frac{x}{\norm{x}_2}}_2  \\
 &=\norm{v-x}_2+\left|\norm{x}_2-1\right|  \\
 &\le2\norm{v-x}_2
 \le2\nu.
\end{align*}
By Remark \ref{rem:norm}, we have $\norm M\le3$. Therefore, we have
\begin{align*}
 \norm{Mv}_2
 &\ge \norm{M\frac{x}{\norm{x}_2}}_2
 -\norm{M}\norm{v-\frac{x}{\norm{x}_2}}_2  \\
 &> \frac{1}{16}c_1\frac{\beta\sqrt\gamma}{\rho}\varepsilon
 -3\cdot2\cdot\frac{1}{192}c_1\frac{\beta\sqrt\gamma}{\rho}\varepsilon  \\
 &\ge
 \frac{1}{16}c_1\frac{\beta\sqrt\gamma}{\rho}\varepsilon
 -\frac{1}{32}c_1\frac{\beta\sqrt\gamma}{\rho}\varepsilon  \\
 &=\frac{1}{32}c_1\frac{\beta\sqrt\gamma}{\rho}\varepsilon.
\end{align*}
Therefore, by iterated expectation, we have
\begin{align*}
    &\Pb\left\{
 \inf_{v\in\operatorname{Comp}(\alpha,\nu)}\norm{Mv}_2
 <\frac{1}{32}c_1\frac{\beta\sqrt\gamma}{\rho}\varepsilon
 \right\}
 \\=&\E_{D_2,R_2}\Pb_{D_1}\left\{
 \inf_{v\in\operatorname{Comp}(\alpha,\nu)}\norm{Mv}_2
 <\frac{1}{32}c_1\frac{\beta\sqrt\gamma}{\rho}\varepsilon
 \right\}
 \\\le&
 \exp\left(-\frac{c_2\gamma n}{2}\right).
\end{align*}

\end{proof}

\section{Invertibility for Incompressible Vectors.} \label{sec:incomp}

To bound $\inf_{x \in \incomp(\alpha, \nu)} \norm{Mx}_2$, we use the following invertibility by distance result from \cite{rudelson2008littlewood}.

\begin{lemma}[Lemma 3.5 in \cite{rudelson2008littlewood}]\label{invertdist} Let $M$ be any random matrix. Let $Y_1,\dots,Y_n$ denote the column vectors of $M$, and let $H_k$ denote the span of all column vectors except the $k$th. Then for every $\alpha,\nu\in(0,1)$ and every $t>0$, one has
\begin{equation}
\mathbb{P}\left(\inf_{x\in \mathrm{Incomp}(\alpha,\nu)} \|Mx\|_2 < t \nu\, n^{-1/2}\right)
\le \frac{1}{\alpha n}\sum_{k=1}^{n}\mathbb{P}\left(\operatorname{dist}(Y_k,H_k)<t\right).
\end{equation}
    
\end{lemma}

We will apply Lemma \ref{invertdist} to our case $M= A+ \varepsilon R$. Let $Y_1,\dots,Y_n$ denote the column vectors of $M$, and let $H_k$ denote the span of all column vectors except the $k$th. Without loss of generality, we analyze
\begin{align*}
    \mathbb{P}\left(\operatorname{dist}(Y_n,H_n)<t\right)
\end{align*}
and all the other terms in the sum
\begin{align*}
    \sum_{k=1}^{n}\mathbb{P}\left(\operatorname{dist}(Y_k,H_k)<t\right)
\end{align*}
can be bounded by exactly the same way.

By constructing a measurable random normal vector $Z$ of the subspace $H_n$, we have
\begin{align*}
    \operatorname{dist}(Y_n,H_n) \ge |\langle Y_n, Z \rangle|
\end{align*}
and therefore it suffices to bound
\begin{align*}
    \mathbb{P}\left(|\langle Y_n, Z \rangle|<t\right)
\end{align*}
To this end, we first prove that the random normal vector $Z$ is incompressible with high probability.

\begin{lemma}[Random normal is incompressible] \label{lem:randnormincomp}
    Let $A$ be an $n \times n$ deterministic matrix such that $\norm{A}= 1$. Let $0<\e<1$. Let $0<\alpha, \beta, \gamma<1$ and $\rho>0$. Let $V$ be an $n$ dimensional pattern matrix with parameters $\alpha, \beta, \gamma,\rho$. Let $K$ and $L$ with $L>K$ be positive integers. Let $(R_1,\{J_i\}_{i \in [n]},\overline{H},H,\{X_{l,i}\}_{i \in [n],l \in[K]},\overline{R}_2,R_2)$ be a $n$-dimensional oblivious perturbation model with pattern matrix $V$ and parameters $\rho$, $K$, and $L$. Let $M= A+ \varepsilon (R_1+R_2)$. 
    Let $Y_1,...,Y_n$ be the columns of $M$. Let $\zeta$ be a measurable function such that for every $(y_1,\dots,y_{n-1})\in(\mathbb{R}^n)^{n-1}$, {
\[
\zeta(y_1,\ldots,y_{n-1})\in S^{n-1}\cap \operatorname{span}(y_1,\ldots,y_{n-1})^\perp .
\]}
Let $Z=\zeta(Y_1,\ldots,Y_{n-1})$.
{Then there exists an absolute constant $c_{\ref{lem:randnormincomp}.1}>0$ such that, if
\[
 n\ge \frac{2}{\gamma},
 \quad
 \alpha <
 \frac{\gamma/2}{
 c_{\ref{prop:comp}.1}
 \log\!\left(
 \frac{c_{\ref{prop:comp}.1}\rho}{\beta(\gamma/2)^{3/2}}
 \right)},
 \quad
 \nu = c_{\ref{prop:comp}.2}\,\varepsilon\frac{\beta\sqrt{\gamma/2}}{\rho},
\]
then}
\begin{align*}
    \mathbb{P}(Z \in \mathrm{Comp}(\alpha,\nu)) \le \exp \paren{-c_{\ref{lem:randnormincomp}.1} \gamma n}.
\end{align*}
\end{lemma}

\begin{proof}
Let $M'$ be the $(n-1) \times n$ matrix such that $M'$ is $M^T$ with last row erased. 
    We have
    \begin{align*}
        \{Z \in \mathrm{Comp}(\alpha,\nu)\} \subset \{ \exists x (x \in \mathrm{Comp}(\alpha,\nu), M' x=0) \}
    \end{align*}

Condition on $D_1$ and $R_2$, and let $A_1=A+\varepsilon R_2$. Since
\[
M^T x=A_1^Tx+\frac{\varepsilon}{\rho\sqrt n}D_2V^TD_1x,
\]
for each $j\in[n-1]$, the $j$th row of $M'$ satisfies
\[
(M'x)_j=(A_1^Tx)_j+\frac{\varepsilon}{\rho\sqrt n}\eta_{2,j}\,\row_j(V^T)D_1x.
\]
For a fixed $\alpha n$-sparse $x\in S^{n-1}$, define
\[
J(x)=\{j\in[n-1]: |\row_j(V^T)D_1x|\ge \beta\}.
\]
By definition of pattern matrix, before erasing the last row there are at least $\gamma n$ such coordinates. Therefore, we have the direct estimate
\[
|J(x)|\ge \gamma n-1\ge \frac{\gamma n}{2},
\]
where the last inequality uses $n\ge 2/\gamma$. For each $j\in J(x)$, we have
\begin{align*}
&\mathbb{P}\left( |(M'x)_j| \ge \frac{\varepsilon \beta}{\rho\sqrt{n}} \;\middle\vert\; R_2,D_1 \right)  \\
&\quad=
\mathbb{P}\left( \left|(A_1^Tx)_j + \frac{\varepsilon}{\rho\sqrt{n}}\eta_{2,j}\row_j(V^T)D_1x \right| \ge \frac{\varepsilon \beta}{\rho\sqrt{n}} \;\middle\vert\; R_2,D_1\right)
\ge \frac12 .
\end{align*}
The coordinates indexed by $J(x)$ are conditionally independent because the signs $\eta_{2,j}$ are independent after conditioning on $(R_2,D_1)$. Tensorization therefore implies that there exist absolute constants $c_1,c_2>0$ such that
\[
\mathbb{P}\left(\|M'x\|_2\le c_1\varepsilon\frac{\beta\sqrt{\gamma/2}}{\rho}\;\middle\vert\; R_2,D_1\right)
\le \exp(-c_2\gamma n).
\]
The results then follows from the same ellipsoid covering argument used in Theorem \ref{prop:comp}.

\end{proof}

Now we have proved that the random normal vector $Z$ is an incompressible vector with high probability. This is useful because incompressible vectors are spread according to the next lemma which follows from \cite{rudelson2008littlewood}.

\begin{lemma}[Incompressible vectors are spread]\label{lem:sharpened-spread}
Let  {\(0<\alpha,\nu<1\), and let } \(x\in \operatorname{Incomp}(\alpha,\nu)\). Then the set
\[
\Theta=\left\{\, i\in[n] : \frac{\nu}{\sqrt{2n}}<|x_i|\le \sqrt{\frac{2}{\alpha n}} \,\right\}
\]
satisfies
\[
|\Theta|>\frac{\alpha n}{2}.
\]
\end{lemma}

\begin{proof}
Define
\[
A=\left\{\, i\in[n] : |x_i|\le \frac{\nu}{\sqrt{2n}} \,\right\}.
\]
Let \(P_A:\mathbb R^n\to\mathbb R^n\) denote the orthogonal projection onto the coordinate subspace supported on \(A\), i.e., 
\[
(P_Ax)_i=
\begin{cases}
 x_i, & i\in A,\\
 0, & i\notin A.
\end{cases}
\]
Then
\[
\|P_Ax\|_2^2
=
\sum_{i\in A}|x_i|^2
\le
|A|\cdot \frac{\nu^2}{2n}
\le
n\cdot \frac{\nu^2}{2n}
=
\frac{\nu^2}{2}.
\]
In other words, we have
\[
\|P_Ax\|_2\le \frac{\nu}{\sqrt 2}<\nu.
\]

We claim that
\[
|A^c|>\alpha n.
\]
To explain why, suppose for contradiction that \(|A^c|\le \alpha n\). And let
\[
y=P_{A^c}x.
\]
Then \(y\neq 0\), because otherwise \(x=P_Ax\) and therefore
\[
1=\|x\|_2=\|P_Ax\|_2\le \frac{\nu}{\sqrt 2}<1,
\]
which is not possible. Now define
\[
z=\frac{y}{\|y\|_2}.
\]
Then \(z\in S^{n-1}\) and
\[
|\operatorname{supp}(z)|\le |A^c|\le \alpha n.
\]
Let
\[
\delta=\|x-y\|_2=\|P_Ax\|_2.
\]
Since \(x=y+P_Ax\) and the two summands have disjoint supports, they are orthogonal. Because \(\|x\|_2=1\), we get
\[
\|y\|_2^2=1-\delta^2.
\]
Therefore, we have
\[
\langle x,z\rangle
=
\left\langle y+P_Ax,\frac{y}{\|y\|_2}\right\rangle
=
\|y\|_2,
\]
so
\[
\|x-z\|_2^2
=
2-2\langle x,z\rangle
=
2-2\|y\|_2
=
2\bigl(1-\sqrt{1-\delta^2}\bigr).
\]
Since \(0\le \delta\le 1\), we have the direct estimate
\[
1-\sqrt{1-\delta^2}
=
\frac{\delta^2}{1+\sqrt{1-\delta^2}}
\le
\delta^2.
\]
Combining above inequalities, we have
\[
\|x-z\|_2^2
\le
2\delta^2
=
2\|P_Ax\|_2^2
\le
\nu^2.
\]
Thus \(\|x-z\|_2\le \nu\), with \(z\in S^{n-1}\) and \(|\operatorname{supp}(z)|\le \alpha n\). This shows that \(x\in \operatorname{Comp}(\alpha,\nu)\), contradicting \(x\in \operatorname{Incomp}(\alpha,\nu)\). So our claim must be true.

Now define
\[
B=\left\{\, i\in[n] : |x_i|>\sqrt{\frac{2}{\alpha n}} \,\right\}.
\]
If \(|B|\ge \alpha n/2\), then \(B\neq \emptyset\) and therefore
\[
1=\|x\|_2^2
\ge
\sum_{i\in B}|x_i|^2
>
|B|\cdot \frac{2}{\alpha n}
\ge
\frac{\alpha n}{2}\cdot \frac{2}{\alpha n}
=
1,
\]
which is a contradiction. So we must have
\[
|B|<\frac{\alpha n}{2}.
\]

Finally, let
\[
\Theta=A^c\setminus B.
\]
Then
\[
|\Theta|
\ge
|A^c|-|B|
>
\alpha n-\frac{\alpha n}{2}
=
\frac{\alpha n}{2}.
\]
If \(i\in \Theta\), then \(i\notin A\) and \(i\notin B\), so
\[
\frac{\nu}{\sqrt{2n}}<|x_i|\le \sqrt{\frac{2}{\alpha n}}.
\]
This proves the desired result.
\end{proof}

The standard argument (e.g., \cite{rudelson2008littlewood}) to bound the probability \begin{align*}
    \mathbb{P}\left(\operatorname{dist}(Y_n,H_n)<t\right)
\end{align*}
uses the fact that
\begin{align*}
    \operatorname{dist}(Y_n,H_n) \ge |\langle Y_n, Z \rangle|
\end{align*}
and then conditions on $Z$. 

In our case, we cannot use this approach directly because $\sigma(Y_n)$ and $\sigma(Y_1,...,Y_{n-1})$ are not mutually independent. Nevertheless, the pure random signs $X_{i,n}$ will be independent with $\sigma(Y_1,...,Y_{n-1})$. To make the argument work, we need to show that there are sufficiently many random signs in the last column which will not be removed eventually ($h_{ij} \ne 0$). To understand when a random sign in the last column will be removed eventually, one key observation is that, a nonzero position in the $n$th column of $\bar R_2$ will be eventually be reset to zero if and only if it is in a row that is already heavy enough in the $1,...,n-1$ columns.

More formally, we define index set of the heavy rows of $\bar H$ up to column $n-1$ as
\begin{align}
		I_{\mathrm{heavy},n}
		&=\left\{ i\in[n]: \sum_{j=1}^{n-1}\bar h_{i,j} \ge L\right\}
		\label{eq:Iheavy}
	\end{align}
and the index set of the light rows up to column $n-1$ 
\begin{align}
		I_{\mathrm{light},n}
		&=\left\{ i\in[n]: \sum_{j=1}^{n-1}\bar h_{i,j} < L\right\}.
		\label{eq:Ilight}
	\end{align}
So the heavy rows may be removed if there is an extra nonzero entry in the column $n$, but the light rows will never be removed no matter what happens in the column $n$. With the above observation, we can write
\begin{align*}
			\langle Y_n,Z\rangle
			=
			\big\langle \mathrm{Col}_n(A+\varepsilon R_1),\, Z\big\rangle
			+\frac{\varepsilon}{L}\sum_{l=1}^K X_{l,n}\, Z_{j(n,l)}\, \mathbf{1}_{\{j(n,l)\in I_{\mathrm{light},n}\}}
		\end{align*}

We will write $\langle Y_n, Z \rangle$ as a function of two independent random elements. To this end, we first define the random elements \[
		\Xi_1=(R_1,\ (X_{l,j})_{l\in[K],\,j\in[n-1]},\ J_1,\dots,J_{n-1}),
		\quad
		\Xi_2=J_n.
		\quad \Xi_3=\{X_{l,n}: l \in [K]\}\] In other words, $(\Xi_1,\Xi_2)$ contains the information of $R_1$, the position of all nonzero entries of $R_2$, and the value of random signs for the entries of $R_2$ only in columns $1,...,n-1$. The random element $\Xi_3$ contains the signs in column $n$. By construction, $\Xi_1,\Xi_2,\Xi_3$ are mutually independent, so it suffices to write $\langle Y_n, Z \rangle$ and a function of $\Xi_1,\Xi_2,\Xi_3$ so that we can use conditioning (iterated expectation). In particular, $I_{\mathrm{heavy},n}$ and $I_{\mathrm{light},n}$ are completely determined by $\Xi_1$ and the random normal vector $Z$ is completely determined by $(\Xi_1,\Xi_2)$ (or more precisely, completely determined by $\Xi_1$ and $J_n \cap I_{\mathrm{heavy},n}$). Therefore, in the expression
        \begin{align*}
			\langle Y_n,Z\rangle
			=
			\big\langle \mathrm{Col}_n(A+\varepsilon R_1),\, Z\big\rangle
			+\frac{\varepsilon}{L}\sum_{l=1}^K X_{l,n}\, Z_{j(n,l)}\,\mathbf{1}_{\{j(n,l)\in I_{\mathrm{light},n}\}}.
		\end{align*}
        the parts, $\big\langle \mathrm{Col}_n(A+\varepsilon R_1),\, Z\big\rangle$ and $\, Z_{j(n,l)}\,\mathbf{1}_{\{j(n,l)\in I_{\mathrm{light},n}\}}$ will be fixed after conditioning on $(\Xi_1,\Xi_2)$, and then bounding $\langle Y_n,Z\rangle$ becomes the Littlewood-Offord problem. To get a good bound, we need sufficiently many good coordinates to show up in the coefficient vector $\{\, Z_{j(n,l)}\,\mathbf{1}_{\{j(n,l)\in I_{\mathrm{light},n}\}}\}_{l \in [K]}$. To this end, the first step is to show there are enough light rows, or in other words, there cannot be too many heavy rows.

\begin{lemma}[Heavy rows are not too many]
		\label{thm:Iheavy}
		Let $(R_1,\{J_i\}_{i \in [n]},\overline{H},H,\{X_{l,i}\}_{i \in [n],l \in[K]},\overline{R}_2,R_2)$ be a $n$-dimensional oblivious perturbation model with pattern matrix $V$ and parameters $K$ and $L$. Let \begin{align*}
		I_{\mathrm{heavy},n}
		&=\left\{ i\in[n]: \sum_{j=1}^{n-1}\bar h_{i,j} \ge L\right\}.
	\end{align*}
    be the index set of heavy rows up to column $n-1$. Then, for $L>K$, we have
		\begin{equation}\label{eq:chernoff-bin}
			\mathbb{E}|I_{\mathrm{heavy},n}|\le n e^{-K}\Bigl(\frac{eK}{L}\Bigr)^{L}.
		\end{equation}
        and for any $\lambda\in(0,1]$, we have
		\begin{align}
			\mathbb{P}\Bigl\{|I_{\mathrm{heavy},n}|>\tfrac{\lambda}{2}n\Bigr\}
			\le&
			\exp\Biggl(-\frac{\bigl(\max\{(\tfrac{\lambda}{2}n-\mathbb{E}|I_{\mathrm{heavy},n}|),0\}\bigr)^2}{2nK^2}\Biggr)\label{eq:main-bound}
            \\\le&
			\exp\Biggl(-\frac{\bigl(\max\{\tfrac{\lambda}{2}n-n e^{-K}(eK/L)^L,0\}\bigr)^2}{2nK^2}\Biggr).\label{eq:main-bound-explicit}
		\end{align}

	\end{lemma}

	\begin{proof}
The idea of the proof is to first upper bound $\mathbb{E}|I_{\mathrm{heavy},n}|$ and then obtain concentration inequalities for $|I_{\mathrm{heavy},n}|$.

		Fix $i\in[n]$ and $j\in[n-1]$. Since $J_j$ is uniform among all $K$--subsets of $[n]$, we have
		\begin{equation}\label{eq:prob-in}
			\mathbb{P}(\bar h_{i,j}=1)=\mathbb{P}(i\in J_j)=\frac{\binom{n-1}{K-1}}{\binom{n}{K}}=\frac{K}{n}.
		\end{equation}
		Because $J_1,\dots,J_{n-1}$ are independent, the random variables
		$\bar h_{i,1}=\mathbf{1}_{\{i\in J_1\}},\dots,\bar h_{i,n-1}=\mathbf{1}_{\{i\in J_{n-1}\}}$
		are independent Bernoulli random variables with success probability $K/n$.
		
		Therefore, for each row $i$, we have
		$\sum_{j=1}^{n-1} \bar h_{i,j}$ follows Binomial distribution with parameters $n-1$ and $K/n$.
		In particular, the distribution does not depend on $i$. Let $S$ be a Binomial random variable with parameters $n-1$ and $K/n$. Observe that $\E[S]=K(1-1/n)$. Assume $L>K$. By Chernoff bound, we have
		\[
		\mathbb{P}\{S\ge L\}\le e^{-K}\Bigl(\frac{eK}{L}\Bigr)^L,
		\]
		
        Since
		\[
		|I_{\mathrm{heavy},n}|=\sum_{i=1}^{n} \mathbf{1}_{\{\sum_{j=1}^{n-1} \bar h_{i,j}\ge L\}},
		\]
        by linearity of expectation, we have
\begin{equation}\label{eq:ET}
			\mathbb{E}|I_{\mathrm{heavy},n}|
			=\sum_{i=1}^n \mathbb{P}\Bigl\{\sum_{j=1}^{n-1} \bar h_{i,j}\ge L\Bigr\}
			=n\,\mathbb{P}\{S\ge L\} \le n\,e^{-K}\Bigl(\frac{eK}{L}\Bigr)^L
		\end{equation}

        To obtain concentration inequalities for $|I_{\mathrm{heavy},n}|$, we can write
        \[
		T=|I_{\mathrm{heavy},n}|=\sum_{i=1}^n \mathbf{1}_{\{\sum_{j=1}^{n-1} \bar h_{i,j}\ge L\}}=f(J_1,\dots,J_{n-1}).
		\]
		Let $k\in[n-1]$ be fixed. Consider two sequences of column-sets
		\[
		(J_1,\dots,J_{n-1})\quad\text{and}\quad (J_1,\dots,J_{k-1},J_k',J_{k+1},\dots,J_{n-1})
		\]
		that differ only in the $k$th coordinate. For any fixed row index $i\in[n]$, replacing $J_k$ by $J_k'$ changes the row sum
		$\sum_{j=1}^n \bar h_{i,j}$ by at most $1$, and it changes it only if $i\in J_k\triangle J_k'$ (symmetric difference).
		Therefore, the indicator
		\(\mathbf{1}_{\{\sum_{j=1}^n \bar h_{i,j}\ge L\}}\)
		can change only for $i\in J_k\triangle J_k'$.
		So the total count $T$ can change by at most
		\begin{equation}\label{eq:bd-2K}
			|J_k\triangle J_k'|\le |J_k|+|J_k'|=2K,
		\end{equation}
		because both $J_k$ and $J_k'$ have exactly $K$ elements.
		Therefore, $T$ satisfies the bounded differences condition of Lemma~\ref{lem:bd} with $m=n-1$ and
		\begin{equation}\label{eq:ck}
			b_1=\cdots=b_{n-1}=2K.
		\end{equation}
		
		\
		By Lemma~\ref{lem:bd} and \eqref{eq:ck}, we have
		\[
		v=\frac14\sum_{k=1}^{n-1} b_k^2=\frac14\,(n-1)\,(2K)^2 \le n K^2.
		\]
		Therefore, for every $t>0$, we have
		\begin{equation}\label{eq:bd-applied}
			\mathbb{P}\{T-\mathbb{E}T>t\}\le \exp\Bigl(-\frac{t^2}{2nK^2}\Bigr).
		\end{equation}
		To bound $\mathbb{P}\{T>\tfrac{\lambda}{2}n\}$, we choose
		\[
		t=\tfrac{\lambda}{2}n-\mathbb{E}T
		\]
        so that
        \[
		\{T>\tfrac{\lambda}{2}n\}=\{T-\mathbb{E}T>t\}
		\]
		If $t\le 0$, then $\tfrac{\lambda}{2}n\le \mathbb{E}T$ and the inequality \eqref{eq:main-bound} holds trivially because its right-hand side equals $1$.
		So it suffices to consider $t>0$.
		so \eqref{eq:bd-applied} implies
		\[
		\mathbb{P}\Bigl\{T>\tfrac{\lambda}{2}n\Bigr\}
		\le \exp\Bigl(-\frac{t^2}{2nK^2}\Bigr)
		=\exp\Biggl(-\frac{\bigl(\tfrac{\lambda}{2}n-\mathbb{E}T\bigr)^2}{2nK^2}\Biggr).
		\]
		Combining the cases $t\le 0$ and $t>0$ proves \eqref{eq:main-bound}.

	\end{proof}

With the above preparation, we can now show that sufficiently many good coordinates show up in the coefficient vector $\left( Z_{j(n,l)}\,\mathbf{1}_{\{j(n,l)\in I_{\mathrm{light},n}\}}\right)_{l \in [K]}$.

\begin{lemma}[Good coordinates]\label{lem:spreadlight}
Let $A$ be an $n \times n$ deterministic matrix such that $\norm{A}= 1$. Let $0<\e<1$. Let $0<\alpha, \beta, \gamma,\nu<1$, and $\rho \ge 1$. Let $V$ be an $n$ dimensional pattern matrix with parameters $\alpha, \beta, \gamma,\rho$. Let $K$ and $L$ with $L>K$ be positive integers. Let $(R_1,\{J_i\}_{i \in [n]},\overline{H},H,\{X_{l,i}\}_{i \in [n],l \in[K]},\overline{R}_2,R_2)$ be a $n$-dimensional oblivious perturbation model with pattern matrix $V$ and parameters $\rho$, $K$, and $L$. Let $M= A+ \varepsilon (R_1+R_2)$. Let $Y_1,...,Y_n$ be the columns of $M$. Let $\zeta$ be a measurable function such that for every $(y_1,\dots,y_{n-1})\in(\mathbb{R}^n)^{n-1}$, we have
		\[
\zeta(y_1,\ldots,y_{n-1})\in S^{n-1}\cap \operatorname{span}(y_1,\ldots,y_{n-1})^\perp .
\] Let $Z=\zeta(Y_1,...,Y_{n-1})$.
Let \begin{align*}
    \Theta_0=\{l \in [n]:\frac{\nu}{\sqrt{2n}}\leqslant |Z_l|\leqslant \sqrt{\frac{2}{\alpha n}}\}
\end{align*}
Let $0< \lambda<1$. Then for any
		$t\in\bigl(0,\lambda/8\bigr)$, we have
		\begin{align*}
			\Pb\Bigl\{\abs{J_n\cap \Theta_0\cap I_{\mathrm{light},n}}\le tK\Bigr\}
			\ &\le\
			\Pb\{\abs{\Theta_0}<\lambda n\}
			+\Pb\Bigl\{\abs{I_{\mathrm{heavy},n}}>\tfrac{\lambda}{2}n\Bigr\}
            \\
			&\quad
			+\exp\Bigl(-2K\bigl(\tfrac34-\tfrac{\lambda}{2}\bigr)^2\Bigr)
			+\exp\Bigl(-\tfrac{K}{2}\bigl(\tfrac{\lambda}{2}-4t\bigr)^2\Bigr).
		\end{align*}
\end{lemma}

\begin{proof}
		The first step is to consider the two good events
		\[
		\mathcal E_1=\{\abs{\Theta_0}\ge \lambda n\},
		\quad
		\mathcal E_2=\Bigl\{\abs{I_{\mathrm{heavy},n}}\le \tfrac{\lambda}{2}n\Bigr\}.
		\]
        and then we can bound $\abs{J_n\cap \Theta_0\cap I_{\mathrm{light},n}}$ under the good event $\mathcal E_1 \cap \mathcal E_2$.
		Formally, by union bound, we obtain
\begin{equation}\label{eq:union-new}
			\Pb\Bigl(\{\abs{J_n\cap \Theta_0\cap I_{\mathrm{light},n}}\le tK\}\Bigr)
			\le
			\Pb(\mathcal E_1^c)+\Pb(\mathcal E_2^c)
			+\Pb\Bigl(\{\abs{J_n\cap \Theta_0\cap I_{\mathrm{light},n}}\le tK\} \cap \mathcal E_1 \cap\mathcal E_2\Bigr).
		\end{equation}
		It remains to bound the term $\Pb\Bigl(\{\abs{J_n\cap \Theta_0\cap I_{\mathrm{light},n}}\le tK\} \cap \mathcal E_1 \cap\mathcal E_2\Bigr)$.
		
As is explained in the overview, the key idea to show that $\Pb\Bigl(\{\abs{J_n\cap \Theta_0\cap I_{\mathrm{light},n}}\le tK\} \cap \mathcal E_1 \cap\mathcal E_2\Bigr)$ is small is to condition on $R_1$, the first $n-1$ columns of $\bar R_2$, and  the set $J_n \cap I_{\mathrm{heavy}, n}$, which in combination completely determines $\Theta_0$. To do this formally, let
		\[
		\Xi_1=(R_1,\ (X_{l,j})_{l\in[K],\,j\in[n-1]},\ J_1,\dots,J_{n-1}),
		\]which determines $R_1$ and the first $n-1$ columns of $\bar R_2$, and let
        \[
		\Xi_2=J_n.
		\]
		By construction, $\Xi_1$ and $\Xi_2$ are independent.
		Also, $I_{\mathrm{light},n}$ and $I_{\mathrm{heavy},n}$ are functions of only $\Xi_1$.
        We observe that the first $n-1$ columns of $R_2$ only depends on $\Xi_1$ and $J_n\cap I_{\mathrm{heavy},n}$ because $J_n$ can affect the first $n-1$ columns of $R_2$ only by causing deletions and what happens in $I_{\mathrm{light},n}$
        will not cause deletion in any case. More formally, for $i \in [n]$ and $j \in [n-1]$, we can explicitly write
        \begin{align*}
            (r_2)_{ij}
=
(\bar r_2)_{ij}
\mathbf 1_{([n]\backslash(J_n\cap I_{\mathrm{heavy},n}))\cap
\{s:\sum_{q=1}^{n-1}\bar h_{s,q}\le L\}}(i).
        \end{align*}
		Since $\Theta_0$ depends only on the first $n-1$ columns of $R_2$ and therefore and only depends on $\bigl(\Xi_1,\ J_n\cap I_{\mathrm{heavy},n}\bigr)$,
         there exists a measurable map $\psi$ such that
		\begin{equation}\label{eq:Theta0-psi}
			\Theta_0=\psi\bigl(\Xi_1,\ J_n\cap I_{\mathrm{heavy},n}\bigr).
		\end{equation}
		
		Conditioning on $R_1$ and the first $n-1$ columns of $\bar R_2$ is relatively easy because of independence between $\Xi_1$ and $\Xi_2$. To separate the randomness from $\Xi_1$ and from $\Xi_2$ and do conditioning, we write
		\[
		\1_{\{\abs{J_n\cap \Theta_0\cap I_{\mathrm{light},n}}\le tK,\ \abs{\Theta_0}\ge \lambda n,\ \abs{I_{\mathrm{heavy},n}}\le \tfrac{\lambda}{2}n\}}
		=\phi(\Xi_1,\Xi_2)
		\]
		for some measurable $\phi$.
		Let $\phi_1(\xi_1)=\E_{\Xi_2}\phi(\xi_1,\Xi_2)$.
		By independence,
		\begin{align}\label{eq:phi1-outer}
			&\Pb\Bigl(\abs{J_n\cap \Theta_0\cap I_{\mathrm{light},n}}\le tK,\ \abs{\Theta_0}\ge \lambda n,\ \abs{I_{\mathrm{heavy},n}}\le \tfrac{\lambda}{2}n\Bigr)
			=\E\,\phi_1(\Xi_1).
		\end{align}
		
		Now we fix (condition on) the randomness from $\Xi_1$. Formally, we fix $\xi_1$ with $\abs{I_{\mathrm{heavy},n}(\xi_1)}\le \tfrac{\lambda}{2}n$, because otherwise we have $\abs{I_{\mathrm{heavy},n}(\xi_1)}>\lambda n/2$, and then $\phi_1(\xi_1)=0$, which reduces to a trivial case.
		Then $I_{\mathrm{light},n}(\xi_1)$ and $I_{\mathrm{heavy},n}(\xi_1)$ are deterministic. To also condition on $J_n\cap I_{\mathrm{heavy},n}$, we define
		\[
		\Gamma_1=J_n\cap I_{\mathrm{heavy},n}(\xi_1),
		\quad
		\Gamma_2=J_n\cap I_{\mathrm{light},n}(\xi_1).
		\]
		By equation \eqref{eq:Theta0-psi} (with $\Xi_1=\xi_1$), we can write $\Theta_0=\psi(\xi_1,\Gamma_1)$.

		To ensure that there is still enough randomness after we conditioning on $\Xi_1$ and $\Gamma_1$, we need to show  $|\Gamma_1|$ is not too large. To separate the bad event that $|\Gamma_1|$ is large (which will happen with low probability), we write (for fixed $\xi_1$ with $\abs{I_{\mathrm{heavy},n}(\xi_1)}\le \tfrac{\lambda}{2}n$)
		\begin{align*}
			\phi(\xi_1,\Xi_2)
			&=\1_{\{\abs{\psi(\xi_1,\Gamma_1)\cap \Gamma_2}\le tK,\ \abs{\psi(\xi_1,\Gamma_1)}\ge \lambda n\}}\\
			&\le \1_{\{\abs{\Gamma_1}>\tfrac34K\}}
			+\1_{\{\abs{\psi(\xi_1,\Gamma_1)\cap \Gamma_2}\le tK,\ \abs{\psi(\xi_1,\Gamma_1)}\ge \lambda n,\ \abs{\Gamma_1}\le\tfrac34K\}}.
		\end{align*}
		Taking expectations over $\Xi_2$ gives
		\begin{align}\label{eq:phi1-split}
			\phi_1(\xi_1)
			&\le \Pb\bigl(\abs{\Gamma_1}>\tfrac34K\bigr)
			+\E\Bigl[\1_{\{\abs{\psi(\xi_1,\Gamma_1)\cap \Gamma_2}\le tK,\ \abs{\psi(\xi_1,\Gamma_1)}\ge \lambda n,\ \abs{\Gamma_1}\le\tfrac34K\}}\Bigr].
		\end{align}

		Now we will show that with high probability $|\Gamma_1|$ is not too large. And the idea is to use Hoeffding's inequality for random sampling without replacement (Proposition~\ref{prop:BM12}). We continue working with fixed $\xi_1$ with $\abs{I_{\mathrm{heavy},n}(\xi_1)}\le \tfrac{\lambda}{2}n$. Since $J_n$ is an unordered random sampling without replacement, by definition \ref{def:unorsamp}, there exists a random sequence \(T_1\) which is a random sampling of an ordered \(K\)-element sequence uniformly from \([n]\)
without replacement, such that
\[
J_n(\omega):=\{{T_1(\omega)(1)},\dots,{T_1(\omega)(k)}\}=\mathrm{Range}(T_1(\omega)).
\]
Then the sequence $S_{1}(1),\dots,S_{1}(K)$ with $S_{1}(j)=\1_{ I_{\mathrm{heavy},n}(\xi_1)}(T_1(\omega)(j))$ is a random sampling of an ordered \(K\)-element sequence uniformly from population $\mathcal S_1$ without replacement, where \[\mathcal S_1=(s_{1,1},\dots,s_{1,n})\] with
		\[
		s_{1,i}=\1_{I_{\mathrm{heavy},n}(\xi_1)}(i)\in\{0,1\}.
		\]
		Then $\sum_{j=1}^K S_{1}(j)$ is the number of sampled indices in $J_n$ that fall in $I_{\mathrm{heavy},n}(\xi_1)$, i.e.
		\[
		\sum_{j=1}^K S_{1}(j) \ {=}\ \abs{J_n\cap I_{\mathrm{heavy},n}(\xi_1)}=\abs{\Gamma_1}.
		\]
		Here $a=0$, $b=1$ in the sense of Proposition~\ref{prop:BM12}, and
		\[
		\frac{1}{n}\sum_{i=1}^n s_{1,i}=\frac{\abs{I_{\mathrm{heavy},n}(\xi_1)}}{n}\le \tfrac{\lambda}{2}
		\]
		Therefore, applying Proposition~\ref{prop:BM12} with $t=\tfrac34-\frac{1}{n}\sum_{i=1}^n s_{1,i}$, we obtain
		\begin{align}
			\Pb\bigl(\abs{\Gamma_1}>\tfrac34K\bigr)
			 \le &\Pb\Bigl(\frac1K\sum_{j=1}^K S_{1}(j)-\frac{1}{n}\sum_{i=1}^n s_{1,i}\ge \tfrac34-\frac{1}{n}\sum_{i=1}^n s_{1,i}\Bigr)\nonumber
			\\\le& \exp\bigl(-2K(\tfrac34-\frac{1}{n}\sum_{i=1}^n s_{1,i})^2\bigr)
			\le \exp\Bigl(-2K\bigl(\tfrac34-\tfrac{\lambda}{2}\bigr)^2\Bigr).
			\label{eq:r-tail-new}
		\end{align}

		Now, (after conditioning on $R_1$, the first $n-1$ columns of $\bar R_2$, and  the set $J_n \cap I_{\mathrm{heavy}, n}$) we can take advantage of the extra randomness  in the choice of $\Gamma_2$ to guarantee that, 
with high probability, sufficiently many of elements of the set $J_n \cap I_{\mathrm{light},n}$  fall in $\Theta_0 \cap I_{\mathrm{light},n}$. More formally, we will show that, conditioning on  the first $n-1$ columns of $\bar R_2$ and $J_n\cap I_{\mathrm{heavy},n}$, the set $J_n \cap I_{\mathrm{light},n}$ is a random uniform subset of $I_{\mathrm{light},n}$ with cardinality $K-|J_n\cap I_{\mathrm{heavy},n}|$. To this end, let $\mu(\theta_1,B)$ be a regular conditional probability of $\Gamma_2$ given $\Gamma_1$ (so intuitively, $\mu(\theta_1,B)=\Pb(B|\Gamma_1=\theta_1)$), and then we have
		\begin{align}
			&\E\Bigl[\1_{\{\abs{\psi(\xi_1,\Gamma_1)\cap \Gamma_2}\le tK,\ \abs{\psi(\xi_1,\Gamma_1)}\ge \lambda n,\ \abs{\Gamma_1}\le\tfrac34K\}} \ \Big|\ \Gamma_1\Bigr]\notag\\
			&\quad=\1_{\{\abs{\psi(\xi_1,\Gamma_1)}\ge \lambda n,\ \abs{\Gamma_1}\le\tfrac34K\}}
			\int \1_{\{\abs{\psi(\xi_1,\Gamma_1)\cap \theta_2}\le tK\}}\ \mu(\Gamma_1,d\theta_2).
			\label{eq:condexp-mu}
		\end{align}
				Fix an index set $\theta_1\subset I_{\mathrm{heavy},n}(\xi_1)$ with $\abs{\theta_1}\le \tfrac34K$.
		By Lemma~\ref{lem:cond-unif-mu} applied with $J=J_n$ and $I=I_{\mathrm{heavy},n}(\xi_1)$, the conditional law $\mu(\theta_1,\cdot)$ is the probability distribution of a random sampling of an unordered $(K-\abs{\theta_1})$-element subset
uniformly from  $I_{\mathrm{light},n}(\xi_1)$ without replacement. 
        In other words, we have
        \begin{align*}
            \int \1_{\{\abs{\psi(\xi_1,\theta_1)\cap \theta_2}\le tK\}}\ \mu(\theta_1,d\theta_2)=\Pb_{\Upsilon}(\{\abs{\psi(\xi_1,\theta_1)\cap \Upsilon}\le tK\})
        \end{align*}
        where $\Upsilon$ is the probability distribution of a random sampling of an unordered $(K-\abs{\theta_1})$-element subset
uniformly from  $I_{\mathrm{light},n}(\xi_1)$ without replacement, in some appropriate auxiliary probability space. Enumerate the elements in $I_{\mathrm{light},n}(\xi_1)$ so that $I_{\mathrm{light},n}(\xi_1)=\{s_{2,1},...,s_{2,\abs{I_{\mathrm{light},n}}}\}$.

By definition \ref{def:unorsamp}, there exists a random sequence \(T_2\), which is a random sampling of an ordered \((K-\abs{\theta_1})\)-element sequence uniformly from \(|I_{\mathrm{light},n}(\xi_1)|\)
without replacement (also in the auxiliary probability space), such that
\[
\Upsilon(\omega)=\{s_{2,T_2(\omega)(1)},\dots,s_{2,T_2(\omega)(K-|\theta_1|)}\}=\mathrm{Range}(s \circ T_2(\omega)).
\]
		Set 
		\[
		A=\psi(\xi_1,\theta_1)\cap I_{\mathrm{light},n}(\xi_1).
		\]
		For fixed $\xi_1$ and $\theta_1$ with $\abs{\psi(\xi_1,\theta_1)}\ge \lambda n$ and $\abs{I_{\mathrm{heavy},n}(\xi_1)}\le \tfrac{\lambda}{2}n$, we have
		\begin{equation}\label{eq:A-lower}
			\abs{A}
			\ge \abs{\psi(\xi_1,\theta_1)}-\abs{I_{\mathrm{heavy},n}(\xi_1)}
			\ge \tfrac{\lambda}{2}n.
		\end{equation}
Then the sequence $S_{3}(1),\dots,S_{3}(K-|\theta_1|)$ with $S_{3}(j)=\1_{A}(s_{2,T_2(\omega)(j)})$ is a random sampling of an ordered \((K-\abs{\theta_1})\)-element sequence uniformly from population $\mathcal S_3$ without replacement, where \[\mathcal S_3=(s_{3,1},\dots,s_{3,|I_{\mathrm{light},n}(\xi_1)|})\] with
		\[
		s_{3,i}=\1_{A}(s_{2,i})\in\{0,1\}.
		\]
Also, we have
		\[
		\frac{1}{K-\abs{\theta_1}}\sum_{j=1}^{K-\abs{\theta_1}} S_3(j)=\frac{\abs{A\cap \Upsilon}}{K-\abs{\theta_1}}=\frac{\abs{\psi(\xi_1,\theta_1)\cap \Upsilon}}{K-\abs{\theta_1}},
		\quad
		\mu=\frac{\abs{A}}{\abs{I_{\mathrm{light},n}(\xi_1)}}\ge \frac{\abs{A}}{n}\ge \tfrac{\lambda}{2},
		\]
		where we used \eqref{eq:A-lower} and $\abs{I_{\mathrm{light},n}(\xi_1)}\le n$.
		Also, $\abs{\theta_1}\le 3K/4$ implies $K-\abs{\theta_1}\ge K/4$, so we have $tK/(K-\abs{\theta_1})\le 4t$.
		Therefore, the event $\{\abs{\psi(\xi_1,\theta_1)\cap \Upsilon}\le tK\}$ implies
		\[
		\frac{1}{K-\abs{\theta_1}}\sum_{j=1}^{K-\abs{\theta_1}} S_3(j) \le 4t,
		\quad\text{i.e.}\quad
		\frac{1}{K-\abs{\theta_1}}\sum_{j=1}^{K-\abs{\theta_1}} S_3(j)-\mu \le -( \mu-4t).
		\]
		Since $t<\lambda/8$, we have $\mu-4t\ge \lambda/2-4t>0$.
		Applying \eqref{eq:BM-hoeffding-lower}, we obtain
		\begin{align}
			\int \1_{\{\abs{\psi(\xi_1,\theta_1)\cap \theta_2}\le tK\}}\ \mu(\theta_1,d\theta_2)
            =&\Pb_{\Upsilon}(\{\abs{\psi(\xi_1,\theta_1)\cap \Upsilon}\le tK\})\notag
			\\\le &\exp\bigl(-2(K-\abs{\theta_1})(\mu-4t)^2\bigr)\notag
			\\\le& \exp\Bigl(-2\cdot\tfrac{K}{4}\bigl(\tfrac{\lambda}{2}-4t\bigr)^2\Bigr)\notag
			\\=&\exp\Bigl(-\tfrac{K}{2}\bigl(\tfrac{\lambda}{2}-4t\bigr)^2\Bigr).
			\label{eq:BM-bound-main}
		\end{align}
		Inserting \eqref{eq:BM-bound-main} into \eqref{eq:condexp-mu} and taking expectations gives
		\begin{equation}\label{eq:small-hit-bound}
			\E\Bigl[\1_{\{\abs{\psi(\xi_1,\Gamma_1)\cap \Gamma_2}\le tK,\ \abs{\psi(\xi_1,\Gamma_1)}\ge \lambda n,\ \abs{\Gamma_1}\le\tfrac34K\}}\Bigr]
			\le \exp\Bigl(-\tfrac{K}{2}\bigl(\tfrac{\lambda}{2}-4t\bigr)^2\Bigr).
		\end{equation}

		Combining \eqref{eq:phi1-split}, \eqref{eq:r-tail-new}, and \eqref{eq:small-hit-bound}, we conclude
		\[
		\phi_1(\xi_1)
		\le \exp\Bigl(-2K\bigl(\tfrac34-\tfrac{\lambda}{2}\bigr)^2\Bigr)
		+\exp\Bigl(-\tfrac{K}{2}\bigl(\tfrac{\lambda}{2}-4t\bigr)^2\Bigr).
		\]
		This bound is uniform in $\xi_1$ with $\abs{I_{\mathrm{heavy},n}(\xi_1)}\le \tfrac{\lambda}{2}n$ (and trivial otherwise since then $\phi_1(\xi_1)=0$).
		Therefore, taking expectations in \eqref{eq:phi1-outer} gives
		\[
		\Pb\Bigl(\abs{J_n\cap \Theta_0\cap I_{\mathrm{light},n}}\le tK,\ \abs{\Theta_0}\ge \lambda n,\ \abs{I_{\mathrm{heavy},n}}\le \tfrac{\lambda}{2}n\Bigr)
		\le \exp\Bigl(-2K\bigl(\tfrac34-\tfrac{\lambda}{2}\bigr)^2\Bigr)
		+\exp\Bigl(-\tfrac{K}{2}\bigl(\tfrac{\lambda}{2}-4t\bigr)^2\Bigr).
		\]
		Combining this with \eqref{eq:union-new}, we obtain the desired result.
	\end{proof}

\begin{corollary}[Sufficiently many good coordinates]\label{lem:lightspread}
	Let $A$ be an $n \times n$ deterministic matrix such that $\norm{A}= 1$. Let $0<\e,\delta<1$. Let $0<\alpha, \beta, \gamma<1$ and $\rho \ge 1$. Let $V$ be an $n$ dimensional pattern matrix with parameters $\alpha, \beta, \gamma,\rho$. Let $K$ and $L$ with $L>K$ be positive integers. Let $(R_1,\{J_i\}_{i \in [n]},\overline{H},H,\{X_{l,i}\}_{i \in [n],l \in[K]},\overline{R}_2,R_2)$ be a $n$-dimensional oblivious perturbation model with pattern matrix $V$ and parameters $\rho$, $K$, and $L$. Let $M= A+ \varepsilon (R_1+R_2)$. Let $Y_1,...,Y_n$ be the columns of $M$. Let $\zeta$ be a measurable function such that for every $(y_1,\dots,y_{n-1})\in(\mathbb{R}^n)^{n-1}$, we have
		\[
        \zeta(y_1,\ldots,y_{n-1})\in S^{n-1}\cap \operatorname{span}(y_1,\ldots,y_{n-1})^\perp .
        \] Let $Z=\zeta(Y_1,...,Y_{n-1})$. Let $W \in \R^K$ be such that $W_l=Z_{j(n,l)}\,\mathbf{1}_{\{j(n,l)\in I_{\mathrm{light},n}\}}$. {Assume
\[
 n\ge \frac{2}{\gamma},
 \quad
 \alpha <
 \frac{\gamma/2}{
 c_{\ref{prop:comp}.1}
 \log\!\left(
 \frac{c_{\ref{prop:comp}.1}\rho}{\beta(\gamma/2)^{3/2}}
 \right)},
 \quad
 \nu = c_{\ref{prop:comp}.2}\,\varepsilon\frac{\beta\sqrt{\gamma/2}}{\rho}.
\]}
If 
\[
L>eK, \quad  K \ge \log \Bigl(\frac{8}{\alpha}\Bigr)+8\log \Bigl(\frac{16}{\alpha \delta}\Bigr)+\frac{128}{\alpha^2} \log \Bigl(\frac{16}{\alpha \delta}\Bigr),
\quad
 n \ge \frac{1}{c_{\ref{lem:randnormincomp}.1}\gamma} \log \Bigl(\frac{16}{\alpha \delta}\Bigr),
 \quad
 \frac{n}{K^2} \ge \frac{128}{\alpha^2} \log \Bigl(\frac{16}{\alpha \delta}\Bigr),
\]
then we have
    \begin{align*}
        {\mathbb P\left(\card\left\{l \in [K]:\frac{\nu}{\sqrt{2n}}\leqslant |W_l|\leqslant \sqrt{\frac{2}{\alpha n}}\right\} \ge \frac{\alpha K}{32}\right) \ge 1-\frac{1}{4}\alpha \delta.}
    \end{align*}
\end{corollary}

\begin{proof}
	Let \begin{align*}
		\Theta_0=\{l \in [n]:\frac{\nu}{\sqrt{2n}}\leqslant |Z_l|\leqslant \sqrt{\frac{2}{\alpha n}}\}
	\end{align*}
	
	Combining Lemma \ref{lem:sharpened-spread} and Lemma \ref{lem:randnormincomp}, we have
	\begin{align*}
		\mathbb{P}(|\Theta_0| < \frac{1}{2}\alpha n) \le \exp \paren*{-c_{\ref{lem:randnormincomp}.1}\gamma n}
	\end{align*}
	
	By Lemma \ref{thm:Iheavy}, we have
	\begin{align*}
		\mathbb{P}\Bigl\{|I_{\mathrm{heavy},n}|>\tfrac{\alpha}{4}n\Bigr\}
		\le
		\exp\Biggl(-\frac{\bigl(\tfrac{\alpha}{4}n-n e^{-K}(eK/L)^L\bigr)_+^2}{2nK^2}\Biggr)
	\end{align*}
	We observe that $(eK/L)^L<1$ if we require $L>eK$. In this case, if we also require $K> -\log (\frac{\alpha}{8})$, then we have
	\begin{align*}
		\mathbb{P}\Bigl\{|I_{\mathrm{heavy},n}|>\tfrac{\alpha}{4}n\Bigr\}
		\le
		\exp\Biggl(-\frac{\bigl(\tfrac{\alpha}{8}n\bigr)^2}{2nK^2}\Biggr)=\exp\Biggl(-\frac{{\alpha^2}}{128K^2}n\Biggr)
	\end{align*}
	
	By Lemma \ref{lem:spreadlight}, for any
	$t\in\bigl(0,\lambda/8\bigr)$, we have
	\begin{align}
		\Pb\Bigl\{\abs{J_n\cap \Theta_0\cap I_{\mathrm{light},n}}\le tK\Bigr\}
		\ &\le\
		\Pb\{\abs{\Theta_0}<\lambda n\}
		+\Pb\Bigl\{\abs{I_{\mathrm{heavy},n}}>\tfrac{\lambda}{2}n\Bigr\}\\
		&\quad
		+\exp\Bigl(-2K\bigl(\tfrac34-\tfrac{\lambda}{2}\bigr)^2\Bigr)
		+\exp\Bigl(-\tfrac{K}{2}\bigl(\tfrac{\lambda}{2}-4t\bigr)^2\Bigr).
	\end{align}
	Choosing $\lambda=\frac{1}{2}\alpha<1$ and $t=\frac{\lambda}{16}=\frac{1}{32}\alpha$, we have
	\begin{align}
		\Pb\Bigl\{\abs{J_n\cap \Theta_0\cap I_{\mathrm{light},n}}\le \frac{1}{32}\alpha K\Bigr\}
		\ &\le\
		\Pb\{\abs{\Theta_0}<\frac{1}{2}\alpha n\}
		+\Pb\Bigl\{\abs{I_{\mathrm{heavy},n}}>\frac{1}{4}\alpha n\Bigr\}\\
		&\quad
		+\exp\Bigl(-\frac {1}{8}K\Bigr)
		+\exp\Bigl(-\frac{K}{2}\bigl(\frac{1}{8}\alpha\bigr)^2\Bigr)
		\\& \le \exp \paren*{-c_{\ref{lem:randnormincomp}.1} n/2}
		+\exp\Biggl(-\frac{{\alpha^2}}{128K^2}n\Biggr)\\
		&\quad
		+\exp\Bigl(-\frac {1}{8}K\Bigr)
		+\exp\Bigl(-\frac{K}{2}\bigl(\frac{1}{8}\alpha\bigr)^2\Bigr)
	\end{align}
	
	Therefore, if we require $n \ge \frac{1}{c_{\ref{lem:randnormincomp}.1}\gamma} \log (16/(\alpha \delta))$, $\frac{n}{K^2} \ge \frac{128}{\alpha^2} \log(16/(\alpha \delta))$, $K \ge \log (\frac{8}{\alpha})+8\log (16/(\alpha \delta))+\frac{128}{\alpha^2} \log (16/(\alpha \delta))$, $L>eK$, then each of the four terms above is at most $\alpha\delta/16$. So we conclude
	\begin{align*}
		\Pb\Bigl\{\abs{J_n\cap \Theta_0\cap I_{\mathrm{light},n}}\le {\frac{1}{32}\alpha K}\Bigr\}
		\le {\frac{1}{4}\alpha \delta.
	}\end{align*}
	{Since
	}\[
	\card\left\{{l \in }[{K}]{:\frac{\nu}{\sqrt{2n}}}\leqslant {|W_l|}\leqslant {\sqrt{\frac{2}{\alpha n}}}\right\}
	{=
	}\abs{J_n\cap \Theta_0\cap I_{\mathrm{light},n}}{,
	}\]
	{we obtain
	} \begin{align*}
		 \Pb \left( \card  \left\{ l \in [K]:\frac{\nu}{\sqrt{2n}}\leqslant |W_l|\leqslant \sqrt{\frac{2}{\alpha n}}  \right\}  \ge \frac{\alpha K}{32} \right)
		  \ge {1-\frac{1}{4}\alpha \delta.
	} \end{align*}
\end{proof}

Next, we prove a small ball probability result that will be used after we condition on $(\Xi_1,\Xi_2)$.

\begin{lemma}[Small ball probability]
		\label{lem:levyforw}
		Fix integers $n,K\ge 1$.
	Let $w=(w_1,\dots,w_K)\in\mathbb R^K$ be deterministic.
	Fix parameters $\alpha>0$ and $\nu>0$ and define the index set
	\begin{equation}
		\label{eq:Theta}
		\Theta
		=\Bigl\{\,l\in[K]: \frac{\nu}{\sqrt{2n}}\le |w_l|\le \sqrt{\frac{2}{\alpha n}}\,\Bigr\}.
	\end{equation}
	Assume that $\Theta\neq\varnothing$.
	Let $(X_{l,n})_{l\in\Theta}$ be independent random signs (Rademacher variables), i.e.
	\begin{equation}
		\label{eq:rad}
		\mathbb P\{X_{l,n}=1\}=\mathbb P\{X_{l,n}=-1\}=\frac12\quad(l\in\Theta).
	\end{equation}
	Let $\e>0$ and $L>0$. 
		Then there exists a constant $c_{\ref{lem:levyforw}.1}>0$ such that for every
		$t\ge \e\nu/L$,
		\[
			\mathcal L\!\left(\frac{\e}{L}\sum_{l\in\Theta} w_l X_{l,n},\;\frac{t}{\sqrt n}\right)
			\le
			\frac{c_{\ref{lem:levyforw}.1}Lt}{\e\nu\sqrt{|\Theta|}}.
		\]
	\end{lemma}
	
	\begin{proof}
    Let
    \begin{align*}
        T=\sum_{l\in\Theta} w_l\,X_{l,n},
        \quad
        r=\frac{Lt}{\e\sqrt n}.
    \end{align*}
		By Lemma~\ref{lem:scale} we have
		\begin{equation}
			\label{eq:reduce}
			\mathcal L\!\left(\frac{\e}{L}\sum_{l\in\Theta} w_l X_{l,n},\frac{t}{\sqrt n}\right)
			=\mathcal L\!\left(\frac{\e T}{L},\frac{t}{\sqrt n}\right)
			=\mathcal L(T,r).
		\end{equation}
		
		Let $m=\card(\Theta)$ and enumerate $\Theta=\{l_1,\dots,l_m\}$. Also, let
		\[
		\zeta_i=w_{l_i}X_{l_i,n},\quad i=1,\dots,m.
		\]
		Then $T=\sum_{i=1}^m \zeta_i$.
		Since the $X_{l,n}$ are independent, the $\zeta_i$ are independent as well.
		\[
		s:=\min\!\left\{r,\frac{\nu}{2\sqrt{2n}}\right\}.
		\]
		For every \(i\in[m]\), the definition of \(\Theta\) gives
		\[
		|w_{l_i}|\ge \frac{\nu}{\sqrt{2n}},
		\]
		so \(s\le |w_{l_i}|/2\). Since \(\zeta_i\) takes the two values \(\pm w_{l_i}\) with probability \(1/2\), any interval of radius \(s\) can contain at most one atom of \(\zeta_i\). Therefore
		\[
		\mathcal L(\zeta_i,s)=\frac12,
		\quad i\in[m].
		\]
		Applying Lemma~\ref{lem:rogozin} with \(r_i=s\) for all \(i\),   {there is an absolute constant $c_1>0$ such that
		} \begin{align*}
		\mathcal L(T,r)
		&\le
		  {\frac{c_1\,r}{\sqrt{\sum_{i=1}^m \bigl(1-\mathcal L(\zeta_i,s)\bigr)s^2}}
		} \\&=
		  {\frac{c_1\,r}{\sqrt{m\,s^2/2}}
		} =
		  {\frac{\sqrt2\,c_1}{\sqrt m}} \,\frac{r}{s}.
		\end{align*}
		By the definition of \(s\), we have
		\[
		\frac{r}{s}
		=
		\max\!\left\{1,\frac{2\sqrt{2n}}{\nu}r\right\}
		=
		\max\!\left\{1,\frac{2\sqrt2\,Lt}{\e\nu}\right\}.
		\]
		Since \(t\ge \e\nu/L\), we have
		\[
		\max\!\left\{1,\frac{2\sqrt2\,Lt}{\e\nu}\right\}
		\le
		\frac{(1+2\sqrt2)Lt}{\e\nu}.
		\]
		Combining this with \eqref{eq:reduce}, we conclude that
		\[
		\mathcal L\!\left(\frac{\e}{L}\sum_{l\in\Theta} w_l X_{l,n},\frac{t}{\sqrt n}\right)
		\le
		  {\frac{\sqrt2(1+2\sqrt2)c_1Lt}{\e\nu\sqrt m}} .
		\]

	\end{proof}
With all the above preparation, we can now implement the conditioning argument to bound $\inf_{x \in \incomp(\alpha, \nu)} \norm{Mx}_2$.
\begin{theorem} [Invertibility for incompressible vectors]\label{thm:incomp}
    Let $A$ be an $n \times n$ deterministic matrix such that $\norm{A}= 1$. Let $0<\e<1$ and $0<\delta<1$. Let $0<\alpha, \beta, \gamma <1$ and $\rho \ge 1$. Let $V$ be an $n$ dimensional pattern matrix with parameters $\alpha, \beta, \gamma,\rho$. Let $K$ and $L$ with $L>K$ be positive integers. Let $(R_1,\{J_i\}_{i \in [n]},\overline{H},H,\{X_{l,i}\}_{i \in [n],l \in[K]},\overline{R}_2,R_2)$ be a $n$-dimensional oblivious perturbation model with pattern matrix $V$ and parameters $\rho$, $K$, and $L$. Let $M= A+ \varepsilon (R_1+R_2)$. Then there exist constants $c_{\ref{thm:incomp}.1}, c_{\ref{thm:incomp}.2},c_{\ref{thm:incomp}.3}$ such that for any \[
 \alpha <
 \frac{\gamma/2}{
 c_{\ref{prop:comp}.1}
 \log\!\left(
 \frac{c_{\ref{prop:comp}.1}\rho}{\beta(\gamma/2)^{3/2}}
 \right)},
 \quad
 \nu = c_{\ref{prop:comp}.2}\,\varepsilon\frac{\beta\sqrt{\gamma/2}}{\rho},
\quad
K=\left\lceil \frac{c_{\ref{thm:incomp}.2}}{\delta^2\alpha^3}\right\rceil,
\quad
L=\lceil 2eK\rceil,
\]
{and}
\[
 n\ge
 \max\!\left\{
 \frac{c_{\ref{thm:incomp}.1}}{\delta^5\alpha^{9}},
 \frac{1}{c_{\ref{lem:randnormincomp}.1}\gamma}\log\!\left(\frac{16}{\alpha\delta}\right),
 \frac{2}{\gamma}
 \right\},
\]
{we have}
    \[
    \Pb \left( \inf_{v \in \incomp(\alpha, \nu)} \norm{Mv}_2
    \ge c_{\ref{thm:incomp}.3}\,\frac{\beta^2\gamma}{\rho^2}\,\e^3 \delta^2\alpha^{3}n^{-1} \right) \ge 1-\delta/2 .
    \]
\end{theorem}

\begin{proof}
As in the statement of Corollary \ref{lem:lightspread}, let $Y_1,...,Y_n$ be the columns of $M$. Let $\zeta$ be a measurable function such that for every $(y_1,\dots,y_{n-1})\in(\mathbb{R}^n)^{n-1}$, we have
		\[
		  {\zeta} (y_1,  {\ldots} ,y_{n-1})  {\in S^{n-1}\cap }\operatorname{span} (y_1,  {\ldots} ,y_{n-1})  {^\perp.
		} \]
Let $Z=\zeta(Y_1,...,Y_{n-1})$.  By assumption we have $0<\nu<1$. Let $W \in \R^K$ be such that $W_l=Z_{j(n,l)}\,\mathbf{1}_{\{j(n,l)\in I_{\mathrm{light},n}\}}$.
Consider the event
\begin{align*}
    \card(\{l \in [K]:\frac{\nu}{\sqrt{2n}}\leqslant |W_l|\leqslant \sqrt{\frac{2}{\alpha n}}\}) \ge \frac{\alpha}{32}K
\end{align*} 

Since this event is completely determined by  $(\Xi_1,\Xi_2)$, we can represent this event as
\begin{align*}
    \{\omega: (\Xi_1(\omega),\Xi_2 (\omega)) \in \mathscr{E}\}
\end{align*}
for some measurable set $\mathscr{E}$ in the corresponding measurable space where $(\Xi_1(\omega),\Xi_2 (\omega))$ takes values.

By Corollary \ref{lem:lightspread}, if   {$n \ge \frac{1}{c_{\ref{lem:randnormincomp}.1}\gamma} \log (16/(\alpha \delta))$, $\frac{n}{K^2} \ge \frac{128}{\alpha^2} \log (16/(\alpha \delta))$} , $K \ge \log (\frac{8}{\alpha})+8\log (16/(\alpha \delta))+\frac{128}{\alpha^2} \log (16/(\alpha \delta))$, $L>eK$,  {and $n\ge 2/\gamma$, } then we have
\begin{align*}
    &\Pb(\card(\{l \in [K]:\frac{\nu}{\sqrt{2n}}\leqslant |W_l|\leqslant \sqrt{\frac{2}{\alpha n}}\}) \ge \frac{\alpha}{32}K)
    \\ =& \Pb(\{\omega: (\Xi_1(\omega),\Xi_2 (\omega)) \in \mathscr{E}\})
    \\ \ge& 1-\frac{1}{4}\delta \alpha
\end{align*}
We will verify these requirements below. Assuming that they hold for the moment, we have
\begin{align*}
    &\Pb(|\langle Y_n, Z \rangle|< \frac{t}{\sqrt{n}})
    \\ \le &\E \mathds{1}_{\{|\langle Y_n, Z \rangle|< \frac{t}{\sqrt{n}}\}}(\omega) \mathds{1}_{\{\omega: (\Xi_1(\omega),\Xi_2 (\omega)) \in \mathscr{E} \}}(\omega) +\Pb((\Xi_1(\omega),\Xi_2 (\omega)) \in \mathscr{E}^C)
    \\=&\E_{\Xi_1,\Xi_2} \E_{\Xi_3}\mathds{1}_{\{|\langle Y_n, Z \rangle|< \frac{t}{\sqrt{n}}\}}(\omega)\mathds{1}_{\{\omega: \Xi_1(\omega),\Xi_2 (\omega)) \in \mathscr{E}\}}(\omega) + \frac{1}{4}\delta \alpha
    \\=&\E_{\Xi_1,\Xi_2}[\mathds{1}_{\mathscr{E}}(\Xi_1,\Xi_2) \E_{\Xi_3}\mathds{1}_{\{|\langle Y_n, Z \rangle|< \frac{t}{\sqrt{n}}\}}(\omega)] + \frac{1}{4}\delta \alpha
    \\=&\E_{\Xi_1,\Xi_2}[\mathds{1}_{\mathscr{E}}(\Xi_1,\Xi_2) \Pb_{\Xi_3}(|\big\langle \mathrm{Col}_n(A+\varepsilon R_1),\, Z\big\rangle
			+\frac{\e}{L}\sum_{l=1}^K X_{l,n}\, Z_{j(n,l)}\,\mathbf{1}_{\{j(n,l)\in I_{\mathrm{light},n}\}}|< \frac{t}{\sqrt{n}})]+\frac{1}{4}\delta \alpha
\end{align*}

For fixed $(\xi_1,\xi_2) \in \mathscr{E}$,   {define
}\[
{\Theta=}\left\{{l\in[K]:\frac{\nu}{\sqrt{2n}}}\leqslant {|Z_{j(n,l)}\mathbf{1}_{\{j(n,l)\in I_{\mathrm{light},n}\}}|}\leqslant {\sqrt{\frac{2}{\alpha n}}}\right\}{.
}\]
{Then $|\Theta|\ge \alpha K/32$. After conditioning further on the signs $(X_{l,n})_{l\notin\Theta}$, the contribution of the indices outside $\Theta$ is a deterministic shift and the signs $(X_{l,n})_{l\in\Theta}$ remain independent. Therefore, for fixed $(\xi_1,\xi_2) \in \mathscr{E}$, we have
} \begin{align*}
    &\Pb_{\Xi_3}  \left( |\big\langle \mathrm{Col}_n(A+\varepsilon R_1),\, Z\big\rangle
			+\frac{\e}{L}\sum_{l=1}^K X_{l,n}\, Z_{j(n,l)}\,\mathbf{1}_{\{j(n,l)\in I_{\mathrm{light},n}\}}|< \frac{t}{\sqrt{n}}  \right)  \\
      \le &{\mathcal{L}} _{  {(X_{l,n})_{l\in\Theta}} }  \left({\frac{\e}{L}\sum} _{l  {\in\Theta} }  X_{l,n}\, Z_{j(n,l)}\,\mathbf{1}_{\{j(n,l)\in I_{\mathrm{light},n}\}},\frac{t}{\sqrt{n}}  \right)  \\
      \le&
			{\frac{c_{\ref{lem:levyforw}.1}Lt}{\e\nu\sqrt{\frac{\alpha}{32}K}},
} \end{align*}
where the last line   follows  from Lemma \ref{lem:levyforw} for $t\ge \e\nu/L$.

Therefore, by iterated expectation, we have
\begin{align*}
    &\E  {_{\Xi_1,\Xi_2}[\mathds{1}_{\mathscr{E}}(\Xi_1,\Xi_2) \Pb_{\Xi_3}(|\big\langle \mathrm{Col}_n(A+\varepsilon R_1),\, Z\big\rangle
			+\frac{\e}{L}\sum_{l=1}^K X_{l,n}\, Z_{j(n,l)}\,\mathbf{1}_{\{j(n,l)\in I_{\mathrm{light},n}\}}|< \frac{t}{\sqrt{n}})]
            } \\ \le & \E_{\Xi_1,\Xi_2}\left[\mathds{1}_{\mathscr{E}}(\Xi_1,\Xi_2)   {\frac{c_{\ref{lem:levyforw}.1}Lt}{\e\nu\sqrt{\frac{\alpha}{32}K}}} \right]
            \\ \le &   {\frac{c_{\ref{lem:levyforw}.1}Lt}{\e\nu\sqrt{\frac{\alpha}{32}K}}.
} \end{align*}
and therefore we have
\begin{align*}
    \Pb  \left( |\langle Y_n, Z \rangle|< \frac{t}{\sqrt{n}}  \right)  \le   {\frac{c_{\ref{lem:levyforw}.1}Lt}{\e\nu\sqrt{\frac{\alpha}{32}K}}} +\frac{1}{4}\delta \alpha
\end{align*}
for $t\ge \e\nu/L$.

If we also require
\begin{align*}
      {\frac{c_{\ref{lem:levyforw}.1}Lt}{\e\nu\sqrt{\frac{\alpha}{32}K}} } \le \frac{1}{4}\delta \alpha ,
\end{align*}
then the same argument, with column $k$ in place of column $n$, gives the same bound for every $k\in[n]$. Applying Lemma \ref{invertdist}, we have
\begin{align*}
\mathbb{P}\left(\inf_{x\in \mathrm{Incomp}(\alpha,\nu)} \|Mx\|_2 < t \nu\, n^{-1}\right)
\le &\frac{1}{\alpha n}\sum_{k=1}^{n}\mathbb{P}\left(\operatorname{dist}(Y_k,H_k)<t n^{-1/2}\right)
\\ \le & \frac{1}{\alpha n} n   \left( \frac{1}{4}\delta \alpha + \frac{1}{4}\delta \alpha  \right)
 \\ = & \frac{1}{2}\delta  {.
} \end{align*}

In summary, the requirements are
\begin{enumerate}
\item \label{ineq:L-gt-eK}
  $L>eK$.
\item \label{ineq:K-lower}
  $K \ge \log \Bigl(\frac{8}{\alpha}\Bigr)+8\log \Bigl(\frac{16}{\alpha \delta}\Bigr)+\frac{128}{\alpha^2} \log \Bigl(\frac{16}{\alpha \delta}\Bigr)$.

\item \label{ineq:n-lower}
    {$n \ge \frac{1}{c_{\ref{lem:randnormincomp}.1}\gamma} \log \Bigl(\frac{16}{\alpha \delta}\Bigr)$ and $n\ge 2/\gamma$} .

\item \label{ineq:n-over-K2}
    {$\frac{n}{K^2} \ge \frac{128}{\alpha^2} \log \Bigl(\frac{16}{\alpha \delta}\Bigr)$} .

\item \label{ineq:prob<delta}
  {$ \frac{c_{\ref{lem:levyforw}.1}Lt}{\e\nu\sqrt{\frac{\alpha}{32}K}} \le \frac{1}{4}\delta \alpha.$
} 

\item \label{ineq:tlowerbd} $t\ge \e\nu/L$.
\end{enumerate}

To satisfy all the requirements, we choose
\[
K=\left\lceil   {\frac{c_1}{\delta^2\alpha^3}} \right\rceil,
\quad
L=\lceil 2eK\rceil, \quad
t= \frac{\e\nu}{L},
\]
where   {$c_1>0$ } is a sufficiently large  {absolute } constant, to be chosen below. Note that the requirement \eqref{ineq:tlowerbd} is already satisfied by the choice of $t$.

We claim that with this choice, if
\[
n\ge   {\frac{1}{c_{\ref{lem:randnormincomp}.1}\gamma} } \log \Bigl(\frac{16}{\alpha \delta}\Bigr) {,
}\quad
{n}\ge {\frac{2}{\gamma},
} \quad\text{and}\quad
n\ge   {\frac{8192\,c_1^2}{\delta^5\alpha^{9}}} ,
\]
then all the requirements are satisfied. Since   {$L=\lceil 2eK\rceil$} , we have   {$L>eK$} . Therefore \eqref{ineq:L-gt-eK} holds.

We next verify the requirement \eqref{ineq:K-lower}. Since   {$0<\alpha\le 1$ and $0<\delta<1$}, we have
\[
\log\!\Bigl(\frac{16}{\alpha\delta}\Bigr)\le \frac{16}{\alpha\delta}\le \frac{16}{\alpha^3\delta^2}.
\]
which implies
\[
8\log\!\Bigl(\frac{16}{\alpha\delta}\Bigr)\le \frac{128}{\alpha^3\delta^2}
\]
and
\[
\frac{128}{\alpha^2}\log\!\Bigl(\frac{16}{\alpha\delta}\Bigr)
\le
\frac{128}{\alpha^2}\cdot \frac{16}{\alpha\delta}
\le
\frac{2048}{\alpha^3\delta^2}.
\]

Also, since   {$0<\alpha, \delta \le 1$}, we have
\[
\log\!\Bigl(\frac{8}{\alpha}\Bigr)\le \frac{8}{\alpha}\le \frac{64}{\alpha^3 \delta^2}.
\]

Combining these bounds, we obtain
\begin{align*}
\log\!\Bigl(\frac{8}{\alpha}\Bigr)
+8\log\!\Bigl(\frac{16}{\alpha\delta}\Bigr)
+\frac{128}{\alpha^2}\log\!\Bigl(\frac{16}{\alpha\delta}\Bigr)
&\le
\frac{64}{\delta^2\alpha^3}
+\frac{128}{\delta^2\alpha^3}
+\frac{2048}{\delta^2\alpha^3}
\\
&=
\frac{2240}{\delta^2\alpha^3}.
\end{align*}
If   {$c_1\ge 2240$}, then we have
\[
K\ge
\log\!\Bigl(\frac{8}{\alpha}\Bigr)
+8\log\!\Bigl(\frac{16}{\alpha\delta}\Bigr)
+\frac{128}{\alpha^2}\log\!\Bigl(\frac{16}{\alpha\delta}\Bigr)
\]
and therefore requirement \eqref{ineq:K-lower} is satisfied.

Next if {$c_1 \ge 512 c_{\ref{lem:levyforw}.1}^2$}, then we have
 \begin{align*}
	  {\frac{c_{\ref{lem:levyforw}.1}Lt}{\e\nu\sqrt{\frac{\alpha}{32}K}} } &=   {\frac{c_{\ref{lem:levyforw}.1}}{\sqrt{\frac{\alpha}{32}K}} } \\
	& \le   {\frac{c_{\ref{lem:levyforw}.1}\sqrt{32}}{\sqrt{c_1}} } \alpha \delta \\
	&\le \frac{\alpha \delta }{4}
\end{align*}
so the requirement \eqref{ineq:prob<delta} is satisfied.

To guarantee requirement \eqref{ineq:n-over-K2},   {it is enough to impose the condition
} \[
\frac{n}{K^2}\ge \frac{128}{\alpha^2}\cdot \frac{16}{\alpha\delta}
=
\frac{2048}{\alpha^3\delta}  {,
} \]
 {because $\log(16/(\alpha\delta))\le 16/(\alpha\delta)$. } Since
\[
K=\left\lceil   {\frac{c_1}{\delta^2\alpha^3}} \right\rceil,
\]
and
\[
  {\frac{c_1}{\delta^2\alpha^3}} \ge 1,
\]
we have
\[
K\le   {\frac{2c_1}{\delta^2\alpha^3}} .
\]
Therefore it is sufficient to require
\[
n\ge \frac{2048}{\alpha^3\delta}
\left(  {\frac{2c_1}{\delta^2\alpha^3}} \right)^2
=
  {\frac{8192\,c_1^2}{\delta^5\alpha^{9}}} .
\]
 {This same bound also guarantees $K\le n$, so the random $K$-subsets of $[n]$ are well-defined.
} 

Therefore, if
\[
n\ge \max   \left\{{\frac{1}{c_{\ref{lem:randnormincomp}.1}\gamma}} \log\!\Bigl(\frac{16}{\alpha\delta}\Bigr),  {\frac{2}{\gamma},\frac{8192\,c_1^2}{\delta^5\alpha^{9}}}\right\},
\]
and if we choose
\[
K=\left\lceil   {\frac{c_1}{\delta^2 \alpha^3}} \right\rceil,
\quad
L=\lceil 2eK\rceil,
\quad
t=\frac{\e\nu}{L},
\]
with   {$c_1\ge \max\{2240,512\,c_{\ref{lem:levyforw}.1}^2,1\}$} , then all the required constraints are satisfied.

Next, we eliminate   {$L$} and obtain a lower bound for our choice of   {$t$ } in terms of   {$\e,\delta,\nu,\alpha$}. Since   {$0<\delta, \alpha<1$}, we have
\[
  {\frac{c_1}{\delta^2\alpha^3}} \ge 1,
\]
and therefore
\[
L
= \lceil 2eK \rceil \le   {\frac{8ec_1}{\delta^2\alpha^3}}
\]
which implies
\[
\frac{1}{L}
\ge
  {\frac{\delta^2 \alpha^3}{8ec_1}}.
\]
Substituting into the formula for   {$t$}, we obtain
\begin{align*}
t
=
\frac{\e\nu}{L}
\ge
  {\frac{\e\nu\delta^2\alpha^3}{8ec_1}}
\end{align*}
which implies
\[
t\ge c  {_2} \e \delta^2\nu\alpha^{3}
\]
where $c  {_2} =  {\frac{1}{8ec_1}}$.

In summary, we have
  \begin{align*}
{t\nu
}& \ge c  {_2 } \e \delta^2\nu^2\alpha^{3}  \\
& =  {\frac{c_{\ref{prop:comp}.2}^2}{16ec_1}\frac{\beta^2\gamma}{\rho^2}} \e^3\delta^2\alpha  {^3.
}\end{align*}
\end{proof}

\section{Invertibility on the whole sphere and condition number.} 
\label{sec:mainthm}

Combining the results on invertibility for compressible and incompressible vectors, we can obtain the following conclusion about invertibility on the whole sphere.
\begin{theorem}[Bound for smallest singular value] \label{thm:invsphere}
    Let $A$ be an $n \times n$ deterministic matrix such that $\norm{A}= 1$. Let $0<\e,\delta<1$. Let $0<\alpha, \beta, \gamma <1$ and $\rho \ge 1$. Let $V$ be an $n$ dimensional pattern matrix with parameters $\alpha, \beta, \gamma,\rho$. Let $K$ and $L$ with $L>K$ be positive integers. Let $(R_1,\{J_i\}_{i \in [n]},\overline{H},H,\{X_{l,i}\}_{i \in [n],l \in[K]},\overline{R}_2,R_2)$ be a $n$-dimensional oblivious perturbation model with pattern matrix $V$ and parameters $\rho$, $K$, and $L$. Let $M= A+ \varepsilon (R_1+R_2)$. Then there exist constants $c_{\ref{thm:invsphere}.1}, c_{\ref{thm:invsphere}.2}, c_{\ref{thm:invsphere}.3}, c_{\ref{thm:invsphere}.4}$ such that for any \[
 \alpha <
 \frac{\gamma/2}{
 c_{\ref{thm:invsphere}.1}
 \log\!\left(
 \frac{c_{\ref{thm:invsphere}.1}\rho}{\beta(\gamma/2)^{3/2}}
 \right)},
 \quad
K=\left\lceil \frac{c_{\ref{thm:invsphere}.2}}{\delta^2\alpha^3}\right\rceil,
\quad
L=\lceil 2eK\rceil,
\]
{and}
\[
 n\ge
 c_{\ref{thm:invsphere}.3}\max\!\left\{
 \frac{1}{\delta^5\alpha^{9}},
 \frac{1}{\gamma}\log\!\left(\frac{16}{\alpha\delta}\right),
 \frac{2}{\gamma}
 \right\},
\]we have
    \[ \Pb \paren*{ s_n(M)  \ge c_{\ref{thm:invsphere}.4} \frac{\beta^2\gamma}{\rho^2} \e^3 \delta^2\alpha^{3}n^{-1} } \ge 1-\delta \]
\end{theorem}
\begin{proof}
    This result now follows directly from combining Theorem \ref{prop:comp} and Theorem \ref{thm:incomp}.
\end{proof}

This smallest singular value bound then implies the following control on the condition number.

\begin{theorem}[Bound for condition number] \label{thm:condnum}
    Let $A$ be an $n \times n$ deterministic matrix such that $\norm{A}= 1$. Let $0<\e,\delta<1$. Let $0<\alpha, \beta, \gamma <1$ and $\rho \ge 1$. Let $V$ be an $n$ dimensional pattern matrix with parameters $\alpha, \beta, \gamma,\rho$. Let $K$ and $L$ with $L>K$ be positive integers. Let $(R_1,\{J_i\}_{i \in [n]},\overline{H},H,\{X_{l,i}\}_{i \in [n],l \in[K]},\overline{R}_2,R_2)$ be a $n$-dimensional oblivious perturbation model with pattern matrix $V$ and parameters $\rho$, $K$, and $L$. Let $M= A+ \varepsilon (R_1+R_2)$. Then there exist constants $c_{\ref{thm:condnum}.1}, c_{\ref{thm:condnum}.2}, c_{\ref{thm:condnum}.3}, c_{\ref{thm:condnum}.4}$ such that for any \[
 \alpha <
 \frac{\gamma/2}{
 c_{\ref{thm:condnum}.1}
 \log\!\left(
 \frac{c_{\ref{thm:condnum}.1}\rho}{\beta(\gamma/2)^{3/2}}
 \right)},
 \quad
K=\left\lceil \frac{c_{\ref{thm:condnum}.2}}{\delta^2\alpha^3}\right\rceil,
\quad
L=\lceil 2eK\rceil,
\]
{and}
\[
 n\ge
 c_{\ref{thm:condnum}.3}\max\!\left\{
 \frac{1}{\delta^5\alpha^{9}},
 \frac{1}{\gamma}\log\!\left(\frac{16}{\alpha\delta}\right),
 \frac{2}{\gamma}
 \right\},
\] we have
    \[ \Pb \paren*{ \kappa(M)  \le \frac{c_{\ref{thm:condnum}.4}n\rho^2} {\beta^2\gamma \e^3 \delta^2\alpha^{3}} } \ge 1-\delta \]
\end{theorem}

\begin{proof}

    Recall that we wish to show that
$$\P \paren*{\kappa(M) \le C(\varepsilon, \delta)n}\ge 1-\delta.$$
where $M = A + \varepsilon(R_1 + R_2)$ for a deterministic matrix $A$ with $\norm{A} \le 1$.

By Remark \ref{rem:norm}, we have $s_1(M) \le \norm{A} + \varepsilon \norm{R_1} + \varepsilon \norm{R_2} \le 3$, so it is sufficient to give a lower bound for
\[s_n(M)= \inf_{x \in \bS^{n-1}} \norm{Mx}_2 \]
which follows from Theorem \ref{thm:invsphere}.

\end{proof}

\begin{proof}[Proof of Theorem \ref{coro:pertur}]
The claim about the smallest singular value follows directly from Theorem \ref{thm:invsphere}, and we can assume $\norm{R} \le 1$ by working with $\e/2$ instead of $\e$. More precisely, assume that we have successfully generated an $n$ dimensional pattern matrix $V$ with parameters $\alpha, \beta, \gamma,\rho$, and a $n$-dimensional oblivious perturbation model $(R_1,\{J_i\}_{i \in [n]},\overline{H},H,\{X_{l,i}\}_{i \in [n],l \in[K]},\overline{R}_2,R_2)$ with pattern matrix $V$ and parameters $\rho$, $K$, and $L$ as in Theorem \ref{thm:invsphere}. If we set $R_n=\frac{R_1+R_2}{2}$ then $\norm{R_n} \le 1$ by Remark \ref{rem:norm} and then $A+\e R=A+(\e/2)(R_1+R_2)$. Using $\delta/3$ instead of $\delta$, we can conclude that, for any $n \times n$ deterministic matrix $A$ with  $\norm{A}= 1$ and $n\ge c_1\delta^{-5}$, we have
    \[ \Pb \paren*{ s_n(A+ \varepsilon R_n)  \ge c_2 \e^3 \delta^2n^{-1} } \ge 1-\delta/3 \]
for some constants $c_1,c_2>0$. So we only need to calculate the number of random bits needed to generate a random perturbation (i.e., including a pattern matrix and the corresponding oblivious perturbation model) with success probability at least $1-2\delta/3$.

By Corollary \ref{cor:randomnessV}, we can generate $V$ with $O(n \log n)$ random bits with success probability at least $1-1/n^2>1-\delta/3$ if we require $n\ge c_3\delta^{-5}$ with some $c_3$ large enough.

To calculate the number of random bits needed to construct $R_1$, recall that
\begin{equation*}
        R_1=\frac{1}{\rho \sqrt{n}} \cdot D_1 V D_2 
    \end{equation*}
for $D_1 = \diag(\eta_{1,1} \etc \eta_{1,n})$ and $D_2 = \diag(\eta_{2,1} \etc \eta_{2,n})$ where $\eta_{1,1} \etc \eta_{1,n}$ and $\eta_{2,1} \etc \eta_{2,n}$ are i.i.d. Rademacher random variables, and $V$ is a deterministic pattern matrix. The number of random bits to generate $D_1$ and $D_2$ are $O(n)$. Also, by Corollary \ref{cor:randomnessV}, the number of random bits needed to generate $V$ is $O(n\log n)$. So, in summary, the number of random bits needed to generate $R_1$ is $O(n \log n)$.

    To calculate the number of random bits needed to construct $R_2$, let $\mathcal{Z}$ be the set of injections from $[k]$ to $[n]$, and we observe that there is a bijection $\Phi$ from $\mathcal X_{n,k}:=[n]\times[n-1]\times\cdots\times[n-k+1]$ to $\mathcal{Z}$. Let
		$
		U=(U_1,\dots,U_k)
		$
		be a random variable with values in $\mathcal X_{n,k}$ such that the coordinates are independent and satisfy
		\[
		U_1\sim \mathrm{Uniform}([n]), \quad U_2\sim\mathrm{Uniform}([n-1]), \quad \dots, \quad U_k\sim\mathrm{Uniform}([n-k+1]).
		\]
        Therefore, by standard results on generating uniform random variables and a truncation argument (e.g., \cite[Chapter XV, Theorem 3.1]{devroye1986nonuniform}, \cite[Theorem 5.11.3]{cover2006elements}), generating each $U_i$ takes $O(\log(n^2K/\delta))$ random bits in the worst case, with failure probability at most $\delta/(3nK)$.
  
		Define
		\[
		T:=\Phi(U).
		\]
		Let
		\[
		\mathcal S=(s_1,\dots,s_n)
		\]
		be a finite ordered sequence with pairwise distinct entries, and define
		\[
		J=\mathrm{Range}(s\circ T)=\{s_{T(1)},\dots,s_{T(k)}\}.
		\]
		Then by Definition \ref{def:orsamp}, the random sequence $T$ is a random sampling of an ordered $k$-element sequence uniformly from $[n]$ without replacement, and therefore, by Definition \ref{def:unorsamp}, the random set $J$ is a random sampling of an unordered $k$-element subset uniformly from $\mathcal J=\{s_1,\dots,s_n\}$ without replacement.
        As a result, taking $k=K$, we can generate $J$ from $O(K\log(n^2K/\delta))$ random bits. Recall that $K=O(\frac{1}{\delta^2})$. So if we require $n\ge c_4\delta^{-5}$ for large enough $c_4>0$, then the total number of random bits needed to generate $R_2$, including its random signs, is $O(nK\log(n^2K/\delta))=O(\frac{n\log n}{\delta^2})$, with sampling failure is at most $\delta/3$ (by a union bound over the $nK$ draws). 

        Combining all these sampling calculations, we conclude that we can generate a random perturbation (i.e., including a pattern matrix and the corresponding oblivious perturbation model) with success probability at least $1-2\delta/3$ using $(\frac{n\log n}{\delta^2})$ random bits.

\end{proof}
 
 \subsection*{Acknowledgments} MD was supported in part by NSF CAREER Grant CCF-233865 and a Google ML and Systems Junior Faculty Award. The authors used ChatGPT only for literature searches, drafting initial TikZ codes for figures based on their specific instructions and hand-drawn sketches, limited text editing on author written materials, and checking the full manuscript for potential errors. The perturbation model and the proofs were developed by the authors. The authors assume responsibility for all content. The authors thank Rikhav Shah for helpful conversations. 

\appendix

\printbibliography

\end{document}